\documentclass[runningheads]{llncs}

\usepackage[mathlines]{lineno}
\usepackage[T1]{fontenc}
\usepackage{tikz}
\usepackage{stmaryrd}
\usepackage{graphicx}
\usepackage{DotArrow}
\usepackage{chemarrow}
\usepackage{array}
\usepackage{tabularx}

\usepackage{lipsum}
\usepackage{dsfont}
\usepackage{amsmath}
\usepackage{amssymb}
\usepackage{mathtools}
\usepackage{keyval}
\usepackage{environ}

\usepackage{algpseudocode}
\usepackage[linesnumbered,ruled,vlined]{algorithm2e}
\usepackage{algorithmicx}
\usepackage{wrapfig}
\usepackage{multicol}

\usepackage{hyperref}
\hypersetup{hidelinks}          %
\usepackage[capitalise,nameinlink]{cleveref}
\usepackage{thmtools}
\usepackage{thm-restate}        %

\newcommand{\wdwidth}{\ensuremath{{\sf wd}}}

\def\multiset#1{\ensuremath{\left\{\kern-.3em\left\{{#1}\right\}\kern-.3em\right\}}}

\usepackage{tabularx}

\newcommand{\problemframe}[2]{%
	\par\smallskip\noindent
	\begingroup
	\setlength{\tabcolsep}{0pt}%
	\fbox{%
		\begin{minipage}{\dimexpr\linewidth-2\fboxsep-2\fboxrule\relax}
			\emergencystretch=1em\relax
			\begin{tabularx}{\linewidth}{@{}l@{\hspace{0.6em}}X@{}}
				Input:    & #1\\[2pt]
				Question: & #2
			\end{tabularx}
	\end{minipage}}%
	\endgroup
	\par\smallskip}

\newcommand{\xdomain} {\domain_\var}

\newcommand{\infun}{\textsf{in}}
\newcommand{\outfun}{\textsf{out}}
\newcommand{\marking}{\texttt{M}}
\newcommand{\places}{\texttt{P}}
\newcommand{\transition}{\texttt{T}}
\newcommand{\var}{\texttt{X}}
\newcommand{\newvar}{\texttt{Y}}

\DeclareSymbolFont{Shuffle}{U}{shuffle}{m}{n}
\DeclareFontFamily{U}{shuffle}{}
\DeclareFontShape{U}{shuffle}{m}{n}{%
	<-8>shuffle7%
	<8->shuffle10%
}{}
\DeclareMathSymbol\shuffle{\mathbin}{Shuffle}{"001}
\DeclareMathSymbol\cshuffle{\mathbin}{Shuffle}{"002}

\newcommand{\word}{w}

\newcommand{\relconst}{\hat{\mathcal{C}}}
\newcommand{\bigO} {\mathcal{O}}

\newcommand{\instr} {\iota}

\newcommand{\lregmap}{\vec{l}}

\newcommand{\reglang}{\mathtt{Reg}}

\newcommand{\run} {\pi}
\newcommand{\irun} {\rho}
\newcommand{\goto} {\textsf{goto}}

\newcommand{\vecregmapinit}{\vec{l}_{\tt init}}

\newcommand{\Labels}{\Lambda}
\newcommand{\reachedLab}{\Labels}
\newcommand{\lab}{\ell}
\newcommand{\initlab}{\lambda_{\tt init}}
\newcommand{\xvar}{ \mathtt{x}}
\newcommand{\yvar}{ \mathtt{y}}
\newcommand{\zvar}{ \mathtt{z}}
\newcommand{\rvar}{ \mathtt{r}}

\newcommand{\nextof}[1]{ \textsf{next}(#1)}

\newcommand{\tuple}[1] {\left< #1 \right>}
\newcommand{\PDS} {\text{Para-Prog}\xspace}
\newcommand{\PCS}{\text{PCS}\xspace}
\newcommand{\pcs}{\mathcal{C}}

\newcommand{\Naturals} {\mathbb{N}}

\newcommand{\sset} {\mathcal{S}}

\newcommand{\tset} {\mathcal{T}}
\newcommand{\aset} {\mathcal{A}}

\newcommand{\vset}{\mathcal{V}}

\newcommand {\const} {d}
\newcommand {\constants} {\mathbb{C}}

\newcommand{\PCRFR} {\text{PFReach}\xspace}

\newcommand{\assign}[2] {#1 := #2}

\newcommand{\States} {Q}

\newcommand{\initstate} {q_{\tt in}}
\newcommand{\Alphabet} {\Sigma}
\newcommand{\Trans} {\Delta}

\newcommand{\domain}{\mathbb{D}}

\newcommand{\op}{\textsf{op}}

\newcommand{\val} {d}

\newcommand{\Conf}{\textsf{Conf}}
\newcommand{\conf}{{c}}

\newcommand{\Confof}[1] {\Conf_{#1}}
\newcommand{\InitConfof}[1]{\textsf{InitConf}_{#1}}

\newcommand{\initgregmap}{\vec{g}_{\tt in}}
\newcommand{\movesto}[1]{\xrightarrow{#1}}
\newcommand{\Movesto}[1]{\overset{#1}{\Longrightarrow}}
\newcommand{\updatefun}[3]{#1\left[ #2 \leftarrow #3\right]}
\newcommand{\fin} {f}

\newcommand{\pds} {\mathcal{P}}
\newcommand{\globalregs} {\mathcal{G}}
\newcommand{\localregisters} {\mathcal{L}}

\newcommand{\gregmap}{\vec{g}}
\newcommand{\labelmap} {\vec{{\bf \lambda}}}

\newcommand{\pdsconf}{\textsf{pc}}
\newcommand{\iconf}{\textsf{ic}}
\newcommand{\labelmapinit} {\vec{{\bf \lambda}}_{\tt init}}

\newcommand{\initpdsconf}{\pdsconf_{\tt init}}

\newcommand{\lreg}{x}
\newcommand{\greg} {r}

\newcommand{\project}[2]{#1\downarrow_{#2}}

\newcommand{\assert} {{\sf assert}}

\newcommand{\hlab} {\ell^{h}}

\newcommand{\ins} {\textsf{ins}}
\newcommand{\prop} {\textsf{prop}}

\newcommand{\del} {\textsf{del}}

\newcommand{\delprop} {\textsf{dprop}}

\newcommand{\Parts}{\textsf{Parts}}
\newcommand{\Partition}{\textsf{Parts}}
\newcommand{\pt}{\textsf{pt}}

\newcommand{\id} {\textsf{id}\xspace}

\newcommand {\ffun} {f}

\newcommand{\supp} {\textsf{supp}}
\newcommand{\asupp} {\textsf{valsupp}}
\newcommand{\msupp}{\textsf{msupp}}

\newcommand{\hfun} {h}
\newcommand{\gfun} {g}

\newcommand{\true} {\textsf{true}}
\newcommand{\false} {\textsf{false}}

\newcommand{\finalStates}{\texttt{Fin}}

\newcommand{\lang} {\mathcal{L}}

\newcommand{\WSTS} {\text{WSTS}\xspace}
\newcommand{\wqo} {WQO}
\newcommand{\pre} {\textsf{Pre}}
\newcommand{\upclosure}[1] {#1\!\uparrow}

\newcommand{\mset} {\textsf{m}}

\newcommand{\norm}[1] {\|#1\|}

\newcommand{\finalABS}{\ensuremath{\ABS_{\tt final}}}

\newcommand{\lossy}{\mathcal{L}}
\newcommand{\lossytuples}{\tuple{\States_\lossy, \Trans_\lossy, \initstate^\lossy, \finalStates_\lossy}}
\newcommand{\counterop} {\textsf{cop}}
\newcommand{\inc} {\textsf{inc}}
\newcommand{\dec} {\textsf{dec}}
\newcommand{\zero} {\textsf{zero}}

\newcommand{\nonce}{\textsf{Nonce}}
\newcommand{\CComment}[1]{\mbox{}}

\newcommand{\MAX} {\mathbf{Max}}

\newcommand{\abstraction} {\texttt{abs}}
\newcommand{\ABS} {\texttt{Abs}}
\newcommand{\initABS} {\ABS_\texttt{init}}
\newcommand{\abs} {\texttt{ab}}
\newcommand{\functions}[2] {[#1 \to #2]}

\newcommand{\kreadrep}{can-replicate-with-\ensuremath{k}-demands\xspace}
\newcommand{\kreadrepPar}[1]{can-replicate-on-\ensuremath{#1}-demands\xspace}
\newcommand{\kwritepr}{writes-\ensuremath{k}-times-before-consumption\xspace}
\newcommand{\kwriteprPar}[1]{writes-\ensuremath{#1}-times-before-consumption\xspace}
\newcommand{\NMAX}{\MAX}

\newcommand{\chain} {\textsf{Chain}}

\newcommand{\Thread}{ \texttt{Th}}
\newcommand{\thread}{ \texttt{th}}
\newcommand {\xleadsto}[1] {\overset{#1}{\rightsquigarrow}}
\newcommand{\tins}[1]{\text{ins}(#1)}
\newcommand{\tdel}[1]{\text{del}(#1)}
\newcommand{\tprop}[1]{\text{prop}(#1)}
\newcommand{\cmd}{\textsf{cmd}}

\title{Parameterized Reachability for Register Machines with Data}
\titlerunning{Parameterized Reachability for Register Machines}

\author{%
	Parosh Aziz Abdulla\inst{1} \and
	Mohammed Faouzi Atig\inst{1} \and
	Abhijith R Nair\inst{2} \and
	K Narayan Kumar\inst{3} \and
	Prakash Saivasan\inst{2}%
}

\authorrunning{P. A. Abdulla et al.}

\institute{%
	Department of Information Technology, Uppsala University, Sweden\\
	\email{\{parosh,mohamed\_faouzi.atig\}@it.uu.se}
	\and
	The Institute of Mathematical Sciences, Homi Bhabha National Institute, India\\
	\email{\{abhijitrn,prakashs\}@imsc.res.in}
	\and
	Chennai Mathematical Institute, India\\
	\email{kumar@cmi.ac.in}%
}

\begin{document}
	
	\maketitle
	
	\begin{abstract}
		We investigate the parameterized reachability problem for concurrent register machines over infinite data domains.
		In this framework, each machine is a program equipped with a set of local registers, the communication
		across machines is mediated through a set of shared registers. Both local and shared registers can take
		values from an infinite data domain. The program's primitive operations include copying values
		between registers, assigning constants, comparing registers for (dis-)equality, and
		nondeterministic assignments that store an arbitrary domain value into a local register.
		
		The parameterized reachability problem considers a program and a target location, asking
		whether there exists some $n \in \Naturals$ such that an execution of $n$ identical machines (referred to as instances)
		results in at least one instance reaching the specified location. We show that this problem is
		{\sc Pspace-complete} in the general case and it becomes undecidable if a freshness assumption (i.e., each assignment must produce a unique value
		distinct from all the constants and from any of the previously generated values)
		is applied to nondeterministic assignments. This undecidability persists even for systems restricted to two
		shared and two local registers. Finally, we establish  decidability results for two
		restricted settings: when each thread is limited to a single local register, or when the
		system utilizes only one shared register.
		
	\end{abstract}
	
	\section{Introduction}
\label{sec:intro}

A \emph{parameterized system} consists of an arbitrary number of identical concurrent instances. 
\emph{Parameterized verification} refers to the task of verifying the correctness of such 
systems independently of the number of participating instances. This topic has been the focus 
of extensive theoretical and practical research for nearly three decades (see, e.g., 
\cite{DBLP:journals/ipl/AptK86,DBLP:journals/jacm/GermanS92,DBLP:journals/sigact/BloemJKKRVW16,DBLP:journals/sttt/AbdullaD16}), 
and is discussed in detail in a dedicated chapter of the \emph{Handbook of Model Checking}~\cite{DBLP:reference/mc/AbdullaST18}.

Initial research in this field primarily considered systems where each instance is \emph{finite-state}~\cite{DBLP:reference/mc/AbdullaST18}. 
Even under this restriction, the verification problem remains inherently difficult, as the unbounded 
number of instances generates an infinite state space that must be effectively managed.
Subsequent studies have extended these techniques to more expressive computational models, 
including pushdown systems~\cite{DBLP:conf/lics/Kahlon08}, timed systems~\cite{DBLP:journals/tcs/AbdullaJ03},  broadcast networks \cite{DelzannoST16,Delzanno16,DelzannoT13}, 
and programs operating under weak memory models~\cite{DBLP:conf/podc/KrishnaG0C22,DBLP:journals/pacmpl/AbdullaAR20}.

In this paper, we investigate parameterized verification in settings where the instances operate on 
data drawn from an infinite domain. We model each process as a classical register machine 
equipped with a set of local registers and access to a finite set of shared global registers. 
Supported operations include comparisons, assignments, and copying. 
Additionally, all the instances have access to a finite set of 
shared global registers.
Assignment operations may include copying the value of one register to another or reading an 
arbitrary value from the environment into a local register. We refer to the latter as a 
\emph{non-deterministic assignment}.
We distinguish between 
two types of \emph{non-deterministic assignments}: the \emph{unconstrained} variant, 
where an instance may assign any domain value, and the \emph{constrained} variant (freshness constraint), 
where assigned values must be distinct from all values previously observed.

We focus on the parameterized \emph{reachability} problem: determining if a system, starting 
from a common initial state, can reach a configuration where at least one instance is in a 
specified \emph{target state}. 
Our results demonstrate that the unconstrained case is {\sc Pspace}-complete. In contrast, 
the constrained case is undecidable, even for systems limited to two global and two 
local registers. We complete the landscape by showing that decidability is recovered 
if the system is restricted to either a single global register or a single local register 
per each instance. 
We prove that the former case is {\sc Pspace}-complete. For the single local register case, we show that 
parameterized reachability is decidable in  non-primitive recursive time via a reduction to data nets \cite{Rosa-VelardoF11}. 
Furthermore, we establish an non-primitive recursive lower bound for this restricted setting.

\noindent
\textbf{Related Work:}
Parameterized verification has been extensively studied (see surveys~\cite{DBLP:journals/sigact/BloemJKKRVW16,DBLP:reference/mc/AbdullaST18}). Since the problem is generally undecidable for certain topologies and communication primitives even with finite-state processes~\cite{DBLP:journals/ipl/AptK86}, research has primarily focused on identifying decidable subclasses based on topology, communication mechanisms, and process features~\cite{DBLP:conf/lics/EsparzaFM99,DBLP:conf/concur/DelzannoSZ10,DBLP:journals/toplas/GantyM12,DBLP:conf/csl/EmersonK04,DBLP:conf/charme/EmersonK03,DBLP:conf/tacas/NamjoshiT16}. Concurrently, various abstraction techniques have been developed for infinite-state reasoning, including regular model checking~\cite{DBLP:journals/tcs/KestenMMPS01,DBLP:conf/cav/BoigelotLW03,DBLP:journals/sttt/BouajjaniHRV12,DBLP:journals/sttt/Abdulla12}, monotonic abstraction~\cite{DBLP:conf/concur/AbdullaCDHHR10}, and symmetry reduction~\cite{DBLP:conf/cav/KaiserKW10,DBLP:conf/lics/EmersonHT00,DBLP:journals/sttt/AbdullaHH16}.

Early work established that model checking arbitrary finite-state processes communicating via rendezvous is {\sc Expspace}-complete~\cite{DBLP:journals/jacm/GermanS92}. This was later extended to pushdown automata: reachability is decidable for arbitrarily many identical pushdown processes~\cite{DBLP:conf/lics/Kahlon08}, and was later shown to be  {\sc Pspace}-complete when generalized to a \emph{leader} with arbitrarily many helpers~\cite{DBLP:conf/fsttcs/Hague11,DBLP:journals/jacm/EsparzaGM16}. %

The parametrized reachability has been studied in  other contexts as well. In real-time systems, parametrized reachability is decidable for processes with a single clock~\cite{DBLP:journals/tcs/AbdullaJ03} but becomes undecidable with two clocks~\cite{DBLP:conf/lics/AbdullaDM04}, though restricted versions remain {\sc Pspace}-complete~\cite{DBLP:conf/concur/AbdullaACMT18}. More recently, verification of finite-state processes has been explored under weak memory models, such as Total Store Order (TSO)~\cite{DBLP:journals/pacmpl/AbdullaAR20} and Release-Acquire semantics~\cite{DBLP:conf/podc/KrishnaG0C22}.

Finally, our formalism also relates to broadcast networks~\cite{ARBCWK2021,Balasubramanian18,DelzannoST16,Delzanno16,DelzannoT13}, where parameterized verification over infinite data domains has been studied~\cite{DelzannoST13,GuillouMW24}. Although these models are generally incomparable to ours, we show that our results yield a {\sc Pspace} upper bound for broadcast networks with a single global register over a unary alphabet. We note that the complexity aspect of the  parametrized reachability for the unary alphabet in this setting was previously unaddressed; in the general (non-unary) case, the problem was shown to be non-primitive recursive complete \cite{GuillouMW24}.

	\section{Notation and Model}

We denote by $\Naturals$ the set of all natural numbers. For any $n \in \Naturals$, we let $[1,n]$ denote the set $\{1, \dots,n\}$. Given sets $\sset$ and $\tset$, $\functions{\sset}{\tset}$ denotes the set of all functions from $\sset$ to $\tset$. For a function $\ffun \in \functions{\sset}{\tset}$, $\updatefun{\ffun}{a}{b}$ represents the updated function $\ffun'$ such that $\ffun'(a) = b$ and $\ffun'(s) = \ffun(s)$ for all $s \in \sset \setminus \{a\}$. For a finite alphabet $\Sigma$ and a word $\word \in \Sigma^*$, the width of the word, $\wdwidth(\word)$, is the maximum frequency of any single letter in $\word$ (e.g., $\wdwidth(ababa) = 3$).  Given a set $X$, we define a \emph{multiset} over it as a map $\mset: X \to \Naturals$, we let $X(m) = \{ x \in X \mid\mset(x) >0 \}$ and say the $\mset$ is finite if $X(m)$  is finite. Given two multisets $\mset_1,\mset_2$, we define their addition as the multiset $\mset_1+\mset_2$, where for all $x \in X$, $\mset_1+\mset_2(x) = \mset_1(x) + \mset_2(x)$. We say $\mset_1 \leq \mset_2$ if for all $x \in X$, $\mset_1(x) \leq \mset_2(x)$, in this case, we define its subtraction as the multiset $\mset_2 - \mset_1$, where for all $x \in X$, $\mset_2 - \mset_1 (x) = \mset_2(x) - \mset_1(x)$. We use $ X^\oplus$ to denote the set of all finite multisets over $X$.
\vspace{-10pt}
\begin{figure}[h!]
	\centering
	\vspace{-1ex} %
	\setlength{\fboxsep}{6pt} %
	\fbox{
		$
		\renewcommand{\arraystretch}{1.2}
		\begin{array}{lcl}
			\mathtt{prog} & ::= & \mathtt{var}\;x^* \;\; \mathtt{instr}^*, \quad \mathtt{instr} ::= \mathtt{lbl}:\mathtt{stmt} \\
			\mathtt{stmt} & ::= & \xvar := \yvar \mid \zvar := * \mid \mathtt{goto}\ \lab \mid \mathtt{if\;exp\;then\;goto\;} \lab \mathtt{\ end} \\
			\mathtt{exp}  & ::= & \xvar = \yvar \mid \xvar \neq \yvar  \mid \star
		\end{array}
		$
	}
	\vspace{-1ex} %
	\caption{Syntax for \PDS. Here, $\xvar \in \localregisters \cup \globalregs$, $\zvar \in \localregisters$, $y \in \globalregs \cup \localregisters \cup \domain$, and $\lab \in \Labels$.}
	\label{prg:rpsyntax:fig}
	\vspace{-2ex} %
\end{figure}

We investigate parameterized programs consisting of an arbitrary number of identical concurrent instances, or \emph{threads}. These programs operate over registers containing values from a fixed, potentially infinite data domain $\domain$. Each thread manages its own private local registers, while communication is achieved through a set of shared (\emph{global}) registers. %

We assume that $\domain$ contains the constant $0$, which serves as the initial value for all registers. A \emph{Parameterized Program} (\PDS) is defined as a pair $\pds=(\globalregs,P)$, where $\globalregs$ is the set of global registers and $P$ is the program  defined by the syntax in Figure \ref{prg:rpsyntax:fig}. A program begins with a declaration of local registers $\localregisters$ (disjoint from $\globalregs$), followed by a finite sequence of labelled instructions of the form $\lambda: s$, where $\lambda$ is a unique label and $s$ is a statement. The label $\initlab$ designates the entry point of the program. Let $\constants \subseteq \domain$ denote the set of constants appearing in program $P$ along with the special constant $0$. We let $\Labels$ be the set of all labels in the program and define the system size as $\norm{\pds} = |\localregisters| + | \globalregs| + |\constants| + |\Labels|$. 

Statements include register assignments, unconditional jumps, and conditional jumps. We specifically include a non-deterministic assignment ($x:=*$) that stores an arbitrary value from $\domain$ into $x$. To facilitate program modelling, we employ the following macros, its implementation is given in appendix \ref{sec:highmac}:

\noindent
- $\mathtt{assert}(e_1, \dots,e_n):$ ensures that a set of expressions evaluate to true (not necessarily atomically); if any expression is false, the execution is blocked.

\noindent
- $\mathtt{if\;e\;} \mathtt{then}\ \{\textsf{instr}_1\}\ \mathtt{else}\ \{\textsf{instr}_2\}\ \mathtt{end}:$ allows for conditional branching based on the evaluation of expression $e$.

\noindent
- $\textsf{If}(\star)\{\textsf{instr}_1\}; \dots; \{\textsf{instr}_k\}:$ represents a non-deterministic choice where the system picks exactly one   instruction to execute from the available collection.

Given a parameterized program $\pds$  and an $n \in \Naturals$, the \emph{$n$-configuration} is the configuration involving $n$-threads. It is a tuple $\pdsconf = \tuple{\labelmap, \lregmap,\gregmap}$, where $\labelmap:[1,n] \to \Labels$ maps each thread to a program label, $\lregmap: [1,n] \to [\localregisters \mapsto \domain ]$ maps each thread to its local register valuation, and $\gregmap: \globalregs \to \domain$ maps each global register to its current value. We let $\Confof{\pds}(n)$ be the set of all $n$-configurations and $\Confof{\pds} = \bigcup_{n \in \Naturals} \Confof{\pds}(n)$ be the set of all possible configurations. The initial $n$-configuration is $\initpdsconf(n) = \tuple{\labelmapinit, \vecregmapinit,\initgregmap}$, where all registers are initialized to $0$ and $\labelmapinit$ maps every thread to the initial label $\initlab$.

For any configuration $\pdsconf = \tuple{\labelmap, \lregmap,\gregmap} \in \Confof{\pds}(n)$ , we let $\labelmap(\pdsconf)=\labelmap$, $\lregmap(\pdsconf)=\lregmap$, and $\gregmap(\pdsconf)=\gregmap$. The set $\domain(\pdsconf) = \{\lregmap(i)(\lreg) \mid \lreg \in \localregisters, i \in [1,n] \} \cup \{\gregmap(\greg) \mid \greg \in \globalregs \}$ are all values present in the configuration. For any $\lab \in \Labels$, $\nextof{\lab}$ indicates the immediate successor label.

Let $\pdsconf_1 = \tuple{\labelmap_1, \lregmap_1,\gregmap_1}$ and $\pdsconf_2 = \tuple{\labelmap_2, \lregmap_2,\gregmap_2}$ be $n$-configurations. A transition $\pdsconf_1 \movesto{\instr}_i \pdsconf_2$ occurs if thread $i$ executes instruction $\instr$ of the form $\lambda: s$ according to the following rules:

\smallskip\noindent
 \textbf{Local assignment:} If $s$ is $\xvar := \yvar$ for $\xvar \in \localregisters$, then $\lregmap_2$ is obtained by updating the thread $i$'s local register with  value $t$ (defined below) and $\labelmap_2$ is obtained by updating its label to the next instruction.  That is, $\gregmap_2 = \gregmap_1$, $\labelmap_2 = \updatefun{\labelmap_1}{i}{\nextof{\labelmap_1(i)}}$ and $\lregmap_2 = \updatefun{\lregmap_1}{i}{\updatefun{\lregmap_1(i)}{\lreg}{t}}$ where  $t = \yvar$ if $\yvar \in \constants$, $t =  \lregmap_1(i)(\yvar) $ if $\yvar \in  \localregisters$, and   $t = \gregmap_1(\yvar)$ if  $\yvar \in  \globalregs$. 
 
	\smallskip\noindent
	 \textbf{Nonce Assignment:} If $s$ is $\xvar := \ast$, then the register $\xvar$ is assigned a value $\val \in \domain$ for some $\val \in \domain$. That is, $\labelmap_2 = \updatefun{\labelmap_1}{i}{\nextof{\labelmap_1(i)}}$. Moreover,  $\gregmap_2=\gregmap_1$ and  $\lregmap_2 = \updatefun{\lregmap_1}{i}{\updatefun{\lregmap_1(i)}{\lreg}{\val}}$,  we let $\nonce(\pdsconf_1 \movesto{\instr}_i \pdsconf_2) = \val$.
	
		\smallskip\noindent
		 \textbf{Global Assignment:} If $s$ is $\rvar := \yvar$ for $\rvar \in \globalregs$, then $\gregmap_2$ is updated with the value of $\yvar$ and the label of the thread $i$ is updated to the next instruction. That is,  $\labelmap_2 = \updatefun{\labelmap_1}{i}{\nextof{\labelmap_1(i)}}$,   $\lregmap_2 = \lregmap_1$ and $\gregmap_2 = \updatefun{\gregmap_1}{\rvar}{t}$, where $t = \lregmap_1(i)(\yvar)$ if $\yvar \in \localregisters$, $t = \gregmap_1(\yvar)$ if  $\yvar \in  \globalregs$, and $ t = \yvar$ otherwise.
	
		\smallskip\noindent
		 \textbf{Unconditional choice:} If $s$ is $\mathtt{goto\;\lab_1}$, then the label of thread $i$ is updated accordingly. That is, $\labelmap_2 = \updatefun{\labelmap_1}{i}{\lab_1}$, other components remain the same.
		 
		\smallskip\noindent
		 \textbf{Conditional choice:} If $s$ is $\mathtt{if\; e \;then\; \mathtt{goto\;\lab_1\ end}}$, thread $i$ moves to $\lab_1$ if $e$ evaluates to $\true$, and to $\nextof{\labelmap_1(i)}$ otherwise. That is,   
	$\labelmap_2 = \updatefun{\labelmap_1}{i}{\lab_1}$ if $e$ evaluates to $\true$, and 
	$\labelmap_2 = \updatefun{\labelmap_1}{i}{\nextof{\labelmap_1(i)}}$ otherwise.  The other components remain unchanged.

		\smallskip\noindent
		 \textbf{Non-deterministic choice:} If $s$ is $\mathtt{if\;\star\;then\; \mathtt{goto\;\lab_1\ end}}$, thread $i$ non-deterministically moves to $\lab \in \{\lab_1, \nextof{\labelmap_1(i)}\}$. That is, $\labelmap_2 = \updatefun{\labelmap_1}{i}{\lab}$, where $\lab \in \{\lab_1,\nextof{\labelmap_1(i)}\}$ is chosen non-deterministically, rest of the components is untouched.

	\smallskip
Observe that if $\pdsconf_1 \movesto{\instr}_i \pdsconf_2$, then $\pdsconf_1$ and $\pdsconf_2$ have the same number of threads.
An \emph{$n$-run} $\pi$ is a sequence of zero or more moves, $\pi = \pdsconf_1 \movesto{\instr_1}_{j_1} \pdsconf_2 \movesto{\instr_2}_{j_2} \cdots \movesto{\instr_{m-1}}_{j_{m-1}} \pdsconf_m$, where each $\pdsconf_i$ is an $n$-configuration. We say that $\pi$ is an initial $n$-run if it starts from the initial configuration, that is $\pdsconf_1 = \initpdsconf(n)$. 

A run is \emph{fresh} if every nonce generated throughout the execution is unique and distinct from any constant in $\constants$. Formally, $\pi$ is fresh if for all $i \in [1, m-1]$ such that $\instr_{i} = (\ell_i:\xvar := \ast)$:
$\nonce(\pdsconf_i \movesto{\instr_{i}} \pdsconf_{i+1}) \not \in \constants$ 
and for all $j < i$, if $\instr_{j} = (\ell_j: \yvar := \ast)$, then $\nonce(\pdsconf_i \movesto{\instr_i} \pdsconf_{i+1}) \neq \nonce(\pdsconf_j \movesto{\instr_j} \pdsconf_{j+1})$. 

We write $\pdsconf \movesto{} \pdsconf'$ to indicate the existence of an instruction $\instr$ and a thread index $i$ such that $\pdsconf \movesto{\instr}_i \pdsconf'$, and  use $\movesto{}^*$ to denote the reflexive-transitive closure of $\movesto{}$. An example of a run is provided in Figure \ref{main:fig:2run}. This specific run involves two threads and satisfies the freshness condition, as the values assigned in each non-deterministic step are pairwise distinct.

Given two $n$-runs $\run_1$ and $\run_2$ such that the terminal configuration of $\run_1$ matches the starting configuration of $\run_2$, their concatenation is denoted by $\run_1\run_2$. We now proceed to formulate the verification problems of interest.

\noindent\textbf{Parameterized Reachability (PReach):}
\problemframe
{	A parameterized program $\pds$ and a target label $\lab \in \Labels$}
{ Does there exist $n \in \Naturals$ and an $n$-run 
	$\initpdsconf(n) \movesto{}^* \tuple{\labelmap, \lregmap,\gregmap}$ such that $\labelmap(1) = \lab$}

\textbf{Parameterized Fresh Reachability (PFReach):}

\problemframe
{	A parameterized program $\pds$ and a target label $\lab \in \Labels$}
{Does there exist $n \in \Naturals$ and a \emph{fresh} $n$-run 
		$\initpdsconf(n) \movesto{}^* \tuple{\labelmap, \lregmap,\gregmap}$ such that $\labelmap(1) = \lab$}

\medskip 
Observe that assuming the target thread index is $i = 1$ is without loss of generality, as the transition semantics are agnostic to thread identities. In what follows, we investigate the decidability and complexity of these two problems. To facilitate our analysis, we introduce the notion of \emph{pseudo-runs} and \emph{projection}.

\begin{figure}[h!]
	\centering
	\includegraphics[trim= 0.5cm 0.25cm 0.25cm 0,scale=0.6]{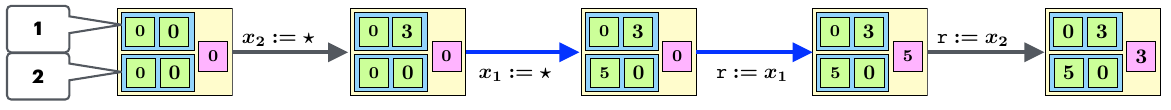} 
	\caption{An example of a $2$-run. In the figure, each thread has $2$ local registers and there is a single shared register. The state component of each thread is abstracted away.}
		\vspace{-2ex} %
	\label{main:fig:2run} 
\end{figure} 
\begin{figure*}
	\centering
	\includegraphics[trim= 0.25cm 0.25cm 0.25cm 0.25cm, scale=0.4]{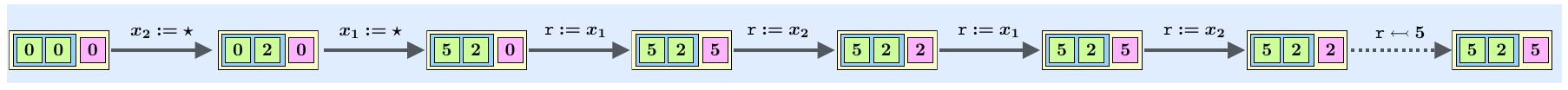}
	\caption{Example of a pseudo-$1$-run, involving $2$ local registers and a single global register}\label{main:fig:psuedorun}
	\vspace{-2ex} %
\end{figure*}

\noindent
\textbf{Pseudo Runs:\ }
Let $\pdsconf_1 = \tuple{\labelmap_1, \lregmap_1,\gregmap_1}$ be a configuration, $\greg \in \globalregs$ a global register and $\val \in \domain$ be a data value. A \emph{pseudo-assignment} transition at $\pdsconf_1$ sets the register $\greg$ to $\val$, resulting in the configuration $\pdsconf_2 = \tuple{\labelmap_1, \lregmap_1,\updatefun{\gregmap_1}{\greg}{\val}}$. We write $\pdsconf_1 \dotarrow{(\greg \leftarrowtail \val)} \pdsconf_2$ to denote such a transition and say $\val$ is \emph{pseudo-assigned} to $\greg$. A \emph{pseudo-run} extends the notion of a run by permitting pseudo-assignment transitions to be interspersed with regular transitions. Thus, a pseudo-run is a sequence of the form $\rho = \pdsconf_1 \movesto{}^*\pdsconf'_1 \dotarrow{(\greg_1 \leftarrowtail \val_1)} \pdsconf_2 \movesto{}^*\pdsconf'_2 \cdots \pdsconf_m \movesto{}^* \pdsconf'_m$ where each $\pdsconf_i \movesto{}^* \pdsconf'_i$ is a run. 

We let the \emph{support} of the pseudo-run $\rho$ (denoted $\supp(\rho)$) be the sequence $(\greg_1, \val_1)(\greg_2, \val_2) \dots (\greg_m, \val_m)$. Observe that the number of threads remains invariant along a pseudo-run; a pseudo $n$-run involves exactly $n$ threads. We extend  $\nonce$ to pseudo-runs where $\nonce(\rho)$ refers to the set of nonce values generated along the standard segments: $\nonce(\rho) = \bigcup_{i=1}^m \nonce(\pdsconf_i \movesto{}^* \pdsconf'_i)$. If a pseudo-run involves $n$ threads, we will refer to it as pseudo-$n$-run. Figure~\ref{main:fig:psuedorun} provides an example of a pseudo-$1$-run.

\smallskip
\noindent
\textbf{Projection:\ } We define the projection of a run onto a specific thread to obtain the pseudo-runs that will be used in our analysis. Let $\run$ be a run and $i$ be any thread. Then $\run$ can be decomposed as $\run = \run_0 \run'_0 \run_1 \run'_1 \run_2 \ldots \run_m \run'_m \run_{m+1}$ where each $\run_j$ consists of moves by threads other than $i$, and each $\run'_j$ consists only of moves by thread $i$. By restricting the configurations in $\run'_j$ onto thread $i$, we obtain a sequence of $1$-runs $\irun_0, \irun_1, \dots, \irun_m$. 

While each $\irun_j = \pdsconf_j \xrightarrow{}^* \pdsconf'_j$ is locally valid, their concatenation $\irun_j\irun_{j+1}$ is a valid $1$-run only if the shared register values in $\pdsconf'_j$ and $\pdsconf_{j+1}$ are identical. By inserting pseudo-assignments between $\irun_j$ and $\irun_{j+1}$ whenever the shared register values differ, we obtain a pseudo-$1$-run $\irun$. We call $\irun$ the \emph{projection of $i$ on $\pi$}.

To distinguish pseudo-runs from standard runs, we use $\Movesto{}$ to denote its transitions: $\irun = \iconf_1 \Movesto{} \iconf_2 \Movesto{} \dots \Movesto{} \iconf_m$, where $\Movesto{}$ represents either $\movesto{}$ or $\dotarrow{}$. We say that a pseudo-run $\irun$ is \emph{fresh} if all nonces generated are non-constants and are distinct from all previously observed values (including nonces and pseudo-assignments) in the run.

\section{Undecidability of PFReach} \label{sec:undec}

We demonstrate that the \textsf{PFReach} problem is undecidable through a reduction from the reachability problem for perfect channel systems. 
Formally, a \emph{perfect channel system} (\PCS) is defined as a tuple $\pcs = \tuple{\States,\Alphabet,\Trans,\initstate,F}$, 
where $\States$ is a finite set of states, $\Alphabet$ is a finite message alphabet, and 
$\Trans \subseteq \States \times \{ !a,?a \mid a \in \Alphabet\} \times \States$ is the transition relation. 
The components $\initstate \in \States$ and $F \subseteq \States$ represent the initial state and the set of final states, respectively. 
A transition of the form $\tuple{q_1,!a,q_2}$ appends message $a$ to the tail of a FIFO channel, while $\tuple{q_1,?a,q_2}$ 
removes message $a$ from the head.

A configuration of a \PCS is a pair $(q,w) \in \States \times \Alphabet^*$. For configurations $\conf_1 = \tuple{q_1,w_1}$ and 
$\conf_2 = \tuple{q_2,w_2}$, we define $\conf_1 \movesto{\tau} \conf_2$ if $\tau = \tuple{q_1,!a,q_2}$ and $w_2 = w_1 \cdot a$, 
or if $\tau = \tuple{q_1,?a,q_2}$ and $w_1 = a \cdot w_2$. The reachability problem, which asks if a state in $F$ can be 
reached from the initial configuration $\tuple{\initstate, \epsilon}$, is known to be undecidable~\cite{BrandZ83}.

\begin{wrapfigure}[10]{r}{0.3\textwidth}
	\vspace{-10pt}
	\centering
	\includegraphics[trim=0.25cm 0.5cm 0.25cm 1cm, scale=0.29]{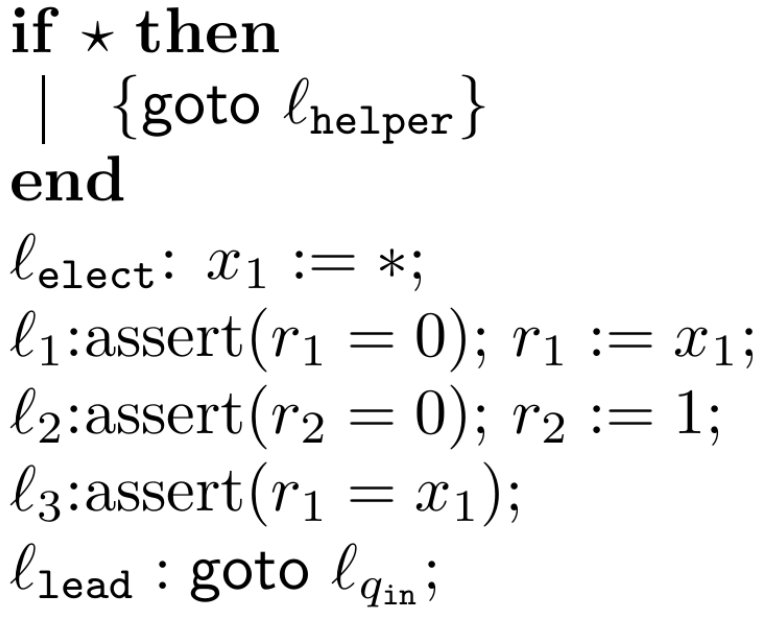}
	\caption{Leader Election}\label{main:fig:LE}
\end{wrapfigure}

For the reduction, given a \PCS $\pcs$, we construct a parameterized program to simulate its execution. The program uses two shared 
registers $\greg_1, \greg_2$ and two local registers $\lreg_1, \lreg_2$. The simulation consists of two phases: 
(i) a leader election phase and (ii) a simulation phase. The  simulation phase is activated only upon the successful election of a unique leader. 
The elected leader simulates the finite-state control of the \PCS and orchestrates the channel simulation, which is maintained 
via a chain of helper threads.

Every thread begins by executing the code fragment illustrated in Figure~\ref{main:fig:LE}. A thread may either jump to 
$\ell_{\tt helper}$ to act as a helper or enter the protocol starting at $\ell_{\tt elect}$ to compete for leadership.

\noindent
\textbf{Leader Election:\ }
The leader election protocol (Figure~\ref{main:fig:LE}) employs global registers $\greg_1, \greg_2$ and local register $\lreg_1$. 
Under the freshness assumption, the instruction at $\ell_{\tt elect}$ generates a unique identity (\emph{id}) for each thread 
in its local register $\lreg_1$. Each thread then attempts to write its id to $\greg_1$ (provided it can read a $0$ from it), 
toggles $\greg_2$ (line $\ell_2$), and verifies if its id persists in $\greg_1$ (line $\ell_3$). We prove that these constraints ensure at most one thread reaches 
the label $\ell_{\tt lead}$.

For an $n$-run $\pi$, let $\reachedLab(\pi,i)$ denote the set of labels reached by thread $i$. The following lemma formalizes the 
uniqueness of the leader:

\begin{lemma}\label{lem:undec:uniqueleader}
	Let $\pi = \pdsconf_1 \movesto{}^* \pdsconf_m$ be any fresh $n$-run. Then 
	$|\{ i \in [1,n] \mid \ell_{\tt lead} \in \reachedLab(\pi,i) \}| \leq 1$.
\end{lemma}
\begin{proof}
	Let $\tset = \{ i \in [1,n] \mid \ell_{\tt lead} \in \reachedLab(\pi,i) \}$. If $|\tset| > 0$, let $\pdsconf'$ be the first 
	configuration in $\pi$ where $\ell_{\tt lead}$ is reached by some thread $i$. At this point, $\greg_1$ must contain the id 
	of $i$ ($\gregmap(\pdsconf')(\greg_1) = \lregmap(\pdsconf')(i)(\lreg_1)$), and $\greg_2$ must be non-zero. 
	Any thread that writes to $\greg_1$ after $i$ will fail the assertion at $\ell_2$. If another thread $j$ had 
	written to $\greg_1$ before $i$, its id would have been overwritten by $i$, causing $j$ to fail the assertion at $\ell_3$. 
	Therefore, only thread $i$ can successfully reach $\ell_{\tt lead}$.
\end{proof}
\vspace{-10pt}
\noindent
\textbf{Simulation of \PCS $\pcs$:\ }
Towards the simulation, we will assume that a leader is elected. Only the global register $\greg_2$ will be used during the simulation. The role of the leader is to maintain the state of the $\PCS$ $\pcs$ and use a chain of helpers to maintain the channel (each thread in the chain stores a message in its state). 

Each helper has a unique id stored in its local register $\lreg_1$. The leader and each helper that is part of the chain stores the id of its successor in the chain in the register $\lreg_2$ (storing $\bot$ if there is none). To insert a message $a$, if the channel is empty, the leader writes $\ins(a)$ to $\greg_2$ starting off a selection protocol; participating helpers write their id onto $\greg_2$, and the leader reads one of them and copies it to its $\lreg_2$, identifying a potential winner. To communicate the identity of the winner unambiguously, it writes $\bot$ followed by the id of the winner to $\greg_2$. The helper that reads $\bot$ (thus avoiding a read of its own write) followed by its own id learns of its victory and conveys this knowledge along with the value it is storing back to the leader. The leader, by selecting only one helper ensures that only a helper with the correct value is inserted to the list. While a late-coming candidate may overwrite $\greg_2$ in many runs, resulting in a deadlock, the existence of a run with a unique winner suffices for our purposes.

Appending to the end of a non-empty chain proceeds as follows: the leader starts the protocol by writing $\prop(a)$ followed by the id of its neighbour. Information about the insertion travels as a wave from the leader to the end of the chain and back: each intermediate helper reads $\prop(a)$ from $\greg_2$ to learn of the operation, then its own id from $\greg_2$ (as the wave moves outwards), and responds by writing its neighbour's id onto $\greg_2$. The last helper executes a protocol similar to the leader's empty-channel induction to add a new helper via the selection process. For inward propagation, each intermediate helper (after ensuring it is not reading its own write) reads the value of its neighbour (stored in $\lreg_2$) from the global register $\greg_2$ and propagates the wave by writing its own id to $\greg_2$.

Dequeueing a message $a$ removes the \emph{head} of the channel (the leader's successor) and restores the chain. The head identifies itself unambiguously using non-determinism. The leader writes $\del(a)$ to $\greg_2$ followed by the id of its successor. If the head is the only thread in the chain, it identifies itself and responds to the leader. Otherwise, the tail thread (which can identify itself uniquely) writes its own id to $\greg_2$, starting a wave towards the leader. Each intermediate thread, upon reading its successor's id from $\greg_2$, replaces it with its own id. Only the thread at the head of the channel can read $\del(a)$, followed by its own id, followed by its successor's id, thus unambiguously confirming its identity. Each helper guesses its role and verifies it; incorrect guesses end in deadlocked configurations. Once identified, the head indicates this to the leader and transfers the id of its successor so that the leader may update its $\lreg_2$. We note that a simpler protocol involving only the leader and head would fail, as helpers must verify their position in the chain to avoid misleading neighbours. Detailed algorithms and correctness proofs are presented in the appendix \ref{sec:undec1}.

\begin{theorem}\label{thm:undec}
	The \PCRFR\ problem is undecidable.
\end{theorem}

	\section{\PCRFR\  with Single Shared Register}\label{sec:oneglobal}

We establish \textsc{Pspace}-completeness of \PCRFR\ for \PDS\ with one
shared register. For clarity, we assume that there are no constants other
than the initial value $0$, and that no run ever writes the value $0$.
Constants, including writes of $0$, add no new ideas, but they make the
lemmas and proofs longer. Since $0$ is never written, no run can support a
pseudo-assignment of $0$. So we only consider pseudo-runs whose support does
not contain $0$. Appendices~\ref{apx:sec:oneglobal:acyclic}
and~\ref{apx:sec:oneglobal:cyclic} show how to handle constants. They also
contain the missing proofs.
 
 Fix a parameterized program $\pds=(\{\greg\},P)$ and a target label $\lab$. To solve the \PCRFR\ problem, we must determine the existence of a fresh execution wherein at least one thread reaches $\lab$. A \emph{witnessing run} is defined as a minimal such execution, meaning no strict prefix of the run visits $\lab$.
 
 For any $\pds$ with a single shared register, the support of any pseudo-run $\rho$ is of the form $\supp(\rho) \in (\{ \greg \} \times \domain)^*$. In this section, we simplify notation by dropping the reference to $\greg$, letting $\supp(\rho) \in \domain^*$. We denote the set of values appearing in this support as $\asupp(\rho)$. The term \emph{width of $\irun$} refers to $\wdwidth(\supp(\irun))$. Throughout this section, all runs are assumed to be fresh unless stated otherwise. We begin by introducing necessary notations followed by an intuition of our proof strategy.
 
 Given a word $w = \val_1 \dots \val_l \in \domain^*$, we say a pseudo-run $\irun = \iconf_0 \movesto{\tau_0} \iconf_1 \dots \movesto{\tau_{m-1}} \iconf_m$ \emph{writes} $w$ if there exists a subsequence of moves $\tau_{i_1}, \dots, \tau_{i_l}$ at configurations $\iconf_{i_1}, \dots, \iconf_{i_l}$ such that each $\tau_{i_j}$ is a write instruction to $\greg$ resulting in $\greg(\iconf_{i_j}) = \val_j$. Essentially, $w$ is written if its constituent values are assigned to the shared register by the moves of the pseudo-run. 
 
 A run $\run$ \emph{supports} a pseudo-run $\irun$ if $\run$ writes $\supp(\irun)$. We define $\run$ as a $k$\emph{-writing run} if it writes a word of width at least $k$. For example, the run in Figure \ref{main:fig:2run} writes the sequence $5, 3$ and is thus $1$-writing. The supporting $1$-run at the bottom of Figure \ref{main:fig:acyclic} is a $2$-writing run, as it writes the sequence $1, 2, 1$.

 A pseudo-run $\irun$ is \emph{cyclic} if $\nonce(\irun) \cap \asupp(\irun) \neq \emptyset$; that is, it is cyclic if it reads—via a pseudo-assignment—a nonce value that it previously generated. Otherwise, $\irun$ is \emph{acyclic}. A witnessing run is classified as cyclic or acyclic based on the pseudo-run of the witnessing thread. The pseudo-run  in Figure \ref{main:fig:psuedorun} is cyclic because the nonce $5$ that it generated is subsequently pseudo-assigned. Conversely, the top part of Figure \ref{main:fig:acyclic} illustrates an acyclic pseudo-$1$-run.
 \vspace{-10pt}
 \begin{figure*}
	\centering
	\includegraphics[trim=0.5cm 0.5cm 0.5cm 0.5cm cm, scale=0.5]{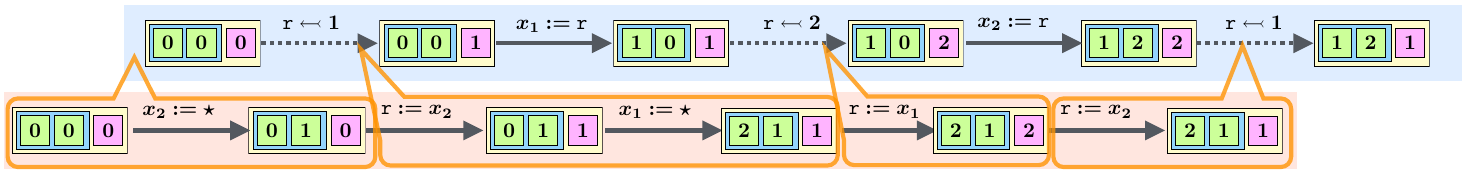}
	\caption{An acyclic pseudo $1$-run whose width is bounded by $2$, a corresponding supporting $1$-run and their combination}\label{main:fig:acyclic}
	\vspace{-2ex} %
\end{figure*}

\noindent
\textbf{Proof strategy:\ }
In any fresh witnessing run $\pi$, there is a thread that reaches the label $\lab$. The projection of $\pi$ onto this thread yields a pseudo-$1$-run reaching $\lab$ (a \emph{witnessing pseudo-$1$-run}). In the inverse direction, the existence of a witnessing pseudo-$1$-run does not inherently guarantee a witnessing run; elevating a pseudo-run to a full run requires the construction of an appropriate supporting run. We identify the necessary and sufficient conditions for this elevation, distinguishing between the cyclic and acyclic cases.  We will remark about the need to consider these two cases separately at the end of this section. 

We demonstrate that an acyclic witnessing run exists if and only if there exists an integer $k$, a witnessing acyclic pseudo-$1$-run $\irun$ of width $k$, and a corresponding $k$-writing run. The intuition for this characterization is that an acyclic witnessing run provides an acyclic pseudo-$1$-run (of some width $k$) through projection, while the run itself acts as the $k$-writing support. Conversely, we show that a $k$-writing run can be configured to write any sequence of values whose width does not exceed $k$, allowing it to support the acyclic pseudo-$1$-run. To achieve a {\sc PSpace} complexity bound, we prove it is sufficient to restrict $k$ such that $k \leq \MAX$, where $\MAX =\norm{\pds}^{ \bigO(\norm{\pds})} $ is at most exponential in the size of the system. We then show that checking the existence of these components for $k \leq \MAX$ can be performed in {\sc PSpace}. Figure \ref{main:fig:acyclic} provides an example of an acyclic witnessing pseudo-$1$-run and a corresponding supporting run. 

The cyclic case is handled via a different mechanism. For a cyclic pseudo-$1$-run, the support must include nonces generated by the pseudo-$1$-run itself. The supporting run must therefore be capable of reading values supplied by the pseudo-$1$-run and writing them back to the global register. We formalize this using the notion of a \emph{replicating run}---a run in which at least one thread reads a value from the shared register and subsequently writes it back. 

While any cyclic witnessing run implies the existence of a cyclic pseudo-$1$-run and a replicating run, not all replicating runs provide sufficient support. A replicating thread might do multiple reads of a value before it can replicate it by writing to the shared register. This necessitates that the pseudo-$1$-run writes that specific value to the register at least that many times. We prove that a cyclic witnessing run exists if and only if there is a cyclic witnessing pseudo-$1$-run $\irun$ that writes each value $\val \in \nonce(\irun) \cap \asupp(\irun)$ at least $k$ times (for some $k$), and a replicating run that writes the value after reading it at most $k$ times. We further show that it is sufficient to bound $k \leq \MAX$ and that searching for such a pair is possible in {\sc PSpace}.
Given that reachability in register automata is already known to be {\sc Pspace}-hard, we then obtain the following result.

\begin{theorem}
	The \PCRFR\  with single shared register is {\sc Pspace complete}.	
\end{theorem}

\subsection{Searching for acyclic witnessing runs} \label{main:sec:acyclic}

The main result of this subsection is to establish the following theorem, showing that determining the existence of an acyclic witnessing run is in {\sc Pspace}. 

\begin{theorem}\label{thm:dec:acyclic}
	The existence of an acyclic witnessing run can be determined in space polynomial in $\norm{\pds}$.
\end{theorem}

We begin with a lemma that provides a necessary and sufficient condition for the existence of an acyclic witnessing runs. The remainder of the section is devoted to establishing that this condition can be effectively checked in {\sc Pspace}. To facilitate our complexity analysis, we first establish a structural characterization that reduces the existence of a global witnessing run to the existence of a local pseudo-run and its corresponding support.

\begin{lemma}\label{lem:dec:1greg:forward}
	There is an acyclic witnessing run if and only if there is an acyclic witnessing pseudo $1$-run and a supporting run.
\end{lemma}

\begin{proof} 
	Let $\run$ be a witnessing acyclic run. By projecting $\run$ onto the witnessing thread and adding suitable pseudo transitions, we obtain an acyclic pseudo $1$-run for which $\run$ itself serves as a supporting run.
	\smallskip
	
	For the converse, let $\irun$ be the given witnessing acyclic pseudo $1$-run and $w = \val_1 \dots \val_k$ be its support. We decompose $\irun$ based on its pseudo-assignments as $\rho = \rho_0 \dotarrow{(\greg \leftarrowtail \val_1)} \rho_1 \dotarrow{(\greg \leftarrowtail \val_2)} \cdots \dotarrow{(\greg \leftarrowtail \val_k)} \rho_k$. 
	
	Since $\pi$ writes $w$, we decompose it as $\pi = \pi_0 \cdot \pi_1 \cdot \pi_2 \dots \pi_k$, where $\pi_0$ is the prefix of $\pi$ up to, but not including, its first write to $\greg$. For each $i \in [1,k]$, $\pi_i$ is the sequence of moves following $\pi_{i-1}$ that begins with a writes and ends with  the earliest move writing $\val_i$ to $\greg$ and includes all subsequent moves that do not write to $\greg$. 
	Thus, for each $i \in [1,k]$, $\pi_i$ starts with a write to the shared register, and its last write to $\greg$ results in the value $\val_i$. This segment also accounts for all subsequent transitions that read this value back.
	
	Consequently, $\pi_i$ can be scheduled to execute independently of the current value of $\greg$, leaving $\greg$ with the value $\val_i$ upon completion. It follows that $\pi'=\pi_0 \rho_0 \pi_1 \rho_1 \pi_2 \rho_2 \dots \pi_k \rho_k$ is a valid run of $\pds$: the execution of $\pi_i$ ($1 \leq i \leq k$) leaves the value $\val_i$ in $\greg$, thereby simulating the effect of $\dotarrow{\greg \leftarrowtail \val_i}$. Furthermore, $\pi_i$ overwrites $\greg$ with its first transition and is thus unaffected by the preceding execution of $\rho_{i-1}$. Clearly, the label $\ell$ is reached at the end of $\pi'$. To guarantee that the resulting run is fresh, we assume $\nonce(\irun) \cap \nonce(\pi) = \emptyset$. This is without loss of generality because (a) nonces in $\rho$ can be renamed and (b) by the acyclicity of $\irun$, none of these nonces appear in the sequence $\val_1, \ldots, \val_k$.
\end{proof}
We next provide an algorithm to determine the existence of witnessing pseudo $1$-runs and corresponding supporting runs in {\sc PSpace}. This is achieved through several steps. First, we show that if a witnessing pseudo $1$-run exists, there is one whose length is bounded by a number $\MAX$, where $\MAX$ is at most exponential in the size of the system; the existence of such a run can be verified in {\sc Pspace}. 

Next, we demonstrate that the problem of computing a supporting run can be decoupled from the specific pseudo $1$-run in the following sense: the existence of a supporting run depends only on the width of the support. That is, if a run exists supporting a word $w$, and $w'$ is any word with the same width as $w$, then a run supporting $w'$ also exists. Thus, for any two acyclic pseudo $1$-runs with supports of identical width, either both have supporting runs or neither does. Since we can limit the width to $\MAX$, we finally show that we can determine the largest integer $k \leq \MAX$ for which a support can be found in {\sc Pspace}. We now describe the details of each of these steps.

\noindent
\textbf{Length Bounds for Pseudo $1$-runs:\ }
We establish a bound on the length of witnessing acyclic pseudo $1$-runs that need to be examined. The core idea is to abstract away the specific data values. Suppose $\iconf_i$ and $\iconf_j$ appear along a pseudo $1$-run (in that order) with the same program label, such that for any pair of registers (local or shared), their values are equal at $\iconf_i$ if and only if they are equal at $\iconf_j$. One would expect that by removing the segment from $\iconf_i$ to $\iconf_j$ and applying a suitable relabelling to the remainder of the run, a shorter witnessing pseudo $1$-run would result. 

While this intuition is largely correct, it requires a subtle refinement: a local register might store a nonce at configuration $\iconf_i$, yet hold a value obtained through a pseudo-assignment at $\iconf_j$. Under these circumstances, a naive relabelling of the data could inadvertently transform an acyclic pseudo-run into a cyclic one. We circumvent this issue by explicitly tracking the provenance of each register's value---specifically, whether it originates from a generated nonce or a pseudo-assignment. This distinction forms the core of the following lemma; a comprehensive proof is provided in the appendix \ref{apx:sec:oneglobal:acyclic}.

\begin{lemma}\label{lem:dec:1pseudobound}
If there is a witnessing acyclic  pseudo $1$-run then there is one whose
length is bounded by $\MAX$
\end{lemma}

Since the width of the support is fundamentally limited by the total number of instructions  in the pseudo-run, the length bound on the run implies a bound on the support width.

\begin{corollary}\label{cor:dec:widthbound}
If there is a witnessing acyclic pseudo $1$-run, there is one whose support has width bounded by $\MAX$.
\end{corollary}
The following lemma establishes that, for a given $k \in \Naturals$, checking for the existence of a witnessing acyclic pseudo-$1$-run whose width is bounded by $k$ can be achieved in {\sc PSpace}. A detailed treatment of the algorithm is provided in the appendix \ref{apx:sec:oneglobal:acyclic}, while a proof sketch is presented below.

\begin{lemma}\label{lem:dec:tracking}
	There is an algorithm that checks, for a given $k$, if there is a witnessing acyclic pseudo-$1$ run $\rho$ such that $\wdwidth(\supp(\rho)) \leq k$, which runs in space polynomial in the size of the program and $\log(k)$. 
\end{lemma}

\begin{proof}(sketch)
	Our first observation is that to track the width of an acyclic pseudo-$1$ run, it suffices to count the number of pseudo-reads of each non-nonce value currently present in the registers. 
	This facilitates a symbolic search for the run: we maintain in the state the current program label, the equality of registers (represented as a partition of the set of registers), an indicator of whether each partition holds a nonce, and, for each non-nonce value, a counter for the number of pseudo-reads of that value. 
	Since these counters only need to be maintained up to $k$, the complexity follows.
\end{proof}
We now turn our attention to the task of finding supporting runs. The following lemma establishes that the capacity to generate a word is equivalent to the capacity to generate any word of the same width. Thus, the existence of a supporting run for a word $w$ implies the existence of a supporting run for any other word $w'$ where $\wdwidth(w) = \wdwidth(w')$.

\begin{lemma}\label{lem:dec:kbounded}
	If there is a $k$-writing run of $\pds$, then for any $w \in \domain^*$ with $\wdwidth(w) \leq k$, there is a run that writes $w$.
\end{lemma}
\begin{proof}(sketch)
	The core of the proof is a constructive argument demonstrating that a run capable of writing a single value $k$ times can be transformed into a run writing any word of width at most $k$ through cloning and relabelling.
	
 Consider a $k$-writing run $\run$ that writes a value $a$ exactly $k$ times. We decompose $\run$ into segments $\run_0, \run_1, \dots, \run_k$, where each $\run_i$ (for $i \in [1, k]$) begins with a write to $\greg$ and concludes by leaving $\greg$ with the value $a$. Because each segment $\run_i$ starts by overwriting the shared register, its internal moves are independent of the register's prior state.
 Let $w = a_1, \dots, a_m$ be the target word. For each value $a_j$ in $w$, we instantiate a copy of the system's threads. Since $a$ and $a_j$ are nonces, we can relabel the original run to write $a_j$ instead of $a$ without violating any transition rules.
 By interleaving these relabelled segments according to the sequence required by $w$, we construct a unified run. 
\end{proof}
As a consequence of Corollary \ref{cor:dec:widthbound} and Lemma \ref{lem:dec:kbounded} we have:
\begin{corollary}\label{cor:dec:acyclic}
Let $k$ be largest value below $\MAX$ such there is a $k$-writing run. There is an acyclic witnessing run iff there is a witnessing pseudo $1$-run whose width is bounded by $k$.
\end{corollary}

The rest of this subsection is devoted to showing how to determine the largest $k \leq \MAX$ for which there is a $k$-writing run. This needs the notion of a replicating (pseudo) run. We say that a value $v$ is \emph{non-local} for thread $i$ in a run/pseudo-run, if $v$ is not  a nonce generated by that thread.
A run/pseudo-run  is \emph{replicating} if there is a thread that writes a non-local value to the shared register $\greg$. Equivalently, a run/pseudo run is {replicating} if there is a thread in it that writes to the shared register, a value that was not generated by itself.  Otherwise it is a non-replicating run/pseudo-run.  An example of a replicating run is shown in the Figure \ref{main:fig:replicating}.
The following two technical lemma allows us to find a $k$-writing run inductively, a detailed proof of this characterization can be found in the appendix. %
\begin{figure*}
	\includegraphics[trim=0.26cm 0.25cm 0.25cm 0cm , scale=0.39]{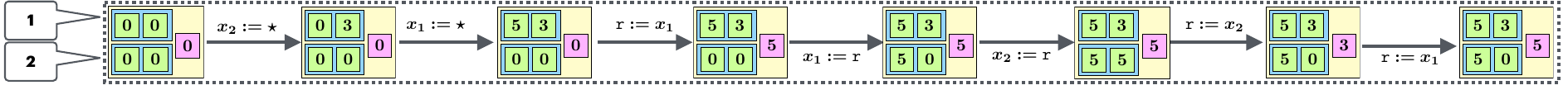}
	\caption{A replicating $2$-run, the thread $2$ write the non-local value $5$ to the shared register}\label{main:fig:replicating}
		\vspace{-2ex} %
\end{figure*}

\begin{lemma}\label{lem:dec:RepToWrite}
	If there is a replicating run, then there is $k$-writing run for any $k$.
\end{lemma}

\begin{lemma}\label{lem:dec:kwriting}  For any $k$, there is a  $k$-writing  run if and only if there is a $j < k$ such that there is a $j$-writing non-replicating run $\run$  and an acyclic pseudo-$1$-run $\irun$ with   $\wdwidth(\supp(\irun)) \leq j$ which is either replicating or $k$-writing.
\end{lemma}

Using Lemma \ref{lem:dec:kwriting}, we can construct the desired algorithm. We first note that for any $j < k \in \Naturals$, verifying the existence of an acyclic $k$-writing non-replicating pseudo-1-run or an acyclic replicating pseudo-1-run (each with support width bounded by $j$) is achievable in polynomial space relative to the system size and $\log(j)$.

\begin{lemma}\label{lem:sglobal:pspace:kwrite}
	There is a {\sc PSpace} algorithm to check the following:
	\begin{itemize}
		\item For a given $0 \leq j < k$, if there is a $k$-writing acyclic pseudo-1-run $\rho$ such that $\wdwidth(\supp(\rho)) \leq j$.
		\item For a given $k$, if there is a replicating acyclic pseudo-1-run $\rho$ such that $\wdwidth(\supp(\rho)) \leq j$.
	\end{itemize}
\end{lemma} 

This allows an iterative search for the maximum $k \leq \MAX$. Since a width-0 pseudo-1-run cannot replicate, a $1$-writing run exists iff there is a $1$-writing non-replicating acyclic pseudo-1-run requiring no support. If no $1$-writing run exists, no higher $k$-writing runs exist. Otherwise, we iterate: a $2$-writing run exists iff there is an acyclic pseudo-1-run with support width $\leq 1$ that is either replicating or $2$-writing. This process continues up to $\MAX$; as the counter and checks are polynomial in space, the result holds. Theorem \ref{thm:dec:acyclic} follows from Corollary \ref{cor:dec:acyclic} and Lemma \ref{lem:dec:maxkwriting}.

\begin{lemma}\label{lem:dec:maxkwriting}
	There is an algorithm that determines the largest $k \leq \MAX$ for which there is a $k$-writing run and runs in space polynomial in $\norm{\pds}$.
\end{lemma}

\subsection{Searching for cyclic witnessing runs}
The objective of this section is to prove that determining the existence of a cyclic witnessing run is decidable in {\sc Pspace}.

\begin{theorem}\label{thm:dec:cyclic}
	The existence of a cyclic witnessing run can be solved in {\sc Pspace}.
\end{theorem}
\noindent
\textbf{Overview of the proof:\ } 
Our strategy for detecting a cyclic witness involves identifying a cyclic pseudo-1-run and pairing it with a suitable support in the form of a replicating run. To formalize this, we introduce the following definitions.

We define a replicating pseudo-1-run $\irun$ as \emph{\kreadrep} if it can be decomposed as $\irun = \irun_1 \movesto{\tau} \irun_2$, where $\tau$ is a replicating transition of some value $\val$, and $\irun_1$ contains no more than $k$ pseudo-assignments of $\val$. Intuitively, this means a thread copies a value back to the shared register while having requested that  value at most $k$ times. 

Note that being \kreadrep is distinct from requiring the run to read the value at most $k$ times, as a run may perform multiple reads of the same value from the shared register between pseudo-assignments. For example, the replicating run in Figure~\ref{main:fig:replicating} is
\kreadrepPar{1}: the shared register receives the value $5$ via a single
pseudo-assignment before the replicating write, even though the
replicating thread performs two separate local reads of that same
resident value in the meantime -- repeated reads of an unchanged value
never count against the  bound.
 A  run is considered \emph{\kreadrep} if its projection onto at least one thread is \kreadrep.

Recall that a pseudo-1-run is {cyclic} if $\nonce(\irun) \cap \asupp(\irun) \neq \emptyset$, meaning the pseudo-run receives its own generated nonce via a pseudo-assignment. We say a cyclic pseudo-1-run $\irun$ is \emph{\kwritepr} if, for every value $\val \in \nonce(\irun) \cap \asupp(\irun)$, the value $\val$ is written to the shared register at least $k$ times before its first pseudo-assignment. 

Formally, let $\sset$ be the set of values in $\supp(\irun)$. For each $\val \in \sset \cap \nonce(\irun)$, $\irun$ admits a decomposition $\irun = \irun^1_\val \irun^2_\val$ such that $\irun^1_\val$ contains at least $k$ writes of $\val$, and $\val$ does not appear in  the support of $\irun^1_\val$. A run (or pseudo-run) is \kwritepr if its projection onto every thread is \kwritepr. For instance, the pseudo-1-run in Figure \ref{main:fig:psuedorun} is cyclic and \kwriteprPar{2}.

We now outline the overall proof strategy. Consider any cyclic witnessing run $\pi$. By projecting $\pi$ onto the witnessing thread, we obtain a cyclic pseudo-1-run $\irun$. Let $k$ be the maximum integer such that $\irun$ is \kwritepr. It follows that $\pi$ is a run that is \kreadrep; since $\irun$ is cyclic and $k$ is maximum, there exists a nonce $\val$ written exactly $k$ times before its pseudo-assignment. %

We demonstrate that the converse also holds: a witnessing pseudo-1-run that is \kwritepr can be combined with a \kreadrep run to produce a cyclic witnessing run. We further show that if a replicating run exists, there must be one that is \kreadrep for some $k \leq \MAX$. Consequently, to find a cyclic witnessing run, we iterate through each $k \leq \MAX$ and check for the simultaneous existence of a \kreadrep run and a \kwritepr pseudo-1-run. We prove that both of these conditions can be verified using space at most polynomial in the size of the system. We will formalize the details below.

\noindent
\textbf{Characterization and Complexity of Cyclic Witnessing Runs: \ }
We characterize the existence of cyclic witnessing runs through the interaction of \kwritepr pseudo-runs and \kreadrep replicating runs.

\begin{lemma}\label{lem:dec:cyclicchar}
	A cyclic witnessing run exists if and only if there is a witnessing cyclic pseudo-1-run that is \kwritepr and a replicating run that is \kreadrep for some $k \in \mathbb{N}$.
\end{lemma}

By Lemma \ref{lem:dec:oneglobal:maxcyclic}, we show that we can bound our search for $k$ within the range $[0, \MAX]$.

\begin{lemma}\label{lem:dec:oneglobal:maxcyclic}
	If a replicating run exists, then there is one that is \kreadrep for some $k \leq \MAX$.
\end{lemma}

The property \kwritepr is downward closed, while \kreadrep is upward closed. Consequently, we can determine the existence of a cyclic witness by iterating through $k \leq \MAX$ and checking if both conditions hold for some $k$. Towards this, in the following lemma we show that checking for the existence of a \kwritepr witnessing pseudo-1-run can be done in {\sc PSpace}.

\begin{lemma}\label{lem:dec:cyclic:tracking}
	Given $k$, the existence of a \kwritepr witnessing pseudo-1-run can be checked in space polynomial in $\norm{\pds}$ and $\log(k)$.
\end{lemma}

To check for a $k \leq \MAX$ such that there is a  \kreadrep replicating run, we use the following inductive criteria.  Moreover in the appendix, we also show that these criteria can be checked in {\sc PSpace}.

\begin{lemma}\label{lem:sglobal:replicatingrun}
	There is a   run that is \kreadrep iff one of the following holds:
	\begin{enumerate}
		\item There is an acyclic \kreadrep pseudo-1-run and a supporting run.
		\item There is a \kreadrepPar{j} run for $j > k$ and a cyclic pseudo-1-run $\rho$ that is \kwriteprPar{j} and \kreadrep.
	\end{enumerate}
\end{lemma}

\noindent
\textbf{Algorithm:\ } Deciding this problem for $k = \NMAX$ is, by the above, equivalent to deciding the existence of
any  replicating run.  Starting at $k = \MAX-1$, we iteratively decrement $k$ and test conditions (1) and (2). Condition (1) reduces to finite-state reachability (since the existence of a replicating run implies the existence of supporting runs). Condition (2) involves searching for specific pseudo-runs within the bounded $j > k$ range. Both are solvable in {\sc Pspace}, the proof can be found in appendix \ref{apx:sec:oneglobal:cyclic}. This completes the proof of  Theorem \ref{thm:dec:cyclic}.

	\section{Single Local Register}\label{sec:Hardness}

\noindent
\textbf{Decidability:\ } Towards proving decidability, we reduce the \PCRFR\ problem to coverability in a $\nu$-net, a problem known to be non-primitive recursive \cite{ROSAVELARDO20114439}. 
Informally, a $\nu$-net is a Petri net where each token is tagged with a value drawn from an infinite domain $\domain$. It operates over two disjoint sets of variables: $\var$ for manipulating existing tags, and $\newvar$ for introducing strictly new tags.
 A $\nu$-net is formally defined as a tuple $N = (\places,\transition,\infun,\outfun)$, where $\places$ is a finite set of places, $\transition$ is a finite set of transitions, and the flow functions are $\infun: \places \times \transition \to \var^\oplus$ and $\outfun: \places \times \transition \to (\var \cup \newvar)^\oplus$. A configuration (or marking) is a mapping $\marking: \places \to \domain^\oplus$. For any two markings $\marking_1, \marking_2$, we write $\marking_1 \geq \marking_2$ if $\marking_1(p) \geq \marking_2(p)$ (multi set ordering) for all $p \in \places$. 

Let $\sigma: \var \cup \newvar \to \domain$ be an injective function. We say $\sigma$ is consistent with a subset $\xdomain \subseteq \domain$ if $\sigma(x) \in \xdomain$ for all $x \in \var$, and $\sigma(y) \notin \xdomain$ for all $y \in \newvar$. This mapping naturally extends to multisets. For any $t \in \transition$, a transition $\marking_1 \movesto{t} \marking_2$ (or simply $\marking_1 \movesto{} \marking_2$) occurs if: 1) there exists an injective map $\sigma: \var \cup \newvar \to \domain$ consistent with the data values present in $\marking_1$ (i.e., $\sigma$ is $\xdomain$-consistent where $\xdomain = \bigcup_{p \in \places} \Sigma(\marking_1(p))$); 2) for each $p \in \places$, $\sigma(\infun(p,t)) \leq \marking_1(p)$; and 3) $\marking_2(p) = \marking_1(p) - \sigma(\infun(p,t)) + \sigma(\outfun(p,t))$. A marking $\marking'$ is coverable from $\marking$ if there exists a sequence $\marking \movesto{}^* \marking''$ such that $\marking'' \geq \marking'$.

Towards the reduction, let the given \PDS\ be $\pds$ with global registers $\globalregs = \{ \greg_1,\dots,\greg_k\}$, local register $\localregisters = \{\lreg\}$, and target label $\ell \in \Labels$. We construct a $\nu$-net $N = (\places,\transition,\infun,\outfun)$ alongside initial and target markings $\marking_{\tt in}, \marking_{\ell}$ such that $\marking_{\ell}$ is coverable from $\marking_{\tt in}$ if and only if $\ell$ satisfies \PCRFR\ in $\pds$. We outline the core reduction here, deferring the comprehensive details and also an alternative direct proof using the \WSTS\ framework to the appendix \ref{sec:onelocal}. 

In the constructed $\nu$-net, $\places$ includes one place for each global register, one for each constant, and one for each label in $\Labels$. In the initial marking $\marking_{\tt in}$, exactly one token resides in the place for the constant $0$, and one token resides in each global register place; all these tokens share an identical tag. Exactly one token is placed in each of the other constant places, carrying distinct tags. Intuitively, this models the initialization of all global registers to $0$, while constants are treated as fixed global registers with unique values. Our construction strictly maintains a single token in the places for global registers and constants, and ensures the tags within the constant places are never altered. 

To allow for unbounded threads, we introduce a transition $t$ that introduces threads starting at the initial label with the local register initialized to $0$. We define $\infun(0,t) = x$ and $\outfun(0,t) = \outfun(\initlab,t) = x$ (omitting zero-assignments). This ensures a new token is introduced in $\initlab$ tagged with the data value of $0$. Each token in the places of $\Labels$ represents a thread and its tag represents the value in its local register. Due to space limitations, we illustrate only the simulation of specific instructions:

For an instruction $\ell: \greg_1 = \greg_2$, we move exactly one token from $\ell$ to $\nextof{\ell}$ and replace the tag in $\greg_1$ with the tag from $\greg_2$. We define a transition $t_1$ with $\infun(\ell,t_1) = x$, $\infun(\greg_1,t_1) = y$, and $\infun(\greg_2,t_1) = z$. The output is $\outfun(\nextof{\ell},t_1) = x$ and $\outfun(\greg_1,t_1) = \outfun(\greg_2,t_1) = z$. This correctly advances the thread and copies the value.

For a non-deterministic assignment $\ell: \lreg = *$, we define a transition $t_2$ with $\infun(\ell,t_2) = x$ and $\outfun(\nextof{\ell},t_2) = y$, for some $y \in \newvar$. This guarantees that the new token arriving at $\nextof{\ell}$ carries a strictly new tag. 

	The complete construction presented in the appendix~\ref{apx:sec:sl:nunet} demonstrates the equivalence between the coverability in the $\nu$-net and the \PCRFR\ problem. With this, we obtain:

\begin{theorem}
	 \PCRFR\   with one  local register is decidable in   $\mathbf{F}_{\omega\cdot 2}$, by \cite{LazicSchmitz16}.	
\end{theorem}
\noindent
\textbf{Hardness:\ }
We next demonstrate that \PCRFR\ reachability in a \PDS\ with a single local register is at least as hard as reachability in lossy counter systems, a problem established as non-primitive recursive ($\mathtt{F}_{\omega}$-complete) \cite{Schnoebelen10,Schmitz19}. A lossy counter system extends a finite-state machine with counters that can be incremented, decremented, or tested for zero; importantly, it is "lossy", permitting counter values to decrease non-deterministically.

Formally, an \emph{$n$-lossy counter system} is defined by the tuple $\lossy(n) = \lossytuples$, where $\States_\lossy$ is a finite set of states, $\initstate^\lossy \in \States_\lossy$ is the initial state, $\finalStates_\lossy \subseteq \States_\lossy$ is the set of final states, and $\Trans_\lossy \subseteq \States_\lossy \times \counterop \times \States_\lossy$ is a finite set of transitions with operations $\counterop = \{ \inc_i, \dec_i, \zero_i \mid i \in [1,n] \}$.

The system's configurations are elements of $\States_\lossy \times \Naturals^n$. A transition between configurations $c = \tuple{q, d_1, \dots, d_n}$ and $c' = \tuple{q', d'_1, \dots, d'_n}$, denoted $c \movesto{} c'$, is valid if and only if either there is a real transition in $\Trans_\lossy$ and the move respects the corresponding counter operation, or there is a lossy transition where the state remains unchanged ($q = q'$) and the counter values non-deterministically decrease ($d'_j \leq d_j$ for all $j \in [1,n]$).

Given a final state $\fin \in \finalStates_\lossy$, the \emph{reachability problem} asks whether there exists a computation path $\tuple{\initstate^\lossy, 0, \dots, 0} \movesto{}^* \tuple{\fin, d_1, \dots, d_n}$ for some $d_1, \dots, d_n \in \Naturals$. We proceed to show the reduction of this problem to the \PCRFR\ problem with a single local register.
The core idea behind simulating a lossy counter machine is as follows: an increment must be simulated strictly without duplication, a decrement may be simulated with duplication, and a zero-test is simulated by logically dropping the counter value to $0$.

Our reduction relies on a "leader" thread---elected using two global variables and one local variable---to simulate the finite-state machine. This leader is assisted by helper threads, each identified by a unique nonce. Each counter is simulated using a dedicated pair of global registers; we describe the simulation for a single counter utilizing registers $r_1$ and $r_2$.

The simulation operates in phases, where a phase begins at a zero-test and concludes at the next zero-test. The current phase is uniquely identified by a value stored in $r_1$. Within a phase, operations are simulated as follows:

\noindent
\textbf{Simulating an increment}:
The leader populates $r_2$ with an $\inc$ instruction and waits. Helpers observing this instruction respond by writing their unique IDs to $r_2$. The leader selects exactly one helper via a protocol similar to the selection protocol described in section \ref{sec:undec}.  The selected helper  records the current phase identifier (from $r_1$) in its local register. 

\noindent
\textbf{Simulating a decrement}:
The leader writes a $\dec$ instruction to $r_2$. Any previously selected helper observing this instruction responds by writing its stored phase number to $r_2$ and subsequently halts. Upon reading the current phase number (matching $r_1$) in $r_2$, the leader confirms that at least one helper from the current phase has dropped out. Because multiple helpers may respond and halt, this naturally models a lossy decrement.

\noindent
\textbf{Simulating a zero-test}:
The leader simply replaces the value in $r_1$ with a new nonce. Consequently, all currently active helpers possess an outdated phase identifier in their local registers. This prevents them from participating in future decrements, effectively resetting the counter to zero.

This construction generalizes to an arbitrary number of counters by allocating an additional pair of global registers for each counter. Detailed construction and correctness arguments are deferred to the appendix \ref{apx:sec:hardness}.

\begin{theorem}\label{thm:dec:sl:hardness}
	 \PCRFR\  for  single local register is  $\mathbf{F}_{\omega}$-hard, by  \cite{Schmitz19}.	
\end{theorem}

	\section{PReach Problem  is {\sc PSpace complete}} \label{sec:PRNoFresh}

Interestingly, the parameterized reachability problem is in {\sc Pspace}, with no restrictions on the number of local/global registers,  if dispense with the freshness assumption.  

\begin{theorem}\label{thm:nofresh:dec1}
	The PReach problem is {\sc Pspace-complete}.
\end{theorem}

Without freshness, any value that is written to a global register $\greg$ can be written to it any number of times, and  at will. This is because, in this setting, we can duplicate a thread and let both of them run in a coupled manner (scheduling their corresponding instructions adjacently, and ensuring they generate the same values for non-deterministic assignments). This can be exploited to search for supporting runs in {\sc Pspace}. The hardness is a consequence of a result in \cite{DemriL09}. The details are in the appendix \ref{sec:PRNoFresh1}.

	\bibliographystyle{splncs04}
	\bibliography{Tacas}

	\section{Appendix for Undecidability}\label{sec:undec1}

\begin{figure}[htbp]
	\centering
	\begin{minipage}{0.47\textwidth}
		\begin{algorithm}[H]
			$\ell_{!a}^{q,{q_i}}:\; $\eIf{$\lreg_2 = \bot$}{
				$\goto\ \ell_{\tt nlist}^{q_i}$
			}{
				$\goto\ \ell_{\tt applist}^{q_i}$
			}
			\caption{Simulating Write}\label{main:alg:ins}
		\end{algorithm}
	\end{minipage}
	\hfill %
	\begin{minipage}{0.49\textwidth}
		\begin{algorithm}[H]
			$\ell_{\tt nlist}^{q_i}: \greg_2:=\text{ins(a)}$;
			$\lreg_2 := \greg_2$\;
			$\assert(\lreg_2 \neq \text{ins(a)})$\;
			$\greg_2:=\bot$;
			$\greg_2:=\lreg_2$\;
			$\assert(\greg_2 = \bot(a))$;
			$\goto\ \ell_{q_i}$\; 
			\caption{Insert to empty list}\label{main:alg:insempty}
		\end{algorithm}
	\end{minipage}
\end{figure}

We start by describing the protocol in more details, providing the pseudo code alongside. Recall that we are interested in simulating a system equipped with a perfect FIFO channel. For this, we need the unique leader to be able to simulate the enqueue and dequeue operations.  We describe each of these below.

\textbf{Enqueue protocol} The leader code  to simulate the operation $\tuple{q,!a,q_i}$ is shown in Algorithm \ref{main:alg:ins}.
In case the list is empty, the leader executes a protocol that inducts a helper into the chain, this is shown in Algorithm \ref{main:alg:insempty}. For this, the leader writes a command $\ins(a)$ into the register and expects one of the free helper to respond with its id. There may be many helper willing to join the list, each of them write their id into the register. The leader chooses one of them, stores it into its register  $\lreg_2$ and writes it back to $\greg_2$. It then awaits an acknowledgement from the helper in form of a $\bot(a)$ in $\greg_2 $.

\begin{wrapfigure}{l}{0.39\textwidth}
	\begin{minipage}{\linewidth}
		\begin{algorithm}[H]
			$\ell_{\tt applist}^{q_i}: \greg_2:=\prop(a)$\;
			$\greg_2:=\lreg_2$\;
			$\assert(\greg_2 \neq \lreg_2)$\;
			$\assert(\greg_2 = \lreg_2)$\;
			$\greg_2:=\bot(a)$\;
			$\goto\ \ell_{q_i}$\; 
			\caption{Append to  list}\label{main:alg:inspropogate}
		\end{algorithm}
	\end{minipage}
	\vspace{-10pt} 
\end{wrapfigure}

In case the queue is not empty, it writes the command $\prop(a)$ (indicating propagate) followed by the id of its neighbour into the global register and awaits an acknowledgement.  The protocol in this case is show in Algorithm \ref{main:alg:inspropogate}.

Each helper thread  starts by participating in an insert protocol. Once it is successful (i.e. part of a chain), it participates in the append protocol. Both these codes are shown in Algorithm \ref{main:alg:hlp:ins} and \ref{main:alg:hlp:append}. The helper thread to participate in the insert protocol, writes its id onto the global register and waits for it to be selected by the leader, it is selected if it can read its own id in the global register. The protocol ensures that at most  one of the many helpers that participated can read back its own id and successfully complete the protocol. The  helper which is selected,  acknowledges by writing a $\bot(a)$ to the global register.  Note that here, an thread can read an insert command and react to it after several other commands have been executed. While it is okay for such a helper to be inducted into the chain, we need to ensure that the helper is storing the right value. That is, it is indeed responding to $\ins(a)$. This is ensured by acknowledging with  $\bot(a)$ that includes the value that is being stored by the inducted thread. 
The helper thread remembers the symbol it has responded to by jumping to an appropriate location of the code.

\begin{figure}[htbp]
	\centering
	\begin{minipage}{0.4\textwidth}
		\begin{algorithm}[H]
			$\hlab_{\tt ins(a)}:\assert(\greg_2 = \text{ins(a)})$\;
			$\lreg_1 := * $\; 
			$\greg_2:=\lreg_1$\;
			$\assert(\greg_2 \neq \lreg_1)$\;
			$\assert( \greg_2 = \lreg_1)$\;
			$\lreg_2:=\bot$\;\ $\greg_2:=\bot(a)$\;
			$\textsf{goto } \hlab_a$\;
			\caption{Helper  Insert$(a)$}\label{main:alg:hlp:ins}
		\end{algorithm}
	\end{minipage}
	\begin{minipage}{0.58\textwidth}
		\begin{algorithm}[H]
			$\hlab_{\tt prop(a)}:\assert(\greg_2 = \prop(a), \greg_2 = \lreg_1)$\;
			\eIf{$\lreg_2 \neq \bot$}{
				$\greg_2 := \lreg_2$;
				$\assert(\greg_2 \neq \lreg_2, \greg_2 = \lreg_2)$\;
				$\greg_2:=\lreg_1$;\ $\goto\ \ell^{\tt hlp}_b$\;
			} {
				$\greg_2:=\text{ins(a)}$;
				$\lreg_2 := \greg_2$\;
				$\assert(\lreg_2 \neq \text{ins(a)})$;
				$\greg_2:=\bot$;
				$\greg_2:=\lreg_2$\;
				$\assert(\greg_2 = \bot(a))$;\ $\greg_2:=\lreg_1$;\ $\goto\ \ell^{\tt h}_b$\;
			}
			\caption{Helper  Append}\label{main:alg:hlp:append}
		\end{algorithm}
	\end{minipage}
\end{figure}

In the case of  append protocol, the behaviour of the helpers that are already included in the chain, depends on whether it has a neighbour or not, it is described in Algorithm \ref{main:alg:hlp:append}. The protocol starts when the helper thread reads the command $\prop(a)$. In  case the helper is not the last in the chain ($\lreg_2 \neq \bot$), on reading its own id, writes to the global register the id of its neighbour. It then waits for an acknowledgement in the form of its neighbour's id. Upon reading it,  writes its own id. In case the helper is last in the chain, it executes the insertion protocol that we discussed earlier to add a new helper into the chain, it also writes its own id to complete the protocol. 

\begin{wrapfigure}{l}{0.47\textwidth}
	
	\begin{minipage}{\linewidth}
		\begin{algorithm}[H]
			$\ell_{?a}^{q,{q_i}}: \greg_2:=\text{del(a)}$\;
			$\greg_2:=\lreg_2;\ \assert(\greg_2 \neq \lreg_2)$;
			$\lreg_2 := \greg_2$\;
			\eIf{$\lreg_2 \neq \textsf{last}$}{
				$\assert(\greg_2 \neq \lreg_2, \greg_2 = \lreg_2)$\;
			}{
				$\lreg_2 := \bot$\;
			}
			\caption{Simulating $\tuple{q,?a,q_i}$}\label{main:alg:del}
		\end{algorithm}
	\end{minipage}
	
\end{wrapfigure}
\textbf{Dequeue protocol}
In the dequeue protocol the leader interacts with its neighbour (dequeue is possible only if the chain is not empty). The neighbour (i.e. the head of the list) is responsible to delete itself and provide the leader with the id of its neighbour (provided there is one). 
The leader indicates its intent to delete by writing the value $\del(a)$  to the global register $\greg_2$, followed by the id of its neighbour ( Algorithm \ref{main:alg:del}). 
Each helper makes a guess as to whether it is the leader's neighbour (head) or not. It executes either $\textsf{DelSelf}(a)$ and $\textsf{DelHelp}$ based on the guess.  
If any helper guesses wrongly the resulting dequeue protocol will not complete and the simulation deadlocks.

A helper other than head executes the code in Algorithm \ref{main:alg:hlp:delhlp} (delHelp). On reading $\del(a)$, it checks if $\lreg_2 = \bot$ to learn if it is the tail
of the queue and in that case writes its \id on $\greg_2$. All other helpers (other than the tail and the head) wait to see their neighbour's id  on $\greg_2$. On seeing it, they replace it with
their own id. Thus in the desired execution, so far, the value of 
$\greg_2$ would start with $\del(a)$ and then be followed one by one 
by the \id  of the members of the queue starting with the tail,
listed in order all the way till the neighbour of the head.
(Of course, no \id would have been written if the queue had only 
one element).  Observe that our protocol forces all helper threads to participate in the deletion and to complete their role for the deletion to be successful. 
\noindent

\begin{figure}[htbp]
	\centering
	\begin{minipage}{0.48\textwidth}
		\begin{algorithm}[H]
			$ \ell^b_\delprop: \assert(\greg_2 = \del(a))$\;
			\eIf{$\lreg_2 = \bot$}{
				$\greg_2 := \lreg_1$;\ $\goto\ \ell^{\tt hlp}_b$\;
			}{
				$ \assert(\greg_2 = \lreg_2)$\;
				$\greg_2 := \lreg_1$;\ $\goto\ \ell^{\tt hlp}_b$\;
			}
			\caption{DelHelp}\label{main:alg:hlp:delhlp}
		\end{algorithm}
	\end{minipage}
	\hfill
	\begin{minipage}{0.49\textwidth}
		\begin{algorithm}[H]
			$\ell^a_{\del(a)}: \assert(\greg_2 = \del(a), \greg_2 = \lreg_1)$\;
			\eIf{$\lreg_2 = \bot$}{
				$\greg_2 := \textsf{last}$;\ $\goto\ \ell_{\tt halt}$;
			}{
				$\assert(\greg_2 = \lreg_2)$\;
				$\greg_2 := \bot$;\ $\greg_2 := \lreg_2$;
				$\goto\ \ell_{\tt halt}$;
			}
			\caption{DelSelf$(a)$}\label{main:alg:hlp:delself}
		\end{algorithm}
	\end{minipage}
\end{figure}

Now we turn our attention back to  the leader and the head.
First we deal with the case when the queue has only one element.
After reading $\del(a)$, and its own id from $\greg_2$ (we will
see why this read succeeds soon)  simply deletes itself if it is
the last element (i.e. its $\lreg_2 = \bot$) and indicates this
to the leader by writing a special value \textsf{last}.  
The leader,
after writing $\del(a)$, writes its neighbour's id on $\greg_2$ (thus,
the leader's neighbour can succeed in reading its own id from $\greg_2$
after reading $\del(a)$). The leader then reads $\greg_2$ and on seeing \textsf{last} simply sets its $\lreg_2$ to $\bot$ and completes the protocol.

If there is only one element in the queue this is the only way
the protocol can execute to completion.
A subtle situation is one where the queue has more than one element and the tail element guesses wrongly that it is the head. In this case, it will not be able to write \textsf{last} to mislead the leader  since it will not succeed in reading its own id as required. Similarly, it is possible that the only helper in the chain can 
decide to execute delHelp procedure. In this case, the protocol will deadlock as the helper will not write an id that is not its own.

In case the queue has more than one element, once again the only thread that can
read $\del(a)$, followed by  its own id is the head. Thus the head can
identify itself unambiguously.  However, to delete itself it must convey its neighbour's id to the leader. It cannot simply 
write it to $\greg_2$, we explain why in the next paragraph.  However, if it waits to read it's neighbour's id from $\greg_2$ then it can be sure that the neighbour has completed its
execution. It can then then convey this id to the leader again.
That is, the head writes a $\bot$ to $\greg_2$ followed
by its neighbour's identity (from $\lreg_2$). The only id that the
leader can read twice interspersed by a $\bot$ is the identity of the
neighbour of the head and this allows the leader to set its neighbour
correctly after the deletion.

While we have described the simulation of the operations, we now describe the main algorithms corresponding to the leader and the helper. The program starts as in Algorithm \ref{main:fig:LE}, there is a unique thread (leader) that reaches $\ell_{q_{in}}$. More over any instance can either choose to participate in the election or non-deterministically jump to the label $\ell_{\tt helper}$.  The code for the leader thread is obtained from a code to code translation. The main leader code has as its labels $\ell_q$ for each state $q \in \States$ of the given PCS. For any state $q \in \States$,  let $\tau_1 = \tuple{q,\op_1,q_1}, \dots \tau_k = \tuple{q,\op_k,q_k}$ be the transitions starting from it.  Then, we have  in the leader code the following:
$$
		\ell_q: \text{if }(\star) \big\{ \text{ goto } \ell^{\op_1}_{q_1}\big\} \cdots  \big\{ \text{ goto } \ell^{\op_k}_{q_k}\big\}
$$

That is a transition is picked non-deterministically and executed. Similarly the code of the helper from the state $\ell^h_a$, non-deterministically guesses the next operation and executes the same.
$$
\ell^h_b: \text{if }(\star) \big\{ \text{ goto } \ell^{b}_{\tt dprop}\big\}  \big\{ \text{ goto } \ell^{b}_{{\tt del}(b)}\big\}   \big\{ \text{ goto } \ell^{b}_{\tt prop}\big\}
$$

This completes the description of the reduction from \PCS to the parameterized program $\pds$. We show the correctness of our construction below. We first  prove the following lemma that relates a fresh run in the parameterized program $\pds$ to a run in the $\pcs$. In particular, if there is a run in    $\pds$ of the form 	$\initpdsconf(n) \movesto{}^* \tuple{\labelmap, \lregmap,\gregmap}$ such that $\labelmap(1) \in \lab_q$ for some state $q \in \States$, then there is a run in $\pcs$ of the form $(q_{\tt in},\epsilon) \xrightarrow{}^* (q,w)$ for some $w \in \Sigma^*$.  Here we assume without loss of generality that the thread instance $1$ is the elected leader. 
\begin{lemma}\label{apx:undec:correct:if}
	Corresponding to a fresh $n$-run in $\pds$ of the form $\initpdsconf(n) \movesto{}^* \tuple{\labelmap, \lregmap,\gregmap}$ such that $\labelmap(1) \in \lab_q$ for some state $q \in \States$, then there is a run in $\pcs$ of the form $(q_{\tt in},\epsilon) \xrightarrow{}^* (q,w)$ for some $w \in \Sigma^*$
\end{lemma}

For purpose of the proof, let us fix a run 	$\pi = \initpdsconf(n) \movesto{}^* \tuple{\labelmap, \lregmap,\gregmap}$ such that $\labelmap(1) \in \lab_q$ for some state $q \in \States$. 
We have already established that there can be at-most a single leader, hence  $\pi$   can be decomposed into $\pi = \pi_1 \cdot \pi_2$, where the unique elected leader is in state $\ell_{q_{\tt in}}$ for the first time in $\pi_2$, equivalently, the leader election happened in $\pi_1$.  For sake of simplification of the proofs,  we will assume in the starting configuration of $\pi_2$ that there is a unique leader in the state $\ell_{q_{\tt in}}$ and every other thread is in the label $\ell_{\tt helper}$. We may assume this since any unsuccessful thread that participated in the leader election will eventually deadlock and can provide no assistance to executions in $\pi_2$.   So to prove Lemma \ref{apx:undec:correct:if}, it is enough to show the following lemma which relates and execution of $\pi_2$ with a run in $\pcs$.

\begin{lemma}
		Corresponding to a fresh $n$-run in $\pds$ of the form $ \tuple{\labelmap_1, \lregmap_1,\gregmap_1} \movesto{}^* \tuple{\labelmap_2, \lregmap_2,\gregmap_2}$, where for all $i \in [2,n]$, $\labelmap_1(i) = \ell_{\tt helper}$,  $\labelmap_1(i) = \ell_{q_{\tt in}}$ and $\labelmap_2(1) = \ell_q$, there is a run in $\pcs$ of the form $(q_{\tt in},\epsilon) \xrightarrow{}^* (q,w)$ for some $w \in \Sigma^*$.
\end{lemma}\label{apx:undec:lem:correctf}

For any $n$-configuration $ \pdsconf =  \tuple{\labelmap, \lregmap,\gregmap}  \in \Confof{\pds}(n)$, we  define the \emph{chain}, $\chain(\pdsconf)$ as as the sequence $(i_0=1) \to i_1 \to i_2 \dots \to i_k$, where each $i_1,i_2, \dots, i_k \in [2,n]$ such that the following holds. 
\begin{itemize}
	\item For all $j \in [1,k]$, $\lregmap(i_j)(x_1) = \lregmap(i_{j-1})(x_2)$. It is the chain obtained by traversing the  id's stored in the local register $x_2$, starting from the leader.
\end{itemize}

For any $n$-configuration $ \pdsconf =  \tuple{\labelmap, \lregmap,\gregmap}  \in \Confof{\pds}(n)$, we say it is \emph{good} if the leader is in $\ell_q$ for some $q \in \States$ and if $\chain(\pdsconf) = (i_0=1) \to i_1 \to i_2 \dots \to i_k$, then the following holds
\begin{itemize}
	\item $\labelmap(i_0) \in \{ \ell_q \mid q \in \States \}$.
	\item For all $j \in [2,k]$, $\labelmap(i_k) \in \{ \ell^h_a \mid a \in \Sigma \}$
\end{itemize}
We define the \emph{word-of-chain} of such a good configuration as $\word(\pdsconf) = \Sigma(i_1) \dots \Sigma(i_k)$, where for any $j \in [1,k]$, $\Sigma(i_j) = a$ if $\labelmap(i_j) = \ell^h_a$. We say a good configuration $ \pdsconf =  \tuple{\labelmap, \lregmap,\gregmap}  \in \Confof{\pds}(n)$ is \emph{very good} if for all helpers $j \in [2,n]$ that is not a part of  the chain, $\labelmap(j) \in \{\ell_{\tt helper}, \ell_{\tt halt}  \}$

Recall that the leader from the state $\ell_q$ for some state $q \in \States$,  can perform a move to the label $\ell^\op_{q'}$, where $\op \in  \{ !a,?a \mid a \in \Alphabet\}$. From these labels, the leader can write $\tins{a}$, $\tprop{a}$ or $\tdel a$ on to $\greg_2$. Given a good configuration $\pdsconf_1$, we say  $\pdsconf_1 \xleadsto{\op} \pdsconf_2$ if the following holds
\begin{itemize}
	\item  $\labelmap(\pdsconf_1)(1),  \labelmap(\pdsconf_2)(1) \in \{ \ell_q \mid q \in \States \}$. 
	\item  The leader writes to the global register $\greg_2$ the commands $\tins{a}$, $\tprop{a}$ or $\tdel a$, exactly once.
\end{itemize}

The following observation is immediate. Suppose $\pdsconf_1 \xleadsto{\op} \pdsconf_2$ and let $\pdsconf_1  \movesto{} \pdsconf'_1$. Suppose $\labelmap(\pdsconf'_1)(1) = \ell^{\cmd}_{q'}$ and $\labelmap(\pdsconf_1)(1) = \ell_{q}$
 then $\tuple{q,\cmd,q'} \in \Trans$. We show in the following lemma that the sequence of moves also simulates the channel operation.
 
 \begin{lemma}
 	Suppose $\pdsconf_1 \xleadsto{\op} \pdsconf_2$ and $\pdsconf_1$ is a very good configuration, the the following holds. 
 	\begin{itemize}
 		\item If $\op = \tins a$ then $\word(\chain(\pdsconf_1)) = \epsilon$ and  $\word(\chain(\pdsconf_2)) = a$
		\item If $\op = \tprop a$ then $\word(\chain(\pdsconf_2)) = \word(\chain(\pdsconf_1)) \cdot a$ 
		\item If $\op = \tdel a$ then $\word(\chain(\pdsconf_1)) = a \cdot \word(\chain(\pdsconf_2))$ 
 	\end{itemize}
 	
 	Moreover there is another run $\pdsconf_1 \xleadsto{\op} \pdsconf'_2$ where $ \pdsconf'_2$ is also very good.
 \end{lemma}
 \begin{proof}
 	\begin{itemize}
 	\item Suppose  $\op = \tins a$,  then leader thread executed the commands in Algorithm \ref{main:alg:insempty}.  This case is enabled only if the local register $\lreg_2 = \bot$ in the leader. With this we have $ \word(\chain(\pdsconf_1))=\epsilon$. To show the latter, we argue that exactly one helper thread is inducted into the chain. More over the inducted helper will be in the state $\ell^h_a$.  Firstly note that we start with a very good configuration $\pdsconf_1$. Hence all the helpers are in either $\ell_{\tt helper}$ or in $\ell_{\tt halt}$ state. Let $\sset$ be the set of helpers in $\pdsconf_2$ whose label is not in 
 	$\{\ell_{\tt helper},\ell_{\tt halt} \}$. Then 
 upon the leader executing a write of $\tins a$ into the register $\greg_2$, the subset of helpers in $\sset$ read this value and a subset $\sset' \subseteq \sset$ of them write their own id onto it (see line 1-3 in Algorithm \ref{main:alg:hlp:ins}). These subset of helper threads are then stalled in line $4-5$ awaiting to read back their own id. The only possibility is for the leader to execute code in line $3$ of Algorithm \ref{main:alg:insempty}, where one of the id of  helper is chosen, recorded into its local register $\lreg_2$ and written back to $\greg_2$.   Note that at this point, helpers from $\sset \setminus \sset'$ can write their id into $\greg_2$. But this would lead to a deadlock where no thread can make any move. Since we start with an assumption that $\pdsconf_1 \xleadsto{\op} \pdsconf'_2$, we may assume this is not the case.  At this point the only possible move is for one of the helper thread in $\sset'$ to read its own id from $\greg_2$ and proceed with execution of lines $6-8$ of Algorithm \ref{main:alg:hlp:ins}. Clearly then exactly one helper is inducted into the chain and currently is in the label $\ell^h_a$.  We next note that all the helpers other than the one inducted into the chain, from the set   $\sset$ are either in lines $1-4$ of Algorithm \ref{main:alg:hlp:ins}. Since none of the write of these threads are useful for the protocol to complete, we may assume another execution with only the inducted helper in the set $\sset$, this will provide us with $\pdsconf_1 \xleadsto{\op} \pdsconf'_2$ where $ \pdsconf'_2$ is also very good.

\item We omit the case when $\op = \tprop a$ as this is very similar to the previous case.

\item Suppose  $\op = \tdel a$,  then we show that there is a unique helper in the chain, that is adjacent to the leader that drops out of the chain. The leader to simulate the $\op$, writes $\tdel a$ onto the register $\greg_2$ and awaits reading a value other than the one it wrote to the memory. Note that for this to happen, the chain cannot be empty.  We now argue that none of the helpers in the chain that are not adjacent can assume the role of helper that is adjacent to the leader. That is, they cannot execute the code DelSelf(a) in Algorithm \ref{main:alg:hlp:delself}.  Note that for this, the helper has to be able to read its own id from the global register $\greg_2$, this is impossible.  Next we argue that the helper that is adjacent to the leader cannot assume to be a helper that will not be deleted. That is, it cannot execute code DelHelp in Algorithm \ref{main:alg:hlp:delhlp}. For this, we note that the leader on writing the the id of the helper adjacent to it, expects either the value $\textsf{last}$ or an id of a thread (other than the one it wrote) written twice. Both of these are not possible if the helper adjacent to the leader executes delHelp. Hence we may assume that each of the helper threads in the chain have correctly guessed its position in the chain and has read the $\tdel a$ from $\greg_2$. Assuming the length of the chain is greater than $1$,  the last helper in the chain writes its id onto $\greg_2$, followed by its predecessor thread and so on till the helper that is adjacent to the leader. Since this helper is assumed to be executing the code DelSelf(a), on reading its own id followed by the id of its sucessor, writes this id again. This ensures that the id is written twice for the leader to read. More over since no other threads other than the helper in chain participate and since the deleted helper jumps to $\ell_{\tt halt}$, the resulting configuration is very good. The case when the length of the chain is $1$, the unique helper thread writes $\textsf{last}$ onto $\greg_2$ as required. 
\end{itemize}
 \end{proof}

 \begin{proof}[Proof of Lemma \ref{apx:undec:lem:correctf}]
 	Consider any run of $\pds$ of the form $\pi = \pdsconf_1 \movesto{}^* \pdsconf_n$. Then the run can be decomposed as $\pdsconf_1\xleadsto{\op_1} \pdsconf_2 \xleadsto{\op_2} \dots \xleadsto{\op_{n-1}} \pdsconf_n$. We may assume that every $\pdsconf_i$ is very good. A helper stuck in the
 	middle of a protocol either joins the chain later, or never joins. In the
 	first case we postpone its steps; in the second we delete them.  Then consider the following sequence,
 	\[
 			\tuple{q_1,w_1}, \tuple{q_2,w_2}, \cdots \tuple{q_n,w_n} \quad \text{ where } \quad \labelmap(\pdsconf_i)(1) = \ell_{q_i} \text{ and } w_i = \word(\chain(\pdsconf_i))
 	\] 
 	
 	Then this forms a valid sequence of $\pcs$, more over any two adjacent configurations are related by a transition, this defines the required run.
\end{proof}

We next show the other direction of the proof.

\begin{lemma}\label{lem:correctb}
	Let $(q_{\tt in}, \epsilon) \movesto{}^* (q, w)$ be a run of $\pcs$ of length
	$n$. Then there is a fresh $(n+1)$-run of $\pds$ of the form
	$\langle \boldsymbol{\lambda}_1, \boldsymbol{l}_1, \boldsymbol{g}_1 \rangle
	\movesto{}^*
	\langle \boldsymbol{\lambda}_2, \boldsymbol{l}_2, \boldsymbol{g}_2 \rangle$
	such that $\boldsymbol{\lambda}_1(1) = \ell_{q_{\tt in}}$,
	$\boldsymbol{\lambda}_2(1) = \ell_q$, and
	$w = w(\chain(\langle \boldsymbol{\lambda}_2, \boldsymbol{l}_2,
	\boldsymbol{g}_2 \rangle))$.
\end{lemma}

\begin{proof}
	Let the run of the perfect channel system $\pcs$ be given by the sequence of transitions:
	\[
	\tuple{q_0,w_0} \movesto{\tau_1} \tuple{q_1,w_1} \movesto{\tau_2} \cdots \movesto{\tau_n} \tuple{q_n,w_n}
	\]
	where $q_0 = q_{\tt in}$ and $w_0 = \epsilon$. Note that this run is of length $n$. We will demonstrate the existence of a corresponding fresh $n+1$-run in the parameterized program $\pds$. The simulation will proceed in rounds; in each round $i$, we will simulate the execution of the single instruction $\tau_i$.
	
	In round $0$, we assume that a single unique thread successfully participates in the leader election protocol and is elected as the leader, reaching the state $\ell_{q_{\tt in}}$.
	
	At the beginning of round $i+1$, assuming we have successfully constructed a run of the form $\pdsconf_1 \movesto{}^* \pdsconf_2 \movesto{}^* \cdots \movesto{}^* \pdsconf_{i}$, we must show how to extend this run to simulate instruction $\tau_{i+1}$. We assume by the inductive hypothesis that at the end of round $i$, the following invariants are strictly maintained:
	\begin{enumerate}
		\item $\pdsconf_{i}$ is a \emph{very good configuration}.
		\item The leader thread is in the correct state and the channel contents match: $\labelmap(\pdsconf_{i})(1) = \ell_{q_i}$ and $\word(\chain(\pdsconf_i)) = w_i$.
		\item There is a sufficient pool of available threads: at least $n-i$ helpers are in the state $\ell_{\tt helper}$.
	\end{enumerate}
	
	We proceed by case analysis on the type of transition $\tau_{i+1}$:
	\begin{itemize}
		\item \textbf{Case 1: Simulating an Enqueue Operation ($\tau_{i+1} = \tuple{q_i,!a,q_{i+1}}$)} \\
		Our simulation allows the leader to simulate $\pdsconf_{i} \xleadsto{\tins a} \pdsconf_{i+1}$ if the channel is currently empty ($\word(\chain(\pdsconf_i)) = \epsilon$), or $\pdsconf_{i} \xleadsto{\tprop a} \pdsconf_{i+1}$ if the channel is non-empty. In both subcases, the simulation initiates from the label $\ell_{!a}^{q_i,q_{i+1}}$. 
		
		This simulation inevitably concludes in a \emph{very good configuration}. Because the enqueue operation requires inducting exactly one new helper into the active chain, the pool of available helpers decreases by one, leaving $n - (i + 1) = n - i - 1$ helpers in the state $\ell_{\tt helper}$. Furthermore, the leader progresses to the next state such that $\labelmap(\pdsconf_{i+1})(1) = \ell_{q_{i+1}}$, and the message $a$ is successfully appended to the chain, ensuring $\word(\chain(\pdsconf_{i+1})) = w_{i+1}$.
		
		\item \textbf{Case 2: Simulating a Dequeue Operation ($\tau_{i+1} = \tuple{q_i,?a,q_{i+1}}$)} \\
		In this scenario, our simulation dictates that the leader simulates $\pdsconf_{i} \xleadsto{\tdel a} \pdsconf_{i+1}$. Similar to the enqueue case, the simulation initiates from the corresponding label $\ell_{?a}^{q_i,q_{i+1}}$. 
		
		This process is guaranteed to end in a \emph{very good configuration}. During a dequeue operation, exactly one helper thread (the head of the chain) is removed and transitions to the terminal state $\ell_{\tt halt}$. Consequently, the number of active helpers in the chain decreases, but the pool of unassigned helpers in $\ell_{\tt helper}$ remains unchanged at $n-i$. Finally, the leader's state updates to $\labelmap(\pdsconf_{i+1})(1) = \ell_{q_{i+1}}$, and the first message $a$ is successfully removed from the simulated channel, maintaining the invariant $\word(\chain(\pdsconf_{i+1})) = w_{i+1}$, as required.
	\end{itemize}
	This completes the induction and the proof.
\end{proof}

	\section{Single Global Reachability  Acyclic Case}\label{apx:sec:oneglobal:acyclic}

In this appendix, we give the missing proofs of
Section~\ref{main:sec:acyclic}. Unlike the main text, we allow constants. We
first adapt the definitions to this setting.

\paragraph{Definitions in the presence of constants.}
In Section~\ref{sec:oneglobal}, we assumed that there are no constants other than $0$,
and that no run ever writes $0$. So every value written to $r$ was a nonce.
Here constants may be written. So we count only non-constant values. A run is
\emph{$k$-writing} if it writes some value of $\domain \setminus \constants$
at least $k$ times. The width of a support counts only its non-constant
values. A run or pseudo-run is \emph{replicating} if some thread writes to
$r$ a value of $\domain \setminus \constants$ that it did not generate
itself. A support may now contain constants, but only realizable ones
(defined below). Under the assumption of Section~\ref{sec:oneglobal}, these
definitions coincide with the ones given there.
As a first step, we bound the length of a run that reaches a given label
$\lambda$.

\begin{lemma}[See Lemma \ref{lem:dec:1pseudobound}]\label{apx:lem:dec:1pseudobound}
	If there is a  acyclic $1$-pseudo run that reaches a label $\lambda$,  then there is one whose
	length is bounded by $\MAX = |\Labels|.2^{(|\localregisters| \times(|\constants|+2))}$
\end{lemma}
\begin{proof}
	
	Our first observation is that certain relabelling of values preserve acyclic $1$-pseudo runs. Let $(N,S)$ be a partition of $\domain$. We say that an acyclic $1$-pseudo run $\irun$  is $(N,S)$-consistent if $\nonce(\irun) \subseteq N$ and $\asupp(\irun) \subseteq S$.  A bijection $\ffun$ from $\domain$ to $\domain$ is a $(N,S)$-relabelling if it preserves the constant, maps $N$ onto $N$ and $S$ onto $S$.  That is for any $\val \in \constants,\   \val = \ffun(\val)$ and for $\val \in \domain\setminus \constants$, $\val \in N \iff \ffun(\val) \in N$, .

	A function $f$ from $\domain$ to $\domain$ naturally induces a function on configurations and we denote this by $f$ as well.  Note that for any configuration $\pdsconf = \tuple{\labelmap, \lregmap,\gregmap}$ and for any register $x \in \localregisters $, $\lregmap(i)(x) \in N \iff \lregmap(i)(\ffun(x))\in N$, for any thread $i$ in it. Similarly for the global register $\greg \in \globalregs$,  $\greg(\pdsconf) \in N \iff \ffun( \greg(\pdsconf)) \in N$. The function $\ffun$ extends to transitions and runs as well.  
	With this, 
	it is easy to check that if $\irun$ is a $(N,S)$-consistent acyclic $1$-pseudo run and $f$ is a $(N,S)$-relabelling then $f(\irun)$ is a also an acyclic $1$-pseudo run that is $(N,S)$-consistent.
	
	Let $\irun = \iconf_1 \dots \iconf_m$ be an acyclic $1$-pseudo run where $\iconf_1$ is not necessarily the initial configuration.  Let $(N,S)$ be any partition of $\domain$ so that $\irun$ is $(N,S)$-consistent. This is always possible since $\irun$ is acyclic.
	
	We now show that if there is a  witnessing acyclic pseudo $1$-run, then there is one whose length is bounded by $\MAX$. Let $\irun = \iconf_1 \dots \iconf_m$ be a witnessing acyclic $1$-pseudo run, with $m > \MAX$. For any $i,j \in [1,m]$,  we say $\iconf_i \equiv \iconf_j$ if the following holds.
	
	\begin{itemize}
		\item They have the same labels, that is $\Labels(\iconf_i) = \Labels(\iconf_j)$
		\item For any $x,y \in \localregisters \cup \globalregs$, $x(\iconf_i) = y(\iconf_i)$ if and only if $x(\iconf_j) = y(\iconf_j)$.
		\item For any $x \in \localregisters \cup \globalregs$, $x(\iconf_i) \in N $ if and only if $x(\iconf_j) \in N $
		\item For any $x \in \localregisters \cup \globalregs$, $x(\iconf_i) = \val$ for some $\val \in \constants$ if and only if $x(\iconf_j) = \val$
	\end{itemize}

	Then, by pigeon-hole principle, we can find $i,  j \in [1,m]$ such that $\iconf_i \equiv \iconf_j$. Let $U_i  = \{ x(\iconf_i) \mid x \in \localregisters \cup \globalregs \}$  and $U_j =  \{ x(\iconf_j) \mid x \in \localregisters \cup \globalregs \}$ be the set of values that appear in these two configurations. Let $N_i = U_i \cap N$ and $N_j = U_j \cap N$, similarly we let $S_i = U_i \cap S$ and $S_j = U_j \cap S$. We show that there is a $(N,S)$-relabelling $\ffun$ that satisfies the following properties. 
	\begin{itemize}
		\item For any $\val \in \domain \setminus ((U_i \cup U_j)) \cup \constants$, $\ffun(\val) = \val$.
		\item The relabelling $\ffun$ induces  a bijection between $N_i,  N_j$ and $S_i, S_j$.
		\item  For any $x \in \localregisters \cup \globalregs$, $\ffun(x(\iconf_j)) = x(\iconf_i)$
	\end{itemize}
	
	We show how to construct the required function $\ffun$.  Let  $\tset_i = U_i \setminus ( (U_i \cap U_j) )$ and $\tset_j = U_j \setminus  ((U_i \cap U_j))$. Note that  $|N_i| = |N_j|$ and $|S_i| = |S_j|$. Then we define a bijective  function $\gfun: \tset_i \to \tset_j$ such that $\val \in N_i$ if and only if $\gfun(\val) \in N_j$ and   $\val \in S_i$ if and only if $\gfun(\val) \in S_j$.
 Clearly such a bijection exists since $|N_i| = |N_j|$ and $|S_i| = |S_j|$. Using this, we define the bijection $\ffun$ as:
	\[
	\ffun(\val) = \begin{cases}
		\val & \text{ If } \val \in \domain \setminus ( U_i \cup U_j ) \text{ or } \val \in \constants \\
		x(\iconf_i) & \text{ If } \val  = x(\iconf_j)\ (\text{ In this case,  } \val \in U_j)\\
		\gfun(\val) & \text{ Otherwise }
	\end{cases}
	\]

	The explanation for the function $\ffun$ is as follows.  Any value that does not appear in $U_i \cup U_j$ or is a constant is mapped to itself. Every value that appear in $U_j$, appears in some register. It is then mapped to the corresponding value from the same register in  $\iconf_i$. Note that this automatically provides a bijection between values in  $U_i \cap U_j$. What remains to be mapped is $U_i \setminus U_j$, this is provided by $\gfun$.
	
	Applying this to run $\iconf_j \dots \iconf_m$ we get another $(N,S)$-consistent acyclic pseudo $1$-run which begins at the configuration $\iconf_i$. Since 
	$\iconf_1 \dots \iconf_i$ is also $(N,S)$-consistent, we get that 
	$\iconf_1 \ldots \iconf_i f(\iconf_{j+1}) \dots f(\iconf_m)$ is a short $(N,S)$-consistent acyclic $1$-pseudo run (It is important that both of them be $(N,S)$ consistent -- otherwise we cannot guarantee acyclicity). More over the run is fresh since every value that is not in $U_i \cup U_j$ is mapped to itself. Since the original run was fresh, it is guaranteed that any value in  $\iconf_{j+1} \dots \iconf_m$ will not be from $N_i \cup N_j$ and hence the new value  in $f(\iconf_{j+1}) \dots f(\iconf_m)$ is same as in the original run.
	Applying this repeatedly, gives us the result.
\end{proof}

We next show that checking if there is an acyclic witnessing run can be accomplished in {\sc PSpace}. The lemma we prove is a generalization of the one in the main section, it involves the presence of constants. In the lemma, we assume support for any word involving constants from $\constants' \subseteq \constants$ and all other values appear at-most $k \in \Naturals$ times.

\begin{lemma}(See Lemma \ref{lem:dec:tracking})
	  For a given $k \in \Naturals$ and $\constants' \subseteq \constants$, there is an algorithm that checks if there is a witnessing acyclic pseudo-$1$ run $\rho$ such that the following holds:
	  \begin{itemize}
	   \item Each non-constant value  in $\supp(\rho)$ appears at-most $k$ times. That is $\project{\wdwidth(\supp(\rho))}{ \domain \setminus \constants} \leq k$ (it is an acyclic $1$-pseudo run  whose support is bounded by  $k$). Here, $\project{\word}{\Sigma}$ returns a word by deleting from $\word$, all the values not in $\Sigma$.
	   \item Each constant appearing  in $\supp(\rho)$ is from $\constants'$. That is, $\project{\wdwidth(\supp(\rho))}{\constants}  =\project{\wdwidth(\supp(\rho))}{\constants'} $
	   \end{itemize}
	   More over, the algorithm runs in space polynomial in the size of the program and $\log(k)$. 
\end{lemma}
\begin{proof}
	We reduce the given problem to reachability in a register automata, operating over a finite domain. It is easy to see that the state space of such an automata is finite and reachability can be checked in space at-most logarithmic is the number of states. For the reduction, we will use the constants $\constants'$, two finite set of domains $\domain_1$ and $\domain_2$ each of size $\eta = | \localregisters \cup \globalregs | +1$. The idea behind the two different domains is to use the values $\domain_1$ as nonces and $\domain_2$  for values from support.  Note that we are only interested in acyclic run. The register automata that we construct has as its local register $\localregisters \cup \globalregs$ as its registers. The states of the register automata is $\Labels \times [1,k]^{\eta}$, let us denote this by $\States_{\tt ac}$.  Then the size of the states is at-most $|\localregisters| \times k^\eta$.
	
	We explain the components of the configuration. The configuration of the register automata is an element from   $\States_{\tt ac} \times\functions{\localregisters \cup \globalregs }{ \domain_1 \cup \domain_2 \cup \constants}  $. Notice that the size of the set of all configurations including the registers is at-most  $|\Labels| \times k^{\eta} \times | \localregisters \cup \globalregs | ^{\eta + |  \constants | }$ and hence a finite state.
	A configuration  is of the form $\tuple{\lab, i_1, \dots i_\eta, \lregmap}$. In such a configuration, $\lab \in \Labels$ is a location of the program, $i_1 \dots i_\eta$ is one entry per each element in $\domain_2$ and denotes the number of times the value was written to the global register (via a pseudo move). This will become clear as we explain the construction. The function $\lregmap$ provides the values stored in each of the register, it can either be a nonce, in which case it contains a value from $\domain_1$, a value it reads from the global register in which case it contains a value from $\domain_2$ or it can be a constant, in which case it is from $\constants'$. We explain the moves between such configurations below. For any pair of configurations $\tuple{\lab, i_1, \dots i_\eta, \lregmap}$ and $\tuple{\lab', i'_1, \dots i'_\eta, \lregmap'}$, we say $\tuple{\lab, i_1, \dots i_\eta, \lregmap} \movesto{} \tuple{\lab', i'_1, \dots i'_\eta, \lregmap'}$ if and only if one of the following holds:
	\begin{enumerate}
		\item If there is a $j \in [1,\eta]$ such that $i_j < k$, then $\ell' = \ell$, for all $m \neq j \in [1,\eta] $, $i'_m = i_m$, $\lregmap' = \updatefun{\lregmap}{\greg}{\val_j}$.
		That is, we allow any value from $\val_j \in \domain_2$ to be written to the global register at any point, provided the current utilization $i_j$ for it is strictly less than $k$.  
		\item If there is a $j \in [1,\eta]$ such that $\val_j \notin \lregmap$, then $\ell' = \ell$ , for all $m \neq j \in [1,\eta] $, $i'_m = i_m$, $\lregmap' = \lregmap$ and $i'_j = 0$. That is, we allow the current utilization of any value from $\val_j \in \domain_2$ to be reset provided it is not present in the configuration.
			\item We have $\ell' = \ell$, for all $m \in [1,\eta] $, $i'_m = i_m$, $\lregmap' = \updatefun{\lregmap}{\greg}{\val}$ for some $\val \in \constants'$.That is, we allow values from $\constants'$ to be pseudo assigned to the global register.
		\item If the current instruction is of the form $\assign{\xvar}{\yvar}$, then $\ell' = \nextof{\ell}$ , for all $m \in [1,\eta] $, $i'_m = i_m$, $\lregmap' = \updatefun{\lregmap}{\xvar}{\lregmap(\yvar)}$ . That is, we allow automata to mimic the assignment moves.
		\item If the current instruction is of the form $\assign{\xvar}{*}$, then $\ell' = \nextof{\ell}$ , for all $m \in [1,\eta] $, $i'_m = i_m$, $\lregmap' = \updatefun{\lregmap}{\xvar}{\val_1}$ for some $\val_1 \in \domain_1$, such that $\val_1 \notin \lregmap$. That is, we allow automata to mimic the non-deterministic assignment moves.
		\item If the current instruction is of the form $\goto\ \ell'$, then  for all $m \in [1,\eta] $, $i'_m = i_m$, $\lregmap' = \lregmap$. That is, we allow automata to mimic the non-deterministic jump.
		
		\item If the current instruction is of the form $\textbf{ if }(\xvar \overset{?}{=} \yvar)\textbf{ then }{ \goto\ \ell'}\textbf{ end}$ and if $\lregmap(\xvar) \overset{?}{=} \lregmap(\yvar)$, then for all $m \in [1,\eta] $, $i'_m = i_m$, $\lregmap' = \lregmap$, here $\ell'$ is as specified in the goto statement if the expression evaluates to $\true$ and $\ell' = \nextof{\ell}$ otherwise. That is, we allow automata to mimic the guarded jump.
	\end{enumerate}

	\smallskip
	\textbf{Correctness} There is an acyclic pseudo-1-run reaching a
	given label, with support bounded by $k$ and using only constants from
	$C'$, if and only if the automaton has a run reaching the same label.
	
	\emph{(Forward direction.)} Given such a real run $\rho$, replay it move
	by move: assign each non-constant support value its own $D_2$-slot for as
	long as it stays in some register, incrementing that slot's counter every
	time the value is pseudo-written (never exceeding $k$, since $\rho$
	respects that bound) and freeing the slot the moment the value leaves
	every register; replay every pseudo-written constant, which must lie in
	$C'$, via rule 3; and replay every ordinary assignment, fresh assignment,
	and jump exactly via rules 4--7, drawing a genuinely fresh $D_1$ label each
	time the program executes $x := \ast$. This yields an accepting run of the
	automaton that reaches the same label.
	
	\emph{(Backward direction.)} Given an accepting run of the automaton,
	unfold it into a real run by inventing, each time a $D_2$-slot or a
	$D_1$-label is claimed after being free, a value that has never been used
	before anywhere in the run, and reusing that same freshly-invented value
	for as long as the slot or label stays claimed. Interpret pseudo-written
	constants as themselves. The resulting run pseudo-writes each non-constant
	support value only as many times as its slot's counter was incremented,
	hence at most $k$ times; it pseudo-writes only constants from $C'$; and --
	because every reused slot or label is given a brand-new real value rather
	than reusing an old one -- the resulting configuration (location together
	with register contents) can never repeat, so the run is automatically
	acyclic, regardless of whether the automaton's own path revisited a
	control state. So it is a witnessing acyclic run, as required.
	
	\smallskip
	\textbf{Space bound} Deciding reachability in this
	automaton only ever needs to hold one configuration in memory at a time:
	guess which rule to fire next, overwrite the stored configuration with the
	result, and accept as soon as the target label is reached (giving up
	after more steps than there are states, which suffices by the standard
	fact that if a state is reachable at all, it is reachable along some
	simple path). Since one configuration already fits in space polynomial in
	the size of the program and logarithmic in $k$, as computed above, so does
	the whole procedure.
\end{proof}

We next show that any realisable constants can be assumed freely, that is if there is a run that writes a constant, then we may assume that it can be generated as many times as required. More over, if a set of constants is generatable, then there is a run that can produce any word involving these constants.

\begin{lemma}
	If there is a run that realizes a constant $\val \in \constants$, then for any $k \in \Naturals$,  there is a run that writes $\val$ at least $k$ times.
\end{lemma}
\begin{proof}
	Let $\run$ be the run that realises $\val$. We assume that the $\run$ writes $\val$ as the last transition. Then $\run$ can be decomposed as $\run_0 \run_1 $, where $\run_0$ does not contain any writes to the global memory and $\run_1$ is such that it starts with a write and ends with a write of value $\val$ to the global register. In order to write the constant $\val$,  $k$ times, we copy the $\run$ as $\run^1, \dots \run^k$ into $k$ many copies. For wach $i \in [1,k]$, we can decompose the run as $\run^i  = \run^i_0 \run^i_1$, following the decomposition of $\run$. Then the required run is obtained as   $\run^1_0 \run^2_0  \dots \run^k_0 \run^1_1 \run^2_1 \dots \run^k_1$. Notice that each of $\run^i_0$ only read the initial value from $\greg$ and $\run^i_1$ start with a write to the global register and writes $\val$ to the global register in the end. Hence they  can be scheduled independently.
\end{proof}

\begin{lemma}\label{apx:lem:oneglobal:const}
	Let $\relconst \subseteq \constants$ be a realizable set of constants, then for any word $w \in \relconst^*$, there is a run that writes $w$ to the global memory. 
\end{lemma}
\begin{proof}
	Since $\relconst$ is given to be realizable, for each $\val \in \relconst$, there is a run $\run_\val$ that realizes it.  Let $w = a_1 a_2 \dots a_m$, then for each $a_i$, there is a run $\run_{a_i}$ that realizes it.   We decompose this run as $\pi_{a_i} = \pi^0_{a_i} \pi^{1}_{a_i}$, where $ \pi^0_{a_i}$ are sequence of transitions that do not write to the global register and $\pi^{1}_{a_i}$ is such that it starts with a write. Then the required run $\run = \pi^0_{a_1} \dots \pi^0_{a_m} \pi^1_{a_1} \dots \pi^1_{a_m}$. Clearly such a run writes $w$ to the global register.
\end{proof}

As an immediate corollary of Lemma \ref{apx:lem:oneglobal:const} and Lemma \ref{lem:dec:kbounded}, we obtain the following lemma.

\begin{lemma}
	If there is a $k$-writing run of $\pds$ and $\relconst \subseteq \constants$ is realizable, then for any $w \in ((\domain\setminus \constants) \cup \relconst)^*$  with $\wdwidth(w\setminus \relconst) \leq k$, there is a run that writes $w$.
\end{lemma}

As a consequence  we have:
\begin{corollary}\label{apx:cor:dec:acyclic}
	Let $k$ be largest value below $\MAX$ such there is a $k$-writing run. Let $\constants' \subseteq \constants$ be the set of all constants that are realizable.  There is an acyclic witnessing run iff there is a witnessing pseudo $1$-run such that  $\project{\wdwidth(\supp(\rho))}{ \domain \setminus \constants} \leq k$ and $\project{\wdwidth(\supp(\rho))}{\constants}  =\project{\wdwidth(\supp(\rho))}{\constants'} $.
\end{corollary}

The rest of this subsection is devoted to showing how to determine the largest $k \leq \MAX$ for which there is a $k$-writing run and the set of all constants $\constants' \subseteq constants$ that are realizable (there is a run that writes each value in it). 
As a first step towards this, we provide the proof of Lemma \ref{lem:dec:RepToWrite} that states that replicating run can also produce any word of width $k$.

\begin{proof} [Proof of Lemma  \ref{lem:dec:RepToWrite} ]
	Lets us assume w.l.o.g. $\run$ has no prefix which is replicating. 
	Notice that the only replicating transition is its last transition.
	Let us also fix a $k \in \Naturals$.

	Let $\irun$ be the projection of $\run$ to the replicating thread (say $i$).
	Let $a$ be the value that is replicated and let $\supp(\irun) = a_1 a_2 \ldots a_m$ ($a$ is one of $a_1, \ldots a_m$).   We decompose, as usual,  $\run$ as $\run_0 \run_1 \dots \run_m\run_{m+1}$  where $\run_i$, $1 \leq i \leq m$, begins with a write and ends in a configuration where $\greg = a_i$. The $\run_0$ is all the initial transitions up to the first write.   We decompose $\irun$ as $\hat\irun_0 \irun_0 \dotarrow{(\greg \leftarrowtail a_1)} \hat\irun_1 \irun_1  \dotarrow{(\greg \leftarrowtail a_2)} \hat\irun_2 \irun_2 \ldots \dotarrow{(\greg \leftarrowtail a_m)} \hat\irun_m \irun_m$ where $\hat\irun_j$ does not  writes to $\greg$ and $\irun_j$ begins with a write to $\greg$, $0 \leq j \leq m$. 
	Notice then that the part $ \hat\irun_j$ is one that reads $a_j$ from $\greg$. Moreover $\irun_j$ can be scheduled independently since it starts with a write to the global register.
	
	Now we make $k$ copies of $\irun$, let them be $\irun^1, \irun^2, \dots,\irun^k$. Moreover for each $i \in [1,k]$, let $\irun^i$ be decomposed as  $\hat\irun^{i}_0 \irun^i_0 \dotarrow{(\greg \leftarrowtail a_1)} \hat\irun^{i}_1 \irun^i_1  \dotarrow{(\greg \leftarrowtail a_2)} \hat\irun^{i}_2 \irun^i_2 \ldots \dotarrow{(\greg \leftarrowtail a_m)} \hat\irun^{i}_m \irun^i_m$, in a manner similar to the decomposition of $\irun$.
	We schedule these copies in conjunction with the decomposition of $\pi$ as follows: First we schedule $\pi_0$ followed by $\hat\irun^{1}_0, \dots, \hat\irun^{k}_0$, note that each of these portions require the initial memory in the global register and hence can be freely scheduled. We follow this up by $\irun^1_0, \dots, \irun^k_0$ of all the copies (one by one), note that these portions start with a write and hence can be scheduled independently. Then we execute $\pi_1$ which again starts with a write and hence is independent of what is in the global register, followed by $\hat\irun^{1}_1, \dots, \hat\irun^{k}_1$ of all the copies one by one, note that each of these portions require $a_1$ to be in the global register and  $\pi_1$ ends with such a write. We follow this by $\irun^1_1, \dots, \irun^k_1$ of all the copies one by one, followed by $\pi_2$, followed by $\hat\irun^{1}_2, \dots, \hat\irun^{k}_2$ of all copies  and so on.  That is, our run will be of the form 
	\[
	\pi_0 \cdot (\hat\irun^{i}_0)_{i \in [1,k]}  \cdot (\irun^i_0)_{i \in [1,k]}  \cdot \pi_1    \ldots \pi_m \cdot (\hat\irun^{i}_m)_{i \in [1,k]}  \cdot (\irun^i_m)_{i \in [1,k]}
	\]
	
	Notice that for any $i \in [1,k]$,  $\hat\irun^{i}_0$ reads the initial memory, for any $j \in [1,m]$, the sequence $\hat\irun^i_j$ reads the value $a_i$ (since in $\pi_j$, the last write to $\greg$ is the value $a_i$). Each of $\irun^i_j$ and $\pi_j$ can be scheduled independent of what is in the shared register since they start with a write to $\greg$.  Hence each of the $k$-copies $\irun_i$ reads the value that was read by $\irun$ in $\pi$. More over the values written by $\irun_i$ is never read by any other thread. Hence, we  conclude that each  $\irun_i$ mimics the behavior of  $\irun$ and writes the value $a$ to the shared register
\end{proof}

We next prove the generalization of the Lemma \ref{lem:dec:kwriting} in the presence of constants. This lemma characterizes a $k$-writing run.

\begin{lemma} [See Lemma \ref{lem:dec:kwriting}]\label{apx:lem:dec:kwriting}  
	For any $k$, there is a  $k$-writing  run if and only if there is a $j < k$ and a realizable set of constants $\relconst$ such that both of the following holds:
	\begin{itemize}
		\item 	There is a $j$-writing non-replicating run $\run$
		\item There is an acyclic
	pseudo $1$-run $\irun$ with $\asupp(\irun) \cap \constants \subseteq \relconst^*$ and   $\wdwidth(\project{\supp(\irun)}{\domain \setminus \constants}  ) \leq j$ which is 
	either replicating or $k$-writing.
	\end{itemize}
\end{lemma}
\begin{proof}
	Let $\run$ be a $k$-writing run of minimum length (i.e. no strict prefix of $\run$ is $k$-writing).  Such a run ends with a transition where a thread $i$ writes $d$ to $\greg$, where $d$ is the only nonce written to the global register $k$ times along $\run$.  Let $\irun$ be the $1$-pseudo run  obtained by projecting this run onto the thread $i$. We first observe that every constant that it reads along this projection is indeed realisable (in fact by a prefix of $\run$).

	Suppose  $\run$ is non-replicating. Then it is acyclic (since each nonce can be written only by one thread). Thus,  the $1$-pseudo run $\irun$ obtained by projecting this run on $i$ is an acyclic $k$-writing $1$-pseudo run. More over, it has  as a support a word whose width say $j$ on the non-constants letters. Since we have assumed our run to be of minimum length, it is easy to see that $j$  is less than $k$. %
	Further $\run$ also serves as the $j$-writing run required by the Lemma.

	Suppose $\run$ is a replicating run. Let $\run'$ be its shortest replicating prefix and let $i$ be the thread
	executing the replicating transition.  The projection of $\run'$ to $i$ serves as an acyclic replicating pseudo $1$-run (whose support has width  $< k$).
	Let the prefix of $\run'$ obtained by omitting its last transition  be $\run''$.
	Now, $\run''$ supports the projection of $\run'$ on the thread $i$ and is also
	non-replicating.
	
	For the converse: Since there is a $j$-writing run and $\relconst$ are realizable,  there is also a run that can write the word that the supports  $\irun$. 
	Following 
	Lemma \ref{lem:dec:1greg:forward}, we can compose them to obtain  $k$ writing run. 
	Otherwise, we have a replicating acyclic $1$-pseudo run $\irun$ and a supporting run $\run$. Then, along the lines of Lemma \ref{lem:dec:1greg:forward}, have a replicating run. Then by Lemma \ref{lem:dec:RepToWrite}, we have the required $k$-writing run.
\end{proof}

We next show how to compute the realizable set constants. The proof of the below lemma is along the lines of the previous Lemma \ref{apx:lem:dec:kwriting}, we omit the proof.

\begin{lemma}\label{apx:lem:dec:acyclicconst}
	There is an acyclic  run that realizes a constant $\const \in \constants$ if and only if there is a $k$-writing run for some $ k \leq \MAX$, a set of realizable constants $\relconst$ and a  pseudo $1$-run $\irun$  that  writes the constant $\const$ to $\greg$ such that $\asupp(\irun) \cap \constants \in\relconst^*$ and    $\wdwidth(\supp(\irun_{\setminus \constants})) \leq k$.
\end{lemma}

Finally the following lemmas will be useful in computing the maximum $k$ for which there is a $k$-writing run and the set of all realizable constants.

\begin{lemma}
	\label{lem:33}
	Given a set of constants $C' \subseteq C$ and $j \in \mathbb{N}$, there is an
	algorithm that checks whether there is a pseudo-1-run $\rho$ and a value
	$a \in \nonce(\rho) \setminus C$ such that
	\begin{enumerate}
		\item $a$ is written to the shared register at least $j+1$ times along
		$\rho$ (i.e.\ $\rho$ is $(j+1)$-writing, witnessed by $a$),
		\item every other non-constant value is written at most $j$ times, i.e.\
		$\mathrm{wd}\big(\mathrm{supp}(\rho) \setminus (C \cup \{a\})\big) \le j$, and
		\item $\mathrm{valsupp}(\rho) \cap C \subseteq C'$.
	\end{enumerate}
	The algorithm runs in space polynomial in the size of the system
	and logarithmic in $j$.
\end{lemma}

\begin{proof}
	As before, we reduce the problem to a reachability question in a
	finite-state register automaton.
	
	Fix the constant set $C' \subseteq C$ together with two auxiliary finite
	domains $D_1$ and $D_2$, each of size $\eta = |L \cup G| + 1$: $D_1$
	supplies fresh, nonce-style values, while $D_2$ supplies the abstract
	``slots'' used to track the part of the support that is not the witness
	value, bounded by $j$. This choice of $\eta$ suffices because, at any
	moment, at most $|L \cup G|$ distinct non-constant values can be sitting in
	the registers at once; adding one more slot for the witness gives exactly
	the $\eta = |L \cup G| + 1$ values we ever need to tell apart. Since the
	program can only compare register contents for equality (rule~9 below) and
	only ever introduces new values through the instruction $x := \ast$
	(rule~7), any two runs that agree up to a consistent renaming of which live
	value occupies which of these $\eta$ slots look identical to the program.
	Nothing is therefore lost by working over a fixed alphabet of this size and
	recycling a slot as soon as its current occupant has left every register.
	
	The automaton keeps $L \cup G$ as its registers, and its control states are
	given by $Q_{j+1} = \Lambda \times [0,j]^{\eta} \times (D_1 \cup \{\bot\}) \times [0,j+1]$;
	a full configuration also carries a register valuation
	$\vec l : L\cup G \to D_1 \cup D_2 \cup C$. Concretely, a configuration
	$\langle \ell, i_1, \dots, i_\eta, x, \kappa, \vec l \rangle$ consists of the
	current program location $\ell$, one usage counter $i_m$ per $D_2$-slot
	(each capped at $j$), the guessed witness value $x$ (held at $\bot$ until it
	is chosen), the number $\kappa$ of times $x$ has so far been written to the
	shared register, and the register valuation $\vec l$ itself. Once $x$ is
	set it never changes again, and it is required to persist in some register
	for the remainder of the run. None of this is expensive to store: the
	counters and $\kappa$ each need only $O(\log j)$ bits, $x$ needs
	$O(\log \eta)$ bits, and $\vec l$ needs $O(|L \cup G| \cdot \log(\eta + |C'|))$
	bits, so a single configuration fits comfortably in space polynomial in the
	size of the system and logarithmic in $j$.
	
	The transitions capture the following possibilities.
	
	\begin{enumerate}
		\item Any of the $\eta$ tracked $D_2$-values may be pseudo-written to the
		register, as long as its counter has not yet reached $j$; doing so
		increments that counter by one and leaves the witness $x$ and its
		count $\kappa$ untouched.
		
		\item Once a $D_2$-value has left every register, its counter may be reset
		to $0$, freeing that slot to represent a different value from then on.
		
		\item Any constant $d \in C'$ may likewise be pseudo-written to the
		register. Since only constants from $C'$ are ever offered here,
		every accepting run automatically respects
		$\mathrm{valsupp}(\rho) \cap C \subseteq C'$ with no further
		bookkeeping needed.
		
		\item As long as the witness has not yet been fixed ($x = \bot$), the
		automaton may non-deterministically commit to one of the values
		$d \in D_1$ currently present in some register, setting $x := d$ and
		resetting $\kappa$ to $0$ — this is the moment the algorithm decides
		which value will play the role of $a$.
		
		\item Once $x$ has been fixed, whenever the program actually executes the
		write instruction $r := x$ (with $\vec l(x) = d$), the automaton
		follows this real assignment, updates $\vec l$ accordingly, and
		increments $\kappa$ by one. This is what lets $\kappa$ climb, one real
		write at a time, all the way up to $j+1$.
		
		\item An ordinary local assignment $x := y$ is simulated faithfully,
		copying $\vec l(y)$ into $x$, and leaves the counters, the witness,
		and $\kappa$ untouched. This rule is guarded so that it never fires on
		the specific witness-write already handled by rule~5, and it requires
		that the witness value continue to sit in some register after the
		update, preserving the invariant that $x$, once chosen, is never lost.
		
		\item A non-deterministic assignment $x := \ast$ is simulated by writing a
		genuinely fresh value $d_1 \in D_1$ — one not currently held in any
		register — into $x$; the counters, witness, and $\kappa$ are again
		unaffected, subject to the same invariant that the witness value
		must remain present in some register throughout.
		
		\item An unconditional jump $\mathtt{goto}\ \ell'$ simply moves control to
		$\ell'$, leaving every other component unchanged.
		
		\item A guarded jump
		$\mathtt{if}\ (x \overset{?}{=} y)\ \mathtt{then\ goto}\ \ell'\ \mathtt{end}$
		follows the true branch to $\ell'$ exactly when $\vec l(x) = \vec l(y)$,
		and otherwise falls through to $\mathrm{next}(\ell)$; all other
		components stay the same either way.
	\end{enumerate}
	
	\smallskip
	\textbf{Correctness} 
	If a real pseudo-1-run $\rho$ witnesses the lemma via some value $a$, replay
	it label by label: every ordinary value that the environment pseudo-supplies
	to the register gets its own $D_2$-value for as long as it stays in a
	register (freeing the value once it leaves, exactly the recycling from (a)). Similarly, every self generated nonce is represented by a value from $D_1$. Notice that each of these cannot exceed the number of total registers at any point.
	The moment $a$ is first generated by the thread's own $x := \ast$
	instruction, guess it as the witness $x$, it is guaranteed to stay in some local register till it is written at-least $j+1$ times.  Simulate each of the (at least
	$j+1$) times the thread's own instruction $r := x$ genuinely writes $a$ to
	the register as a firing of the counting rule (c), driving $\kappa$ to
	$j+1$. Since $\rho$ obeys the $\le j$ bound on every other, pseudo-supplied
	value, no $D_2$-slot's counter ever needs to exceed $j$ either. The result
	is an accepting run.
	
	Conversely, given an accepting run, we provide concrete values to the abstracted values: every time a
	$D_2$-value is introduced fresh, invent a brand-new pseudo-supplied value for
	it and keep using that same value until the value is freed (these are the
	values the environment feeds into the register); treat pseudo-written
	constants as themselves; and the moment the witness is guessed, generate
	one fresh value $a$ for the thread itself to hold, and let the thread
	genuinely write $a$ to the register, via its own instruction $r := x$,
	once for every later firing of the counting rule (c). This produces an
	actual pseudo-1-run: the thread itself writes $a$ at least $j+1$ times,
	once per firing of (c); every other, pseudo-supplied value, having lived
	in exactly one slot whose counter never passed $j$, occurs at most $j$
	times in the support; and only constants from $C'$ are ever pseudo-supplied.
	So it is a witnessing run via $a$, as required.
	
	\smallskip
	\textbf{Space bound} Deciding whether the automaton has an
	accepting run is just a reachability search, and reachability search never
	needs to remember more than the single configuration it currently has in
	hand: repeatedly guess which rule to fire next, overwrite the stored
	configuration with the result, and accept as soon as the target location
	is reached with $\kappa = j+1$ (giving up after enough steps to rule out an
	infinite search). Since one configuration already fits in space polynomial
	in the size of the system and logarithmic in $j$, as computed above, so
	does the whole procedure.
\end{proof}

A similar construction will yeild the following lemma which states that checking if constants can be generated can also be achieved in {\sc PSpace}.

\begin{lemma}
	\label{lem:replicatepspace}
	Given a set of constants $C' \subseteq C$ and $j \in \mathbb{N}$, there is an
	algorithm that checks whether there is an acyclic replicating pseudo-1-run
	$\rho$ such that
	\begin{enumerate}
		\item every non-constant value in $\mathrm{supp}(\rho)$ is written at
		most $j$ times, i.e.\ $\mathrm{wd}\big(\mathrm{supp}(\rho) \setminus C\big) \le j$,
		\item $\mathrm{valsupp}(\rho) \cap C \subseteq C'$.
	\end{enumerate}
	The algorithm runs in space polynomial in the size of the system
	and logarithmic in $j$.
\end{lemma}

\begin{proof}[Proof Idea]
	Recall that a run is replicating exactly when some thread writes to the
	shared register a value that is non-local to it -- i.e.\ a value it did
	not itself generate as a nonce, so it must have obtained it earlier from
	the register. In the pseudo-1-run setting this means: at some point, the
	thread performs a genuine write $r := x$ where $\vec l(x)$ is a value that
	was earlier pseudo-supplied to it (tracked as a $D_2$-value, as in
	Lemma~\ref{lem:dec:tracking}), rather than a nonce it minted itself via
	$x := \ast$ (a $D_1$-value). We reuse the construction of
	Lemma~\ref{lem:dec:tracking} unchanged -- so the width-$j$ bound on the
	support and the restriction of constants to $C'$ are enforced exactly as
	before -- and simply replace its target label with the condition of
	reaching a configuration from which the next instruction assigns to the
	global register a value from $D_2$. Conversely, any such reachable
	configuration unfolds, via the same freshness argument as in
	Lemma~\ref{lem:dec:tracking}, into a genuine bounded run that becomes
	replicating the moment that one instruction is executed.
\end{proof}

The following algorithm then checks for the existence of an acyclic witness in {\sc Pspace}.

\begin{algorithm}
	\caption{Writing bound, realisable constants and demand bound}
	\label{alg:saturate}
	\KwOut{$K$: the largest $k \le \mathrm{MAX}$ such that a $k$-writing run exists;\\
		$\constants'$: the set of realisable constants;\\
		$J$: A $j$ such that a \kreadrepPar{j} run exists}
	$K \gets 0$;\quad $\constants' \gets \emptyset$;\quad $J \gets \infty$;\quad
	$\mathit{rep} \gets \false$;\quad $\mathit{changed} \gets \true$\;
	\While{$\mathit{changed}$}{
		$\mathit{changed} \gets \false$\;
		\If{$\neg\mathit{rep}$}{
			\uIf{there is an acyclic, replicating, $(K,\constants')$-supported pseudo-1-run}{
				$\mathit{rep} \gets \true$;\quad $J \gets K$;\quad $K \gets \mathrm{MAX}$;\quad
				$\mathit{changed} \gets \true$\;
			}
			\ElseIf{$K < \mathrm{MAX}$ \textbf{and} there is an acyclic,
				$(K,\constants')$-supported pseudo-1-run that writes one of its
				own nonces at least $K+1$ times}{
				$K \gets K + 1$;\quad $\mathit{changed} \gets \true$\;
			}
		}
		\tcp{Constants written by acyclic threads (Lemma 24)}
		\ForEach{$\val \in \constants \setminus \constants'$}{
			\If{there is an acyclic, $(K,\constants')$-supported pseudo-1-run
				that writes $\val$ to $r$}{
				$\constants' \gets \constants' \cup \{\val\}$;\quad $\mathit{changed} \gets \true$\;
			}
		}
		 \If{$\mathit{rep}$}{
			\tcp{Lower the demand bound (Lemma 13)}
			\If{$J > 1$ \textbf{and}
				(there is an acyclic, replicating, $(\mathrm{MAX},\constants')$-supported
				pseudo-1-run that is \kreadrepPar{(J-1)}
				\textbf{or} there is a cyclic pseudo-1-run with support constants in
				$\constants'$ that is \kwriteprPar{J} and \kreadrepPar{(J-1)})}{
				$J \gets J - 1$;\quad $\mathit{changed} \gets \true$\;
			}
			\ForEach{$\val \in \constants \setminus \constants'$}{
				\If{there is a cyclic pseudo-1-run with support constants in $\constants'$
					that is \kwriteprPar{J} and writes $\val$ to $r$}{
					$\constants' \gets \constants' \cup \{\val\}$;\quad $\mathit{changed} \gets \true$\;
				}
			}
		}
	}
	\Return $(K, \constants', J)$\;
\end{algorithm}

\begin{algorithm}
	\caption{Checking for existence of an acyclic witness}
	\label{alg:acyclic}
	Compute $(K, \constants', J)$ with Algorithm~\ref{alg:saturate}\;
	\tcp{Corollary 4}
	\eIf{there is an acyclic witnessing $(K,\constants')$-supported pseudo-1-run}{
		\Return \true\;
	}{
		\Return \false\;
	}
\end{algorithm}

\begin{algorithm}
	\caption{Checking for existence of a cyclic witness}
	\label{alg:cyclic}
	Compute $(K, \constants', J)$ with Algorithm~\ref{alg:saturate}\;
	\tcp{Lemma 10; $J < \infty$ holds exactly when a replicating run exists}
	\eIf{$J < \infty$ \textbf{and} there is a cyclic witnessing pseudo-1-run with
		support constants in $\constants'$ that is \kwriteprPar{J}}{
		\Return \true\;
	}{
		\Return \false\;
	}
\end{algorithm}

	\section{Single Global Reachability  Cyclic Case}\label{apx:sec:oneglobal:cyclic}

As stated in the main section , our objective is to prove that determining the existence of a cyclic witnessing run is decidable in {\sc Pspace}.

As a start, we show in the below lemma states that finding a cyclic witness in the presence of constants amount to finding a \kwritepr  cyclic pseudo $1$-run and its corresponding support.

\begin{lemma}\label{apx:lem:constcycwitness} (See Lemma \ref{lem:dec:cyclicchar})
	There is a cyclic witnessing run if and only if there are a \kreadrep replicating run and   a  witnessing \kwritepr  cyclic pseudo $1$-run  $\irun$ with $\asupp(\irun) \cap \constants \subseteq \relconst$, for some $k \in \Naturals$ and some realizable set $\relconst \subseteq \constants $.
\end{lemma}
\begin{proof}
	One direction of the lemma is easy to see. Note that if there is a cyclic witnessing run $\run$, projecting on to the witnessing thread provides us with a cyclic pseudo $1$-run $\irun$. Clearly this pseudo run is \kwritepr for some $k$. More over, the run itself is a \kreadrep since $\irun$ is part of this run. Further more any constants that $\irun$ reads is clearly realizable (in fact by some prefix of $\run$). 
	
	For the other direction, recall that if there is a run that realizes a set of constants $\relconst$, then we can write any word $w \in \relconst^*$ (see Lemma \ref{apx:lem:oneglobal:const}).  Let $\irun$ be the given witnessing cyclic pseudo $1$-run, let $\run'$ be the given \kreadrep. Let $w = \project{\supp(\irun)}{\constants}$ be the word got from the support of $\irun$ by retaining only the constants ( then $w \in \relconst^*$). Since $\relconst$ is given to be realisable, we may assume a run $\run_1$ that can write $w$. Similarly,  since we are promised a replicating run ( \kreadrep replicating run ), we may in fact assume that the run $\run_1$ that writes the values in $\supp(\irun) \setminus (\nonce(\irun))$.  That is, it provides for all the values in the supportof $\irun$ that are not a part of its  nonce.
	
	By combining this run $\run_1$ with $\irun$ (as in Lemma \ref{lem:dec:1greg:forward}), we can obtain a pseudo run $\irun_1$ such that $\supp(\irun_1) \in \nonce(\irun)^*$.  In Lemma \ref{apx:lem:dec:replrun} (stated and proved below), we prove that if there is a pseudo run that is a \kwritepr and there is a run that is \kreadrep, then we can combine these in such a way to obtain a pseudo run $\irun_2$ such that $\asupp(\irun_2) \cap \nonce(\irun_2) = \emptyset$. Now combining $\irun_1$ with the assumed \kreadrep $\run'$, we can then obtain the required witnessing run $\run$
\end{proof}

\begin{lemma}\label{apx:lem:dec:replrun}
If there is a   witnessing cyclic pseudo $1$-run $\irun$ that  \kwriteprPar{k} such that $\asupp(\irun) \subseteq \nonce(\irun)$ and a  \kreadrepPar{j} for some $j \leq k$  run, then there is  a witnessing cyclic run. More over, the runs we construct is also  a  \kreadrep.
\end{lemma}
\begin{proof}
	
		Firstly note that the run $\pi$, replicates a value that was generated in $\pi$, where as the pseudo run reads back its own nonce.  The challenge is to combine $\pi$ with $\irun$  in such a way that each value in $\asupp(\irun') \cap \nonce(\irun')$ is written by a thread other than $\irun$.
	
	Since  pseudo run $\irun$ is one that  \kwritepr, it writes each nonce value $\val$, at least $k$ times. Our idea is to induct on number of elements in $\asupp(\irun)$ (recall that $\asupp(\irun) \subseteq \nonce(\irun)$) and reducing it each time till there are none. Suppose it is empty, we have nothing to do. Otherwise, there is at least one element (say $\val$). The pseudo run can be decomposed as   $\irun =  \irun^0 \irun^1_\val \irun^2_\val \dots \irun^k_\val \irun^{k+1}$, where $\irun^0 $ is the execution up to the first write to the shared register, for each $i \in [1,k]$,  $\irun^i_\val$  is an execution that starts with a write to $\greg$ and is  up to the point where $\val$ is written to the shared register followed by all the transitions that do not write to the shared register and the last $\irun^{k+1}$ is rest of the execution.%
	
	Note that the pseudo run expects to read back $\val$ only in $\irun^{k+1}$. We show how to combine $  \irun$ with the given \kreadrepPar{j} run $\run_1$ to obtain another run, without loss of generality, since $ j \leq k$, we will assume $j=k$. In particular, this run can replicate/write $\val$. Note that $\irun$ can read this value arbitrarily many times in  $\irun^{k+1}$.  For the sake of readability,	we will assume that $\irun^{k+1}$ requires the value exactly  once, the idea can easily be generalized to other values by considering additional copies of $\pi$.

	Assume $\pi$ is \kreadrep, then there is a value say $e$ that it replicates. Let us assume that a thread $t$ replicates this value after reading it exactly $k$ times.  Without loss of generality, we will assume that $\pi$ ends in a replicating transition.  Then, the value $e$ is written by other threads at least $k$ times before the thread $t$ replicates.  We decompose $\pi$ as 
	
	$$
	\pi = \pi_0 \cdot \pi^1_e \cdot \hat\pi^1_e \cdot  \pi^2_e \cdot \hat\pi^2_e \cdots \pi^k_e \cdot\hat\pi^k_e\cdot \pi^{k+1}
	$$
	Where $\pi_0$ is the sub-execution of $\pi$ up to the first write to $\greg$,  for each $i \in [1,k]$, $ \pi^i_e $ is sub-execution that start with a write to the shared register and ends in a transition that writes the value $e$ to the shared register and $\hat\pi^i_e $ is the sub-execution of  $\pi$ that do not write to the shared register, immediately after $\pi^i_e $, $\pi^{k+1}$ is rest of the execution. Note that $\pi^{k+1}$ ends with a transition that replicates $e$.
	
	We will decompose  for each $i \in [1,k]$, $ \hat\pi^i_e $ further as $ \hat\pi^i_e  =  \hat\pi^i_1  \hat\pi^i_2$, where $ \hat\pi^i_1 $ is the execution of all the threads other than $t$ and $ \hat\pi^i_2$ is the execution of $t$ in $ \hat\pi^i_e$. Note that such a decomposition is possible since the execution in each thread involves no writes to $\greg$. Now consider the following pseudo run $\bar\irun$ that combines $\irun$ with $\pi$.
	$$
	\bar\irun = 	\irun_0 \cdot \pi_0 \cdot \pi^1_e \cdot\hat\pi^1_1 \cdot \irun^1_\val \cdot \pi^1_2 \cdot  \pi^2_e \cdot\hat\pi^2_1 \cdot \irun^2_\val \cdot \pi^2_2 \cdots \pi^k_e \cdot\hat\pi^k_1 \cdot \irun^k_\val \cdot \pi^k_2  \cdot \pi^{k+1}
	$$
	
	Firstly it is easy to see that such an execution is possible. Notice that in $\bar\irun $, the thread $t$ reads the value $\val$ instead of $e$ and the value read by any other thread  from $\pi$ is exactly the same. Moreover, the values other than $e$ that is  read by the thread $t$ is also the same as in $\pi$. Upon executing the replicating transition, the thread $t$ replicates the value $\val$ instead of $e$.  By arranging enough copies of $\pi$, we can also arrange $\val$ to be replicated arbitrary number of times. 
\end{proof}

The following lemma characterises the witnessing cyclic run.

\begin{lemma}\label{apx:lem:dec:replruncomplete} (See Lemma \ref{lem:sglobal:replicatingrun})
	There is a witnessing  cyclic run if and only  there is a witnessing run that \kreadrep and there is a witnessing cyclic pseudo $1$-run $\irun$ that  \kwritepr  with $\asupp(\irun) \cap \constants \subseteq \relconst$ for som $k \in \Naturals$ and realisable set of constants $\relconst$.

\end{lemma}
\begin{proof}
	Let $\run$ be the given cyclic witness. Then there is an witnessing thread,  projecting on this thread, provides us a $1$-pseudo run $\irun$ that is  \kwritepr for some $k$. More over the run itself is \kreadrep as required.

	For the other direction,  We are given  is a \kwritepr cyclic witnessing $1$-pseudo run  and a  \kreadrep run. More over we have that the constants appearing in $\asupp(\irun)$ are realizable.  Then, we may combine these as in Lemma \ref{apx:lem:dec:replrun} to obtain the cyclic witness run. 
\end{proof}

As a next step, we prove the below lemma that bounds the length of a minimal  replicating ( and hence also the $k$ in a minimal \kreadrep).

\noindent
\textbf{Lemma \ref{lem:dec:oneglobal:maxcyclic}. }{\it 
	If a replicating run exists, then there is one that is \kreadrep for some $k \leq \MAX$.
}
\begin{proof}[proof sketch]
	We establish the $\NMAX$ bound as follows:  Let $\pi$ be a replicating run. Without loss of generality, assume that it has no prefix that is replicating. Thus its projection on any thread is acyclic. In particular,
	projecting on the only thread that replicates gives us an acyclic replicating pseudo $1$-run $\irun$  whose length can be bounded by $\MAX$ (along the same lines as Lemma \ref{lem:dec:1pseudobound}). The number of pseudo assignments here is bounded by $k  \leq\MAX$ and hence, if we can arrange  a supporting run, by Lemma \ref{lem:dec:1greg:forward}, we can obtain  a replicating run that \kreadrep.  Let $j$ be the width of the support $\irun$. Since $\pi$ is replicating, by Lemma \ref{lem:dec:RepToWrite}, we have $j$-writing run and hence by Lemma \ref{lem:dec:kbounded}, we also have a supporting run for $\irun$.
\end{proof}

We next show that the existence of a witnessing $1$-pseudo run  \kwritepr can be checked in space polynomial in the size of the system and logarithmic in $k$. Similarly, we can also check for the existence of a $1$-pseudo run that is \kreadrepPar{k-1} can be achieved in the same space. The idea here is to keep track of the number of times each nonce is written to the global register and allow its consumption only if it has been written at-least $k$ times. The proof is similar to \ref{lem:33}, we omit the same.

\begin{lemma}
	Given $k$ and $\constants' \subseteq \constants$, the existence of a \kwritepr witnessing pseudo-1-run $\irun$, such that $\asupp(\irun) \cap \constants \subseteq \constants'$   can be checked in space polynomial in the system size and $\log(k)$.
\end{lemma}

\begin{lemma}
	Given $k$ and $\constants' \subseteq \constants$, the existence of a \kwritepr  pseudo-1-run $\irun$ that is also \kreadrepPar{k-1}, such that $\asupp(\irun) \cap \constants \subseteq \constants'$   can be checked in space polynomial in the system size and $\log(k)$.
\end{lemma}

We next show characterization on when  there is a \kreadrep  run.

\begin{lemma}\label{apx:lem:dec:replrun:char}
	There is a \kreadrep replicating run if and only if there is a subset $\relconst \subseteq \constants$ of constants that is realizable and one of the following holds
	\begin{itemize}
		\item There is a  run that \kreadrepPar{j} for some $j < k$ and there is a cyclic pseudo $1$-run $\irun$ that  \kwriteprPar{j} and \kreadrep  with $\asupp(\irun) \cap \constants \subseteq \relconst$.
		\item There is a  acyclic $1$-pseudo run \kreadrep   and a corresponding supporting run.
	\end{itemize}
\end{lemma}
\begin{proof}
	Let $\run$ be the given \kreadrep replicating run, we will assume that it is the shortest such run. Then there is a thread that  \kreadrep.  Projecting on this thread, provides us a $1$-pseudo run $\irun$. This pseudo run can either be cyclic or acyclic.  If it is acyclic, then clearly $\run$ itself is a supporting run. Suppose it is cyclic,  since we have assumed $\run$ to of minimal length, the run clearly is  \kreadrepPar{j} for some $j < k$. Clearly, the replicating thread is \kreadrepPar{j}. It is also easy to see that the constants appearing in $\irun$  are all realisable by the very run.

	For the other direction,  suppose there is a   \kreadrep acyclic $1$-pseudo run and a  supporting run.  Then we may combine both of these as in Lemma \ref{lem:dec:1greg:forward}. Other wise we are given s a \kreadrep cyclic $1$-pseudo run that is  \kwriteprPar{j} and a  \kreadrepPar{j} run. We also have that, the set of constants $\relconst$ are realizable.  Then, we may combine these as in Lemma \ref{apx:lem:constcycwitness} to obtain the  \kreadrep replicating run. 
\end{proof}

So all that remains to be shown is how to compute the set of all realizable constants. Note that a constant can either be realized by an acyclic run (that is, the instance writing the constant is acyclic) or by a cyclic run. We show in each case how to find it. We  note that in order to write a constant, we will have to reach a state/location in the program which executes $ \assign{\greg}{\val}$ or  $ \assign{\greg}{\yvar}$ where $\yvar = \val$. This is similar to finding a witnessing run and can be summarized using the following lemmas. With this, we also obtain that the set of all realizable constants can be found in {\sc PSpace}.

\begin{lemma}\label{apx:lem:dec:acyclicconst:cyc}
	There is an acyclic  run that realizes a constant $\const \in \constants$ if and only if there is a $k$-writing run for some $ k \leq \MAX$, a set of realizable constants $\relconst$ and a  pseudo $1$-run $\irun$  that  writes the constant $\const$ to $\greg$ such that $\asupp(\irun) \cap \constants \in\relconst^*$ and    $\wdwidth(\supp(\irun_{\setminus \constants})) \leq k$.
\end{lemma}

\begin{lemma}\label{apx:lem:cycrealize}
	A constant $\const \in \constants$ is realizable by a cyclic  run if and only if there is a $\relconst \subseteq \constants $ that is realizable, $k \in \Naturals$ such that there is a that is \kreadrep replicating run and   a   \kwritepr  cyclic pseudo $1$-run  $\irun$ with $\asupp(\irun) \cap constants \subseteq \relconst$ that writes $\const$ to $\greg$.
\end{lemma}

\section*{The overall procedure to find witnessing run}

To discover the maximum $k$ for which there is a $k$-writing run and the set of constants with acyclic run. We iterate the following sequence, for each $i$ upto $\MAX$, starting from $1$ till it is no longer possible.
\begin{itemize}
	\item If there is a $\val \in \constants$ such that there is an acyclic pseudo run $\irun$ whose support is bounded by $i-1$ that writes $\val$, then we add $\val$ to $\relconst$.
	\item  Check for $i$-writing acyclic  pseudo run $\irun$ with  $\supp(\wdwidth(\irun \setminus \constants)) \leq i-1$ and $\asupp(\irun) \subseteq \relconst$.
\end{itemize}

The above procedure provides us with maximum $i$ for which there is a $i$ writing run and the set of  realizable constants using an acyclic run. 
We check if there is an acyclic  replicating pseudo $1$-run $\irun$ with  $\supp(\wdwidth(\irun \setminus \constants)) \leq i$ and $\asupp(\irun) \subseteq \relconst$. In this case we let $i = \MAX$.  If there is a $\val \in \constants$ such that there is an acyclic pseudo run $\irun$ whose support is bounded by $\MAX$ that writes $\val$, then we add $\val$ to $\relconst$. This provides us with  the set of all realizable constants using an acyclic run. 

We then check for witnessing acyclic $1$-pseudo run (assuming $\relconst$ and $i$-writing run) and return success if found.

We  then do the following to discover the minimum $i$ for which there is a \kreadrep replicating run and the set of all constants realizable by a cyclic run. For $i$ starting from $\MAX$ and decreasing each time, we repeat the following till the second action fails.

\begin{itemize}
	\item If there is a $\val \in \constants$ such that there is a cyclic pseudo run $\irun$ that is  \kwriteprPar{i} and $\asupp(\irun) \subseteq \relconst$ that writes $\val$, then we add $\val$ to $\relconst$.
	\item  Check for an \kwriteprPar{i}  cyclic pseudo run $\irun$, that is also \kreadrepPar{i-1} with   $\asupp(\irun) \subseteq \relconst$.
\end{itemize}

This gives us a minimum $i$ for which there is a \kreadrepPar{i} replicating run.  We check for a  witnessing \kwritepr cyclic pseudo run (assuming $\relconst$ and  \kreadrepPar{i} replicating run and return success if found. We return failure otherwise. A detailed pseudo code is provided in \ref{alg:acyclic} and \ref{alg:cyclic}.

\paragraph{$\nu$-nets.} \label{apx:sec:sl:nunet}
A $\nu$-net is a Petri net in which every token carries a tag from the
infinite domain $\domain$. Transitions refer to tags through two disjoint
sets of variables: variables in $\var$ match tags that are already present,
and variables in $\newvar$ create tags that are not present.
Formally, a $\nu$-net is a tuple $N = (\places,\transition,\infun,\outfun)$,
where $\places$ is a finite set of places, $\transition$ is a finite set of
transitions, and $\infun: \places \times \transition \to \var^\oplus$ and
$\outfun: \places \times \transition \to (\var \cup \newvar)^\oplus$ are the
flow functions. A \emph{marking} is a map $\marking: \places \to \domain^\oplus$.
We write $\marking_1 \geq \marking_2$ if $\marking_1(p) \geq \marking_2(p)$
for every $p \in \places$, and $\domain(\marking)$ for the set of tags
occurring in $\marking$.

For a transition $t$, let $\mathit{Var}(t)$ be the set of variables occurring
in $\infun(\cdot,t)$ or $\outfun(\cdot,t)$. A \emph{mode} of $t$ in a marking
$\marking$ is an injective map $\sigma: \mathit{Var}(t) \to \domain$ such that
$\sigma(x) \in \domain(\marking)$ for every $x \in \mathit{Var}(t) \cap \var$
and $\sigma(y) \notin \domain(\marking)$ for every
$y \in \mathit{Var}(t) \cap \newvar$. A mode extends to multisets of variables
by $\sigma(\mset)(\val) = \sum_{z \,:\, \sigma(z) = \val} \mset(z)$. The transition
$t$ is \emph{enabled} in $\marking$ with mode $\sigma$ if
$\sigma(\infun(p,t)) \leq \marking(p)$ for every $p \in \places$; firing it
yields the marking $\marking'$ with
$\marking'(p) = \marking(p) - \sigma(\infun(p,t)) + \sigma(\outfun(p,t))$,
and we write $\marking \movesto{t} \marking'$. We write $\movesto{}^*$ for the
reflexive transitive closure of $\movesto{}$. A marking $\marking'$ is
\emph{coverable} from $\marking$ if $\marking \movesto{}^* \marking''$ for some
$\marking'' \geq \marking'$.

\paragraph{The construction.}
Let $\pds$ be a program with a single local register $\xvar$, and let
$\ell_{\tt fin} \in \Labels$ be the target label. We build a $\nu$-net $N$
with an initial marking $\marking_{\tt in}$ and a target marking
$\marking_{\tt fin}$ such that $\marking_{\tt fin}$ is coverable from
$\marking_{\tt in}$ if and only if $\ell_{\tt fin}$ is reachable by a fresh
run of $\pds$.

The places are
$\places = \globalregs \cup \constants \cup \Labels \cup \{\mathit{done}\}$:
one place for each global register, each constant and each label, and one
extra place $\mathit{done}$. A token in a label place $\ell$ with tag $\val$
represents a thread at label $\ell$ whose local register holds $\val$. The
place of a global register holds a single token, whose tag is the register's
current value. The place of a constant $c$ holds a single token tagged $c$,
which never changes.

In $\marking_{\tt in}$, the place of each global register and the place of
the constant $0$ each hold one token tagged $0$; the place of every other
constant $c$ holds one token tagged $c$; the label places and $\mathit{done}$
are empty. The target marking $\marking_{\tt fin}$ has a single token tagged
$0$ in $\mathit{done}$ and no other tokens.

\paragraph{Transitions.}
The transitions are listed in Table~\ref{tab:nu-net}. In the
\emph{Consumes} column, $p\colon u$ means $\infun(p,t) = \{u\}$; in the
\emph{Produces} column, it means $\outfun(p,t) = \{u\}$. For places that are
not listed, $\infun(p,t)$ and $\outfun(p,t)$ are empty.
\begin{itemize}
	\item $t_{\tt in}$ creates a new thread at $\initlab$ whose local register
	holds $0$; it reads, and returns, the token of the constant $0$.
	\item An assignment moves the thread's token from $\ell$ to
	$\nextof{\ell}$ and copies a tag: into the register's place for
	$g := \xvar$ and $g := h$, and into the thread's own token for
	$\xvar := h$. The instruction $\xvar := *$ gives the thread's token a tag
	that is not present in the marking.
	\item A test has one transition per branch. Since modes are injective,
	using the \emph{same} variable on both compared places forces equal tags
	(the \emph{then} branch), and using \emph{different} variables forces
	different tags (the \emph{else} branch). Tests with $\neq$ swap the two
	branches; tests comparing a place with itself, and the assignments
	$\xvar := \xvar$ and $g := g$, are handled as unconditional jumps.
	\item $t_{\tt fin}$ moves a thread at $\ell_{\tt fin}$ to $\mathit{done}$
	and retags it with $0$, so that $\marking_{\tt fin}$ can use a fixed tag.
\end{itemize}
Because modes are injective, distinct variables of a transition always denote
distinct tags. To allow values to coincide, we also add, for every transition
in Table~\ref{tab:nu-net}, each transition obtained by renaming some of its
$\var$-variables to a common variable, except that the two variables compared
in the \emph{else} branch of a test are never merged.

\begin{table}[t]
	\centering
	\caption{Transitions of the $\nu$-net. Here $u$ is the tag of the executing
		thread, $v, w \in \var$, $y \in \newvar$, $g \in \globalregs$ and
		$h \in \globalregs \cup \constants$ (in tests, $g \in \globalregs \cup
		\constants$ as well), and $g \neq h$. Variants that merge variables are added
		as described in the text.}
	\label{tab:nu-net}
	\renewcommand{\arraystretch}{1.2}
	\setlength{\tabcolsep}{8pt}
	\begin{tabular}{@{}lll@{}}
		\hline
		Instruction & Consumes & Produces \\
		\hline
		$t_{\tt in}$ (new thread) & $0\colon u$ & $0\colon u,\ \initlab\colon u$ \\
		$\ell$: $g := \xvar$ & $\ell\colon u,\ g\colon v$
		& $\nextof{\ell}\colon u,\ g\colon u$ \\
		$\ell$: $g := h$ & $\ell\colon u,\ g\colon v,\ h\colon w$
		& $\nextof{\ell}\colon u,\ g\colon w,\ h\colon w$ \\
		$\ell$: $\xvar := h$ & $\ell\colon u,\ h\colon v$
		& $\nextof{\ell}\colon v,\ h\colon v$ \\
		$\ell$: $\xvar := *$ & $\ell\colon u$ & $\nextof{\ell}\colon y$ \\
		$\ell$: goto $\ell'$ & $\ell\colon u$ & $\ell'\colon u$ \\
		$\ell$: if $\star$ (then) & $\ell\colon u$ & $\ell'\colon u$ \\
		$\ell$: if $\star$ (else) & $\ell\colon u$ & $\nextof{\ell}\colon u$ \\
		$\ell$: if $\xvar = h$ (then) & $\ell\colon u,\ h\colon u$
		& $\ell'\colon u,\ h\colon u$ \\
		$\ell$: if $\xvar = h$ (else) & $\ell\colon u,\ h\colon v$
		& $\nextof{\ell}\colon u,\ h\colon v$ \\
		$\ell$: if $g = h$ (then) & $\ell\colon u,\ g\colon v,\ h\colon v$
		& $\ell'\colon u,\ g\colon v,\ h\colon v$ \\
		$\ell$: if $g = h$ (else) & $\ell\colon u,\ g\colon v,\ h\colon w$
		& $\nextof{\ell}\colon u,\ g\colon v,\ h\colon w$ \\
		$t_{\tt fin}$ & $\ell_{\tt fin}\colon u,\ 0\colon v$
		& $\mathit{done}\colon v,\ 0\colon v$ \\
		\hline
	\end{tabular}
\end{table}

\begin{lemma}\label{lem:nu-net}
	The label $\ell_{\tt fin}$ is reachable by a fresh run of $\pds$ if and only
	if $\marking_{\tt fin}$ is coverable from $\marking_{\tt in}$.
\end{lemma}

\begin{proof}
	By inspection of Table~\ref{tab:nu-net}, every transition returns exactly one
	token to each register or constant place it reads, and no transition changes
	the tag of a constant place. Hence, in every marking reachable from
	$\marking_{\tt in}$, each register and constant place holds exactly one
	token, and the place of each constant $c$ holds the tag $c$. Such a marking
	$\marking$ encodes a configuration $\mathit{conf}(\marking)$ of $\pds$: one
	thread for each token in a label place, at that label and with the token's
	tag as its local value, and each global register holding the tag in its
	place.
	
	($\Rightarrow$) Let $\pi$ be a fresh $n$-run of $\pds$ that reaches
	$\ell_{\tt fin}$. Firing $t_{\tt in}$ $n$ times from $\marking_{\tt in}$ yields
	a marking that encodes the initial $n$-configuration. We then simulate the
	steps of $\pi$ in order. If a thread at $\ell$ executes $\tau$, we fire the
	transition for $\tau$ (for a test, the one for the branch taken) in the
	variant that merges exactly the variables whose values coincide, with the
	mode that maps each variable to the corresponding value; this mode is
	injective by the choice of variant. For $\xvar := *$, the nonce $\val'$
	chosen in $\pi$ is neither a constant nor an earlier nonce, whereas every tag
	in the current marking is a constant or an earlier nonce; hence
	$\val' \notin \domain(\marking)$ and we may map $y$ to $\val'$. Each step
	preserves the encoding. When a thread reaches $\ell_{\tt fin}$, we fire
	$t_{\tt fin}$, which puts a token tagged $0$ into $\mathit{done}$; the
	resulting marking covers $\marking_{\tt fin}$.
	
	($\Leftarrow$) Suppose $\marking_{\tt in} \movesto{}^* \marking'' \geq
	\marking_{\tt fin}$. Only $t_{\tt fin}$ produces tokens in $\mathit{done}$,
	so it fires; the prefix of the firing sequence before its first firing ends
	in a marking with a token in $\ell_{\tt fin}$. Since $t_{\tt in}$ only reads
	and returns the token of the constant $0$, whose tag never changes, all
	firings of $t_{\tt in}$ in this prefix can be moved to its beginning; let $n$
	be their number. Every other transition in the prefix corresponds to one step
	of a thread of $\pds$ on the encoded configurations: assignments copy tags as
	the instruction prescribes, and, since modes are injective, a test fires its
	\emph{then} transition exactly when the compared values are equal. This
	yields an $n$-run of $\pds$ that reaches $\ell_{\tt fin}$.
	
	This run need not be fresh. A variable $y \in \newvar$ is mapped to a tag
	absent from the current marking, so the new value is not a constant (constant
	tags are always present) and is not held by any register or thread; but it
	may coincide with a nonce generated earlier that has since disappeared. We
	repair this from left to right: whenever a nonce step produces a value $\val$
	that already occurred earlier in the run, we replace $\val$ by a value that
	occurs nowhere in the run, in all configurations from this step up to the
	next step that generates $\val$ again. Since $\val$ is absent just before the
	step, every occurrence of $\val$ in that segment originates from this nonce,
	so the replacement preserves every equality and inequality tested by the run.
	The result is a fresh $n$-run of $\pds$ that reaches $\ell_{\tt fin}$.
\end{proof}

By Lemma~\ref{lem:nu-net}, \PCRFR\ with a single local register reduces to
coverability in $\nu$-nets, which is decidable. This proves decidability of
\PCRFR\ with a single local register.

	\section{Single Local Register a Direct \WSTS\ Approach} \label{sec:onelocal}

In this section, we show that \PCRFR\ is decidable if each thread has only one local register.  We reduce this problem to reachability in a \emph{well-structured system} \cite{561359,FINKEL200163}.
We
fix $\pds = (\globalregs,P)$ with $\globalregs = \{ \greg_1, \dots,\greg_m\}$,
we let the unique local register in $P$ to be $\lreg$. For the 
sake of simplicity, we  assume that the program does not use 
constants and later explain how handle constants separately.

Let $R: \domain \to \domain$ be a bijection respecting the constants.  For a configuration $\pdsconf$, $R(\pdsconf)$ is the configuration obtained by replacing any value $\val$ by $R(\val)$ in every register (local and global). Quite clearly if $\pdsconf_1 \movesto{} \pdsconf_2 \movesto{} \ldots \pdsconf_m$ is a witnessing run then so is
$R(\pdsconf_1) \movesto{} R(\pdsconf_2) \movesto{} \ldots R(\pdsconf_m)$. 
Keeping this in mind, we define the following relation on configurations:
$\pdsconf_1 \equiv \pdsconf_2$ iff there is a bijection $R$ on $\domain$ such
that $\pdsconf_2 = R(\pdsconf_1)$. Observe that this relation is an equivalence relation. This equivalence yields an abstract state space on which we can faithfully map (to be made precise)  concrete moves as abstract moves yielding a well-structured system.  

The idea behind the abstraction is the following, first,  we do not maintain the exact value of the global registers but just maintain a partition on $\globalregs$ (where two registers are in the same set iff they have the same value), secondly, for each thread we simply track its label and whether  its value is equal to one of the global registers (if so which)  or not.  The space of abstract configurations is contained in $(\Partition(\globalregs) \times (\functions{\Labels}{\Naturals}_\greg)_{\greg \in \globalregs} \times \functions{\Labels}{\Naturals}^*) $. 
In such an abstract configuration $\tuple{\pt, (\ffun_\greg)_{\greg \in \globalregs},\word}$, the first component holds the partition of the global registers, the second component is a tuple of functions --- one for each global register $\greg$ --- that counts, for each program label, the number of threads that are in this label and whose local variable holds the same value as $\greg$. If $\greg$ and $\greg'$ are equivalent according to $\pt$ then corresponding function should be the same. Finally, the last component is a sequence of functions -- each element of this sequence corresponds to a distinct value that appears in the local register of some thread but is different from all the global registers. Each such element is a function that counts, for each program label, the number of threads with this value that are currently at that label.

  The mapping abstraction $\abstraction$ that associates with each configuration $\pdsconf \in\Confof{\pds}$ an abstract configuration is defined as follows: 
  Let $\pdsconf= \tuple{\labelmap, \lregmap,\gregmap} \in \Confof{\pds}(n)$ for some $n \in \Naturals$ be any configuration  with $n$ threads. Let $\vset = \{ \gregmap(\greg)\mid \greg \in \globalregs \} = \{\val_1, \dots,\val_k\}$ be the set of all values appearing in the global register.  We let $\Parts(\pdsconf) = \tuple{\{\greg\in \globalregs \mid  \gregmap(\greg) = \val_1\}, \dots, \{\greg\in \globalregs \mid \gregmap(\greg) = \val_k\}}$  be the partition of the global registers based on the values in $\vset$. 
For any $\val \in \domain$ we define a map $\hfun^\pdsconf_\val : \Labels \to \Naturals$ as follows: Given any $\lab \in \Labels$, it returns the number of threads that are in  label $\lab$ and with the value $\val$ in its local register. Recall that all this is with respect to the  configuration $\pdsconf = \tuple{\labelmap, \lregmap,\gregmap} $.  
$$\hfun^\pdsconf_{\val}(\lab) = |\{i \in [1,n] \mid \lregmap(i)(\lreg) = \val, \labelmap(i) = \lab \}|$$
For each $\greg \in \globalregs$, we let $\ffun_\greg = \hfun^\pdsconf_{\gregmap(\greg)}$. %
Let $\bar{\vset} = \domain(\pdsconf) \setminus \vset$ be the set of all values in $\pdsconf$ that do not appear in any global register. Suppose $\bar\vset =   \{\val'_1, \dots , \val'_{k'} \}$, we let $\word(\pdsconf) = \hfun^\pdsconf_{\val'_1} \hfun^\pdsconf_{\val'_2} \dots \hfun^\pdsconf_{\val'_{k'}}$ to be a sequence of functions, one for each  value in $\bar{\vset} $. We will assume they are listed according to some fixed ordering on the domain, we define the abstraction of $\pdsconf$ as: 
$$			\abstraction(\pdsconf) = \tuple{\Parts(\pdsconf), (\ffun_\greg)_{\greg \in \globalregs}, \word(\pdsconf)} 
$$
We let $\ABS = \{\abstraction(\pdsconf) \mid \pdsconf \in \Confof\pds\}$ to be the set of all possible abstract configurations.
Given two  abstractions $\abs_1 = \tuple{\pt^1, (\ffun^1_\greg)_{\greg \in \globalregs},\word^1} ,\abs_2 = \tuple{\pt^2, (\ffun^2_\greg)_{\greg \in \globalregs},\word^2} \in \ABS$, we say $\abs_1 \equiv_\abs \abs_2$ if $\pt^1 = \pt^2$, for each $\greg \in \globalregs$, $\ffun^1_\greg = \ffun^2_\greg$ and there is a bijection between $\word^1$ and $\word^2$ (i.e. one is a permutation of the other). We  abuse  notation and say $\abs_1 = \abs_2$ if $\abs_1 \equiv_\abs \abs_2$.
Our first lemma explains the connection between $\equiv_\abs$ and $\equiv$.

It is easy to see that from $\abs$ we can determine unambiguously, the number of threads involved in any associated concrete configuration, the number of threads at a particular program label, the number of distinct data values in use and so on. Essentially all the relevant information other than the exact data values.

\paragraph{Freshness.}
In a step of $\movesto{}$, a nonce assignment $\xvar := *$ may pick any
value. A fresh run needs a value that was never seen before. Here we use a
weaker, local condition. A step is \emph{locally fresh} if every nonce it
produces is neither a constant nor a value present in the current
configuration. In this section, we abuse notation: we write $\movesto{}$ only
for locally fresh steps, and runs use only such steps. So, in the
abstraction, a nonce assignment always puts the thread's value into a new
class of its own.

Locally fresh is as good as fresh. Every fresh run is locally fresh, because
every value present in a configuration is a constant or an earlier nonce.
Conversely, a locally fresh run may reuse a value that was used earlier but
is no longer present anywhere. This value cannot be compared with its old
occurrences any more. So, from that step on, we can rename it to a value that
occurs nowhere in the run. Doing this for every reuse gives a fresh run that
reaches the same labels. When a larger configuration already holds the chosen
value, we pick another new value; the results agree up to $\equiv$.

For any two abstractions $\abs_1,\abs_2 \in \ABS$, we say $\abs_1 \movesto{}_{\ABS} \abs_2$ (or simply $\abs_1 \movesto{}_{} \abs_2$) if and only if there are $\pdsconf_1, \pdsconf_2$ such that $\abstraction(\pdsconf_1) = \abs_1$, $\abstraction(\pdsconf_2) = \abs_2$ and $\pdsconf_1 \movesto{} \pdsconf_2$. We will  let $\movesto{}^*_\ABS$ to also denote the reflexive transitive closure of $\movesto{}_{\ABS}$.  The following  lemma relate an abstract configuration with a concrete one. The proof of the lemma can be found in appendix. %

\begin{lemma}\label{lem:onelocal:absequiv}
	$\pdsconf_1 \equiv \pdsconf_2$ if and only if $\abstraction(\pdsconf_1) = \abstraction(\pdsconf_2)$
\end{lemma}
\begin{proof}
	For the forward direction, let   $\pdsconf_1 \equiv \pdsconf_2$ and   $\abs(\pdsconf_1) = \tuple{\pt^1, (\ffun^1_\greg)_{\greg \in \globalregs},\word^1}$, $\abs(\pdsconf_2) = \tuple{\pt^2, (\ffun^2_\greg)_{\greg \in \globalregs},\word^2}$. We will show $\abstraction(\pdsconf_1) = \abstraction(\pdsconf_2)$. 
	Since $\pdsconf_1 \equiv \pdsconf_2$, there is a bijection $R$ such that $\pdsconf_2 = R(\pdsconf_1)$, and hence $\pdsconf_1,\pdsconf_2 \in \Confof{\pds}(n)$ for some $n \in \Naturals$.  Notice that for any $\greg_1,\greg_2 \in \globalregs$, $\gregmap(\pdsconf_1)(\greg_1) = \gregmap(\pdsconf_1)(\greg_2) $ if and only if $\gregmap(\pdsconf_2)(\greg_1) = \gregmap(\pdsconf_2)(\greg_2) $, hence  $\pt^1 = \pt^2$. Similarly, we also have for any $\greg \in \globalregs$ and a thread $i \in [1,n]$, $\gregmap(\pdsconf_1)(\greg)  = \lregmap(\pdsconf_1)(\lreg)(i)$ if and only if $\gregmap(\pdsconf_2)(\greg)  = \lregmap(\pdsconf_2)(\lreg)(i)$. With this we also have $\ffun^1_\greg = \ffun^2_\greg$. Finally for any $i,j \in [1,n]$ we have $\lregmap(\pdsconf_1)(\lreg)(i) = \lregmap(\pdsconf_1)(\lreg)(j)$ if and only if $\lregmap(\pdsconf_2)(\lreg)(i) = \lregmap(\pdsconf_2)(\lreg)(j)$, hence there is a bijection between $\word_1 $ and $\word_2$. %
	With this we have 	$\abstraction(\pdsconf_1) =  \abstraction(\pdsconf_2)$

	For  the other direction, since we have  $\abstraction(\pdsconf_1) = \abstraction(\pdsconf_2)$, there is a bijection between $\word^1$ and $\word^2$. This induces a natural bijection between the data values that appear in $\pdsconf_1$ and the data values that appear in $\pdsconf_2$.  Any bijective function $R$ that extends this bijection, will provide witness for $\pdsconf_1 \equiv \pdsconf_2$.
\end{proof}

The below lemma relates the moves of two equivalent configurations,  the proof for it can be found in the appendix.  As a consequence, we will have that  $\movesto{}$ on configurations is faithfully lifted to abstract configurations. 

\begin{lemma}\label{lem:onelocal:pcequiv}
        For any two configurations $\pdsconf_1,\pdsconf_2 \in \Confof{\pds}$, if $\pdsconf_1 \movesto{} \pdsconf_2$, then the following holds
        \begin{itemize}
                \item For any $\pdsconf \equiv \pdsconf_1$, there is a $\pdsconf' \equiv \pdsconf_2$ such that $\pdsconf \movesto{} \pdsconf'$.
                \item  For any $\pdsconf' \equiv \pdsconf_2$, there is a $\pdsconf \equiv \pdsconf_1$ such that $\pdsconf \movesto{} \pdsconf'$.
        \end{itemize}
\end{lemma}

\begin{proof}[proof sketch of Lemma \ref{lem:onelocal:pcequiv}]
	Since $\pdsconf \equiv \pdsconf_1$,  there is a bijection $R$ such that $R(\pdsconf_1) = \pdsconf$. Since $\pdsconf_1 \movesto{} \pdsconf_2$ and since any move of the program is dependant only on the relationship between the registers, we also have $R(\pdsconf_1) \movesto{} R(\pdsconf_2)$, then the required configuration is $\pdsconf' = R(\pdsconf_2)$.
\end{proof}

\paragraph*{ Ordering on $\ABS$}
Next we define an order on the set of all  abstract configurations.
Given two  abstractions $\abs_1 = \tuple{\pt^1, (\ffun^1_\greg)_{\greg \in \globalregs},\word^1} ,\abs_2 = \tuple{\pt^2, (\ffun^2_\greg)_{\greg \in \globalregs},\word^2} \in \ABS$, we  say $\abs_1 \preceq \abs_2$ if  $\pt^1 = \pt^2$   and the following holds. 
\begin{itemize}
	\item For each $ \greg \in \globalregs, \ffun^1_\greg \leq \ffun^2_\greg$. For  $\ffun,\gfun \in \functions{\Labels}{\Naturals}$,  we say $\ffun \leq \gfun$ if for all $\lambda \in \Labels$, $\ffun(\lambda) \leq \gfun(\lambda)$.
	\item $\word^1 \leq^* \word^2$. Here, for any two sequence $\alpha=\ffun_1 \ffun_2 \dots \ffun_k \in  \functions{\Labels}{\Naturals}^*$ and $ \beta=\gfun_1 \gfun_2 \dots \gfun_{k'} \in \functions{\Labels}{\Naturals}^*$,  we say $\alpha \leq^* \beta$ if there is an injective function $\hfun: [1,k] \to [1,k']$ such that for all $i \in [1,k]$, $\ffun_i \leq \gfun_{{\hfun(i)}}$.
\end{itemize}

Let $ \tset_\ABS - \tuple{\tuple{\ABS, \preceq},\movesto{}_{\ABS}}$, we  next show that $\tset_\ABS$ is a \WSTS. We first start with a lemma that will be useful, we first introduce some notations. 

Let $\pdsconf_1 = \tuple{\labelmap_1, \lregmap_1,\gregmap_1} \in \Confof{\pds}(n)$ and $\pdsconf_2 =  \tuple{\labelmap_2, \lregmap_2,\gregmap_2} \in \Confof{\pds}(m)$ for some $n < m$. 
We say $\pdsconf_1 \leq \pdsconf_2$ ($\pdsconf_1 < \pdsconf_2$ when $\pdsconf_1 \neq \pdsconf_2 $) if the following holds.
\begin{itemize}
\item The shared registers hold the same value and the configuration  i.e.  for any $\greg \in \globalregs$, we have $\greg(\pdsconf_1) = \greg(\pdsconf_2)$.
\item $\pdsconf_1$ can be obtained from $\pdsconf_2$ by deleting from threads in it. That  is, there is an injection $\hfun: [1,n] \to [1,m]$ such that: 
$\gregmap_1 = \gregmap_2$
and for all $i \in [1,n]$, we have $\labelmap_1(i) = \labelmap_2(\hfun(i)) $ and $\lregmap_1(i) = \lregmap_2(\hfun(i)) $. 
\end{itemize}
When the the function $\hfun$ is important, we say $\pdsconf_1 \leq_\hfun \pdsconf_2$.  It is easy to see that if $\pdsconf_1 \leq \pdsconf_2$ then $\abstraction(\pdsconf_1) \preceq  \abstraction(\pdsconf_2) $. The following lemma will be useful to prove  $\tset_\ABS$ is a \WSTS.

\begin{lemma}\label{apx:onelocal:pcsim}
	For any two configurations $\pdsconf_1,\pdsconf_2 \in \Confof{\pds}(n)$ for some $n \in \Naturals$, if $\pdsconf_1 \movesto{} \pdsconf_2$, let the move be due to the thread $j \in [1,n]$ then, for any $\pdsconf \in \Confof\pds(m)$ the following holds.
	\begin{enumerate}
		\item  If  $\pdsconf_1 < \pdsconf$, then there is a $\pdsconf' \in \Confof\pds(m)$ such that  $\pdsconf_2 < \pdsconf'$ and  $\pdsconf \movesto{} \pdsconf'$.
		\item If  $\pdsconf <_\hfun \pdsconf_1$  and there is an thread $i \in [1,m]$ such that $h(i) = j$, then there is a $\pdsconf' \in \Confof\pds$ such that $\pdsconf' < \pdsconf_2$ and  $\pdsconf \movesto{} \pdsconf'$.
	\end{enumerate}
\end{lemma}
\begin{proof}
	We will prove the former, the proof of the second item of the lemma is similar. Since  $\pdsconf_1 < \pdsconf$, we have that global registers of both $\pdsconf_1$ and $ \pdsconf$ are the same, that is, for all $\greg \in \globalregs$, $\gregmap(\pdsconf_1)(\greg) = \gregmap(\pdsconf)(\greg) $. More over, by definition, we have a $\hfun$ such that  for all $j \in [1,n]$,  we  have $\lregmap(\pdsconf_1)(j) = \lregmap(\pdsconf)(\hfun(j)) $ and $\labelmap(\pdsconf_1)(j) = \labelmap(\pdsconf)(\hfun(j)) $. Suppose $k \in [1,n]$ was the thread that made the move from $\pdsconf_1$. Since this thread is also present in $\pdsconf$ as $\hfun(k)$ and since all the global registers have the same value, the move  can be replicated by $\hfun(k)$ in  $\pdsconf$ to obtain the desired $\pdsconf'$
\end{proof}

With this lemma in place, we can now prove that $\tset_\ABS$ is a \WSTS.

\begin{lemma}\label{lem:onelocal:wsts}
	The transition system $\tset_\ABS $  is a \WSTS.
\end{lemma}
\begin{proof}[Proof of Lemma \ref{lem:onelocal:wsts}]
	
	We first show that $\tuple{\ABS, \preceq}$ is a \wqo. Any element in the set $\ABS$ is of the form $\tuple{\pt, (\ffun_\greg)_{\greg \in \globalregs},\word} $. The first component is a finite set with the order being equality. The second component can be thought of as a tuple of natural numbers with point wise ordering between them. The third component can be thought of as the usual subword $\leq^*$ ordering over a tuple of natural numbers. Hence as an direct application of Dickson's lemma and Higman's lemma \cite{579572cb-3498-3043-b697-8509d85bae0c}\cite{higman1952ordering}, we obtain that $\tuple{\ABS, \preceq}$ is a \wqo.
	
	We show that $\tset_\ABS $ satisfies the monotonicity property. For this, if $\abs_1 \movesto{}_{\ABS} \abs_2$ and if $\abs_1 \preceq \abs_3$, we need to show that there is a $\abs_4$ such that $\abs_2 \preceq \abs_4$ and $\abs_3 \movesto{} _{\ABS} \abs_4$. We know from definition of $\movesto{}_\ABS $, that there are $\pdsconf_1, \pdsconf_2 \in \Confof\pds$ such that	$\abstraction(\pdsconf_1) =\abs_1$, $\abs_2 =  \abstraction(\pdsconf_2 )$ and $\pdsconf_1 \movesto{} \pdsconf_2$. Let us assume that this move was due to some thread $j$.
	Let  $\pdsconf_3$ be such that $\abstraction(\pdsconf_3) = \abs_3$. Without loss of generality we may assume  $\pdsconf_1 < \pdsconf_3$. Then by Lemma \ref{apx:onelocal:pcsim}, we can find a $\pdsconf_4$ such that $\pdsconf_3 \movesto{} \pdsconf_4$ with $\pdsconf_2 < \pdsconf_4$. Now from the observation that if $\pdsconf < \pdsconf'$ then $\abstraction(\pdsconf) \preceq \abstraction(\pdsconf')$, we have  $\abs_2 \preceq \abs_4$, where  $\abs_4 = \abstraction(\pdsconf_4)$. From the definition of $\movesto{}_\ABS$, we also have $\abs_3 \movesto{} \abs_4$ as required.
\end{proof}

We next show that the \PCRFR problem reduces to the coverability problem in $\tset_\ABS $.  The coverability question asks, given  $\abs \in \ABS$ and a set of initial configurations $\initABS$, whether there is a run $\abs_{\tt init} \movesto{}^* \abs'$ for some $\abs_{\tt init} \in \initABS$ and $\abs' \in \upclosure{\abs}$. 
We first argue that the \PCRFR\ problem PDS reduces to the coverability problem for $\tset_\ABS$.  %

Let   $\finalABS  = \{ \abstraction(\pdsconf) \mid\pdsconf \in \Confof{\pds}(1),  \labelmap(\pdsconf)(1) = \ell  \} $ be the abstractions of all the final configurations with a single thread.  We show that there is a witnessing run for   \PCRFR\ problem if and only if  some configuration in $\finalABS$ is coverable from  $\initABS = \{\abstraction(\pdsconf) \mid \pdsconf \in \InitConfof{\pds}\}$. 
Even though we have described the set $\finalABS$ as abstractions of potentially infinite set, it is possible to provide an explicit definition as a finite set.

\begin{lemma}
	There is a $\pdsconf_1 \in \InitConfof\pds$ and a run $\pdsconf_1 \movesto{}^*  \pdsconf_m$such that $\labelmap( \pdsconf_m)(1) = \ell$ if and only there is a run from an $\abs_1 \in \initABS$ and a run $\abs_1 \movesto{}^* \abs_m$ with $\abs_m \in \finalABS\uparrow$.
\end{lemma}
\begin{proof}
	From 	$\pdsconf_1 \movesto{}^*  \pdsconf_m$, we immediately obtain a run $\abstraction(\pdsconf_1) \movesto{}^* \abstraction(\pdsconf_m)$, directly from the definition of $\movesto{}_\ABS$.

	For the other direction, suppose $\abs_1 \movesto{}^* \abs_m$, we can induct on the length of $m$ to construct the necessary $\pdsconf_1 \movesto{}^*  \pdsconf_m$. If $m=2$, then we obtain $\pdsconf_1 \movesto{}  \pdsconf_2$, immediately from the definition. If $m > 2$, we construct the run backwards starting from $m$.    Let $\pdsconf_2 \movesto{}  \pdsconf_{m}$ be the run corresponding to $\abs_2 \movesto{}^* \abs_{m}$ obtained inductively. By definition, if $\abs_1 \movesto{} \abs_2 $ then $\pdsconf'_1 \movesto{} \pdsconf'_2$ for some $\abs_1 = \abstraction(\pdsconf'_1)$, $\abs_2 = \abstraction(\pdsconf'_2) $. Moreover we have $\pdsconf'_2  \equiv \pdsconf_2$. Then by invoking Lemma \ref{lem:onelocal:pcequiv}, we obtain $\pdsconf_1$ such that  $\pdsconf_1 \movesto{}^*  \pdsconf_m$.
\end{proof}

 Hence to obtain  decidability for the \PCRFR, it is sufficient to solve the  coverability problem of $\tset_\ABS$. Towards this, we use the famous result of     Theorem \ref{thm:wsts:dec} which provides us with an algorithm for the coverability problem of \WSTS.

\begin{theorem}\label{thm:wsts:dec} \cite{10.5555/29601.29644}\cite{561359}
	Given a well structured transition system $\tuple{\tuple{\Gamma,\leq},\movesto{}}$, a set of initial configurations $\Gamma_{\tt init}$ and a $\gamma \in \Gamma$ it is possible to decide, if there is a run of the form $\gamma_{\tt init} \movesto{}^* \gamma'$ for some $\gamma' \in \upclosure{\gamma}$ if  1)\  $\leq$ is effectively decidable, 
		2) \  Given an $\alpha \in \Gamma$, a finite set $\aset$ with $\pre^*(\upclosure{\alpha}) \supseteq \upclosure{\aset} \supseteq \pre(\upclosure{\alpha})$ is computable and 
		3) \  Given an $\alpha \in \Gamma$, it is possible to check if $\alpha \in \Gamma_{\tt init}$.
	
\end{theorem}

Thus, it suffices to solve the coverability problem for $\tset_\ABS$.
It is easy to check that $\preceq$  and
membership in $\initABS$ are decidable.  We only need to show item (2)
of Theorem \ref{thm:wsts:dec}.

\paragraph*{Property  $2$ of Theorem \ref{thm:wsts:dec}}

The key observation that allows us to compute a finite $\aset$ such that $\pre(\upclosure\abs) \subseteq \upclosure \aset  \subseteq \pre^*(\upclosure\abs )  $ is that we can limit the number of extra threads to be considered to just one more than that in $\abs$. 
That is, let  $\sset = \{ \abs'' \in \upclosure{\abs}  \mid   |\abs''| - |\abs| \leq 1\}$, where $|\abs|$ represents the number of threads in it. We show  $\pre(\upclosure{\abs}) = \upclosure{\aset } $,     where $\aset = \pre(\sset) $.

Clearly, ${\aset} \subseteq  \pre(\upclosure{\abs})$, we show that $\pre(\upclosure{\abs}) \subseteq \upclosure{\aset}  $. 
Let $\abs_1 \in \pre(\upclosure{\abs})$, since $\abs_1 \in \pre(\upclosure{\abs})$, there is a $\abs_2 \in \upclosure{\abs}$ such that $\abs_1 \movesto{} \abs_2$. 
Suppose $|\abs_2| - |\abs| \leq 1$,   then $\abs_2 \in \sset$ and  clearly $\abs_1 \in \aset$. We will assume otherwise. 
Since $\abs_1 \movesto{} \abs_2$, there are $\pdsconf_1 \movesto{} \pdsconf_2$ for some $\pdsconf_1,\pdsconf_2 \in \Confof{\pds}(n)$ such that  $\abstraction(\pdsconf_1) = \abs_1$ and $\abstraction(\pdsconf_2) = \abs_2$. Let the thread that executed the move be $j \in [1,n]$. Without loss of generality, since $  \abs \preceq \abs_2$,  we may assume there is $\pdsconf$ such that $\abstraction(\pdsconf) = \abs$ and $\pdsconf < \pdsconf_2$.  Suppose the thread that made the move is also present in $\pdsconf$, i.e.  $\pdsconf <_j \pdsconf_2$, then by using Lemma \ref{apx:onelocal:pcsim} one can  find a $\pdsconf' \leq \pdsconf_1$ such that $\pdsconf' \movesto{} \pdsconf$ and $\abstraction(\pdsconf') \in \aset$. 
Suppose not, then  add the thread $j$ to $\pdsconf$ to obtain a $\pdsconf''$, let $\abs'' = \abstraction(\pdsconf'')$.  We  have  $|\abs''| - |\abs| =1$ and hence $\abs'' \in \sset$. More over we also have $\pdsconf < \pdsconf'' <_j \pdsconf_2$. Using the argument similar to the previous case,  we obtain a $\pdsconf' < \pdsconf_1$ such that $\pdsconf' \movesto{} \pdsconf''$. Let  $\abstraction(\pdsconf') = \abs'$, then we also have $\abs' \movesto{} \abs''$ and hence $\abs' \in \aset$ as required.

\begin{lemma}
 Given an	$\abs \in \ABS$, it is possible to decide,  whether there is a run of the form $\abs_{\tt init} \movesto{}^* \abs'$ for some $\abs' \in \upclosure{\abs}$ and $\abs_{\tt init} \in \initABS$.
\end{lemma}
\begin{theorem}
	 The \PCRFR\ problem for  \PDS\  with a single local register  is decidable.	
\end{theorem}

We also remark that to deal with the constants, we may introduce one global register for each constant in the program and initialize them with the value of the constant and the program may then use these registers instead of the constants.

	\section{Hardness of \PDS with a Single Local Register}
\label{apx:sec:hardness}

We reduce reachability in lossy counter systems to \PCRFR\ with a single
local register. Fix a lossy counter system $\lossy(n) = \lossytuples$ with
$n$ counters. We build a program $\pds$ from it.

\paragraph{Overview.}
First, one thread is elected as the leader, as in the undecidability proof.
The election uses two global registers of its own. They are not used
afterwards. Every other thread picks a counter $i \in [1,n]$ and runs the
helper code for counter $i$ (Algorithm~\ref{main:alg:lossy:help}).

The leader simulates the states of $\lossy(n)$. The helpers store the counter
values. For each counter $i$, there are two global registers,
$\greg^i_1$ and $\greg^i_2$. The register $\greg^i_1$ holds the current
\emph{phase} of counter $i$. The register $\greg^i_2$ carries messages between
the leader and the helpers. So $\pds$ has $2n+2$ global registers and one
local register $\xvar$. The values $\inc$, $\dec$ and $\bot$ are constants.

A helper of counter $i$ is \emph{active} if it is at label
$\ell^i_{\mathit{wait}}$ or at label $\ell^i_{\mathit{count}}$. The value of
counter $i$ is the number of active helpers of counter $i$ whose local
register holds the current phase. A zero test writes a fresh value into
$\greg^i_1$. No helper holds this value. So the counter becomes $0$.

The leader has one label $\ell_q$ for each state $q \in \States$. It starts
at $\ell_{\initstate^\lossy}$. Let $(q,\op_1,q_1), \ldots, (q,\op_m,q_m)$ be
x`the transitions of $\lossy(n)$ from $q$. The code at $\ell_q$ is
\[
\ell_q:\ \mathit{if}(\star)\
\begin{cases} [\op_1];\\ \goto\ \ell_{q_1}; \end{cases}
\ \cdots\
\begin{cases} [\op_m];\\ \goto\ \ell_{q_m}; \end{cases}
\]
Here $\mathit{if}(\star)$ picks one block nondeterministically. The snippet
$[\op]$ simulates the operation $\op$. The snippets and the helper code are
shown in Figure~\ref{fig:lossy}.

\begin{itemize}
	\item $[\inc_i]$: The leader writes $\inc$ into $\greg^i_2$. A new helper
	answers with its id. The leader reads the id into $\xvar$. It checks that
	$\xvar$ is neither $\inc$ nor the current phase. It writes $\bot$, and
	then the id. Then it waits until $\greg^i_2 = \bot$.
	\item Helper of counter $i$: The helper picks a fresh id. It waits for
	$\inc$ and writes its id. It waits until its id is overwritten. Then it
	waits until its id comes back. Now it knows that the leader chose it. It
	copies the current phase into $\xvar$ and writes $\bot$. Now it is active.
	\item An active helper may halt at any time. This models a loss.
	\item $[\dec_i]$: The leader writes $\dec$ into $\greg^i_2$. It waits until
	$\greg^i_2$ holds the current phase. An active helper writes this value
	and halts. Then the leader writes $\bot$.
	\item $[\zero_i]$: The leader writes a fresh value into $\greg^i_1$. This
	starts a new phase.
\end{itemize}

\begin{figure}[t]
	\centering
	\begin{minipage}[t]{0.45\textwidth}
		\vspace{0pt}
		\begin{algorithm}[H]
			$\greg^i_2 := \inc$\;
			$\xvar := \greg^i_2$\;
			$\assert(\xvar \notin \{\inc, \greg^i_1\})$\;
			$\greg^i_2 := \bot$\;
			$\greg^i_2 := \xvar$\;
			$\assert(\greg^i_2 = \bot)$\;
			\caption{$[\inc_i]$}
			\label{main:alg:hlp:pa}
		\end{algorithm}
	\end{minipage}\hfill
	\begin{minipage}[t]{0.40\textwidth}
		\vspace{0pt}
		\begin{algorithm}[H]
			$\greg^i_2 := \dec$\;
			$\assert(\greg^i_2 = \greg^i_1)$\;
			$\greg^i_2 := \bot$\;
			\caption{$[\dec_i]$}
			\label{main:alg:hlp:da}
		\end{algorithm}
		\vspace{-8pt}
		\begin{algorithm}[H]
			$\xvar := *$\;
			$\greg^i_1 := \xvar$\;
			\caption{$[\zero_i]$}
			\label{main:alg:hlp:za}
		\end{algorithm}
	\end{minipage}
	
	\begin{algorithm}[H]
		$\hlab_i$: $\xvar := *$;\ $\assert(\greg^i_2 = \inc)$;\ $\greg^i_2 := \xvar$\;
		$\assert(\greg^i_2 \neq \xvar)$;\ $\assert(\greg^i_2 = \xvar)$\;
		$\xvar := \greg^i_1$;\ $\greg^i_2 := \bot$\;
		$\ell^i_{\mathit{wait}}$: \textbf{if} $\star$ \textbf{then} $\goto\ \ell_{\tt halt}$\;
		$\ell^i_{\mathit{count}}$: $\assert(\greg^i_2 = \dec)$\;
		$\greg^i_2 := \xvar$;\ $\goto\ \ell_{\tt halt}$\;
		\caption{[Helper-$i$]}
		\label{main:alg:lossy:help}
	\end{algorithm}
	\caption{The snippets of the leader and the code of a helper of counter $i$.}
	\label{fig:lossy}
\end{figure}

The correctness of the construction follows from the lemma below.
Reachability in lossy counter systems is not primitive recursive. Together,
they give Theorem~\ref{thm:dec:sl:hardness}.

\begin{lemma}\label{main:lem:lossy2pds}
	A state $q \in \States_\lossy$ is reachable in $\lossy(n)$ if and only if the
	label $\ell_q$ is reachable by a fresh run of $\pds$.
\end{lemma}

\begin{proof}
	For a configuration $\pdsconf$ of $\pds$, let $\#_k(\pdsconf)$ be the set of
	active helpers of counter $k$ whose local register holds the value of
	$\greg^k_1$. When the leader is at a label $\ell_q$, let
	\[
	\Conf_{\mathsf{lcs}}(\pdsconf) = \bigl(q, |\#_1(\pdsconf)|, \ldots,
	|\#_n(\pdsconf)|\bigr).
	\]
	
	\smallskip
	\noindent\textbf{($\Rightarrow$)}
	Take a run of $\lossy(n)$ from $(\initstate^\lossy, 0, \ldots, 0)$ to
	$(q, a_1, \ldots, a_n)$. The run may contain losses. We build a fresh run of
	$\pds$. We use one leader and one helper for each increment in the run.
	
	After the election, the leader is at $\ell_{\initstate^\lossy}$. No helper is
	active. So $\Conf_{\mathsf{lcs}}$ is $(\initstate^\lossy, 0, \ldots, 0)$.
	We now simulate the steps of the run one by one. We keep this invariant:
	$\Conf_{\mathsf{lcs}}$ is equal to the current configuration of $\lossy(n)$.
	\begin{itemize}
		\item \emph{Increment of counter $k$.} A new helper registers for counter
		$k$. It runs in parallel with the leader's $[\inc_k]$. It becomes active
		and stores the current phase. So $|\#_k|$ grows by one.
		\item \emph{Decrement of counter $k$.} Then $a_k \geq 1$. One helper in
		$\#_k$ moves to $\ell^k_{\mathit{count}}$. The leader runs $[\dec_k]$. The
		helper answers and halts. So $|\#_k|$ drops by one.
		\item \emph{Zero test of counter $k$.} Then $a_k = 0$. So $\#_k$ is empty.
		The leader runs $[\zero_k]$. The new phase is fresh. So $\#_k$ stays
		empty.
		\item \emph{Loss on counter $k$.} One helper in $\#_k$ halts at
		$\ell^k_{\mathit{wait}}$. So $|\#_k|$ drops by one.
	\end{itemize}
	In each case, the other counters do not change. The leader moves to the
	right label $\ell_{q'}$. The nonces are the helper ids and the new phases. We
	choose each of them fresh. So the run is fresh. At the end, the leader is at
	$\ell_q$.
	
	\smallskip
	\noindent\textbf{($\Leftarrow$)}
	Take a fresh run $\rho$ of $\pds$ in which the leader reaches $\ell_q$. Only
	the leader visits labels of the form $\ell_{q'}$ and the labels of the
	snippets. We cut $\rho$ into segments. The first segment ends when the leader
	first reaches $\ell_{\initstate^\lossy}$. Each later segment starts at a
	label $\ell_{q'}$ of the leader. It ends at the next such label. In each later
	segment, the leader runs one snippet $[\op]$ for one transition
	$(q', \op, q'')$.
	
	We use three facts. They rely on freshness: ids are distinct, and phases are
	different from ids and from constants.
	\begin{enumerate}
		\item \emph{Joining.} A helper joins $\#_k$ only when it finishes its
		registration. This happens only during $[\inc_k]$, exactly once.
		Indeed, a helper finishes only after its id comes back. Only the leader
		writes an id back, once per $[\inc_k]$. Ids are distinct, so only one
		helper sees its id. The leader leaves $[\inc_k]$ only after it sees
		$\bot$ again. Only this helper can write that $\bot$. So it finishes
		inside this $[\inc_k]$. Just before, it copies $\greg^k_1$. The leader
		does not change $\greg^k_1$ during $[\inc_k]$. So the helper joins
		$\#_k$.
		\item \emph{Leaving.} A helper leaves $\#_k$ when it halts, when it answers
		a decrement, or when a zero test changes $\greg^k_1$. Leaving only lowers
		a counter.
		\item \emph{Decrement.} In $[\dec_k]$, the leader passes its check only if
		some thread writes the current phase into $\greg^k_2$. Only a helper in
		$\#_k$ can do this. This helper then leaves $\#_k$.
	\end{enumerate}
	
	We show by induction: at the end of each segment, $\lossy(n)$ can reach
	$\Conf_{\mathsf{lcs}}$.
	
	\emph{Base case.} At the end of the election, no helper is active. So
	$\Conf_{\mathsf{lcs}} = (\initstate^\lossy, 0, \ldots, 0)$. This is the
	initial configuration of $\lossy(n)$.
	
	\emph{Inductive step.} Let a segment run $[\op]$ for $(q', \op, q'')$. Let it
	start with counter values $a_1, \ldots, a_n$ and end with
	$a'_1, \ldots, a'_n$. The leader reaches $\ell_{q''}$, so the transition
	$(q', \op, q'')$ exists in $\lossy(n)$. By fact~1, no helper joins $\#_j$
	for $j \neq k$. By fact~2, these counters can only drop. Such drops are
	losses.
	\begin{itemize}
		\item $\op = \inc_k$. By fact~1, exactly one helper joins $\#_k$. Others
		may leave. So $a'_k \leq a_k + 1$. The system increments counter $k$ and
		then loses the extra units.
		\item $\op = \dec_k$. By fact~3, at least one helper leaves $\#_k$. No
		helper joins. So $a'_k \leq a_k - 1$. This helper was in $\#_k$ at the
		start, so $a_k \geq 1$. The system decrements counter $k$ and then loses
		the extra units.
		\item $\op = \zero_k$. At the end, $\greg^k_1$ holds a fresh value. No
		helper holds it. So $a'_k = 0$. The system loses all units of counter $k$
		and then does the zero test.
	\end{itemize}
	So $\lossy(n)$ can reach $(q'', a'_1, \ldots, a'_n)$. For the last segment,
	this shows that $q$ is reachable in $\lossy(n)$.
\end{proof}

\begin{corollary}\label{cor:single-local-complexity}
	\PCRFR\ with a single local register is in $\mathbf{F}_{\omega\cdot 2}$ and
	is $\mathbf{F}_{\omega}$-hard. In particular, it is not primitive recursive.
\end{corollary}

\paragraph{Complexity.}
Here $\mathbf{F}_{\omega}$ is the class of problems solvable in Ackermann
time. $\mathbf{F}_{\omega\cdot 2}$ is the class of problems solvable in
``double-Ackermann'' time. Both classes contain problems that are not
primitive recursive.

For the upper bound, we use the reduction of Lemma~\ref{lem:nu-net}. The net
has one place for each global register, constant and label, plus the place
$\mathit{done}$. Each instruction gives a constant number of transitions,
including their variants. So the net is computable in polynomial time.
Coverability in $\nu$-nets is complete for $\mathbf{F}_{\omega\cdot 2}$
\cite{LazicSchmitz16}. So \PCRFR\ with a single local register is in
$\mathbf{F}_{\omega\cdot 2}$.

For the lower bound, we use the reduction of this section. It uses $2n+2$
global registers and one local register. It is computable in polynomial time.
Reachability in lossy counter systems is $\mathbf{F}_{\omega}$-complete
\cite{Schnoebelen10}. So \PCRFR\ with a single local register is
$\mathbf{F}_{\omega}$-hard.

The two bounds do not match. The exact complexity lies between Ackermann and
double-Ackermann time. We leave it open.

	\section{Parametrized	 Reachability (without Freshness) is {\sc PSpace complete}} \label{sec:PRNoFresh1}

In this section, we prove that the PReach problem is {\sc Pspace-complete}, informally our strategy is as follows: We first compute a set of (global register,value) pairs for which there is a run that writes the value to the global register.  Now  under this assumption that such values can be read from the global register at any point, we check if there is a witnessing $1$-pseudo run. That is, we check if there is a $1$ pseudo run with support  only from the computed pairs.  Without loss of generality, we assume that $0$ cannot be written to any register in any instruction. Note that any writes to $0$ can be replaced by a new constant that is treated identically.

The support of any $1$-pseudo runs is potentially over an infinite set of alphabets. In order to compute the (register, value) pairs effectively,  we introduce a modified version of the support over a  finite set of alphabets. We define a morphism $\msupp: \globalregs \times \domain \to \globalregs \times (\constants \cup \{\ast\})$  as $\msupp(\tuple{\greg,\const}) = \tuple{\greg,\const}$ for any $\greg \in \globalregs, \const \in \constants$ and $\msupp(\tuple{\greg,\val}) = \tuple{\greg,\ast}$ otherwise. With an abuse of notation, given a $1$-pseudo run $\rho$, we let $\msupp(\rho) =  \msupp(\supp(\rho))$. The idea here is, if there is a run that writes a value from $\domain \setminus \constants$ to a global register, then we may assume any value from $\domain$ can be written to that register any number of times (recall that we are operating under the non-freshness assumption). Hence the value itself is not important.  

Computing such (value register pairs) has additional challenges. In order for a thread to write a value to the global register, it may depend on yet another value being available in the global register.  We show how to iteratively discover these interdependencies and the set of values that can be written to the global registers and the  witnessing $1$-pseudo run..

Given a Program $\pds$, a regular language $\reglang$ over  $\Sigma = \globalregs \times (\constants \cup  \{\ast \})$ and a pair $\tuple{\greg,\const}$, where   $\const  \in \constants$ and  $\greg \in \globalregs$, we say $\tuple{\greg,\const}$ is generatable by $\reglang$  if there is an $1$-pseudo run $\rho = \initpdsconf \Movesto{} \pdsconf$ such that $\msupp(\rho) \in \lang(\reglang)$ and $\greg(\pdsconf) =\const$. That is, there is an $1$-pseudo run that only uses values from the regular language $ \reglang$ in its support. Similarly given a global register $\greg \in \globalregs$, we say $\tuple{\greg,\ast}$ is generatable by $\reglang$ if there is an $1$-pseudo run $\rho  = \initpdsconf \Movesto{} \pdsconf$,  such that $\msupp(\rho) \in \lang(\reglang)$ and $\gregmap(\pdsconf)(\greg) \in \domain \setminus \constants$. Given a label $\lab \in \Labels$  if there is an pseudo $1$-run $\rho$ that reaches a configuration $\pdsconf$ with $\labelmap(\lab) = \lab$ and  $\msupp(\rho) \in \reglang$ then we say $\lab$ is \emph{pseudo-reachable} by $\reglang$. 

 The crucial idea  to compute the generatable set from $\globalregs \times ( \constants \cup  \{\ast \} )$ is to guess the order in which the dependencies can occur and then  to compute the values in that order.  
 The very first generatable pair in the order has no dependencies. The subsequent pairs are only dependent on the previously computed pairs. 
 However one needs to be careful with the dependencies that can arise when a value is generated. Consider the tuples $\tuple{\greg_1,\const_1}$, $\tuple{\greg_2,\const_2}$, and $\tuple{\greg_3,\const_3}$.  Suppose $\tuple{\greg_2,\const_2}$ is generatable under the assumption that $\tuple{\greg_1,\const_1}$ is available and  $\tuple{\greg_3,\const_3}$ is generatable under the assumption that $\tuple{\greg_2,\const_2}$ is available.  Then there is an indirect dependency on  $\tuple{\greg_1,\const_1}$, when generating $\tuple{\greg_3,\const_3}$. Suppose there a $1$-pseudo-run that generates $\tuple{\greg_3,\const_3}$ and along this run it read the value $\const_2$ from $\greg_2$.  Since $\tuple{\greg_1,\const_1}$ is required  to generate  $\tuple{\greg_2,\const_2}$,  it can no longer read  the initial value from $\greg_1$. %
To address this, we guess the word  $\word \in (\globalregs \times (\constants\setminus \{ 0 \} \cup \{\ast \})^*$  such that there is no repetition of letters in it, let  $\word = a_1 \dots a_m$.  This word corresponds to the set of all values that can be generated, along with an ordering between them. 
We say the word is  a \emph{generatable sequence} if for each $i \in [1,m]$, $a_i$ is generatable by the regular language $a_1^* a_1 \{a_1,a_2\}^* a_2 \{a_1,a_2,a_3\}^* \dots a_{i-2} \{a_1,\dots a_{i-1}\}^*a_{i-1}$.  Consider $i = 3$, then the regular language corresponding to it is  $a_1^* a_1 \{a_1,a_2\}^*a_2$. Notice that the language ensure that the dependency of $a_2$ is necessarily present before the consumption of $a_2$ is possible.
The following lemma states that  to check if there is a witnessing run, we need to check if $\lab$ can be reached by a  pseudo $1$-run that only reads values from   a generatable sequence.  Finding a witnessing $1$ pseudo run can be achieved in {\sc Pspace}  in a manner similar to the proof in Section \ref{apx:sec:oneglobal:acyclic}. The same can also be obtained by using the result in  \cite{DemriL09}.

\begin{lemma}\label{lem:nofresh:dec}
	Given a program $\pds$ and a label $\lab \in \Labels$, the label $\lab$
	is parametrized reachable if and only if there exists a generatable sequence
	$ \word = a_1 \cdots a_m $
	such that $\lab$ is pseudo-reachable by
	$ a_1^* a_1 \{a_1,a_2\}^* a_2 \{a_1,a_2,a_3\}^* \cdots a_m \{a_1,\ldots,a_m\}^*. $
\end{lemma}

\begin{proof}
	We prove the two implications separately.
	
	\smallskip
	\noindent\textbf{($\Rightarrow$)}
	Assume that $\lab$ is parametrized reachable in $\pds$. Hence, for some
	$n \in \Naturals$, there is a witnessing run
	$ \pi = \initpdsconf(n) \movesto{}^* \tuple{\labelmap,\lregmap,\gregmap} $
	such that $\labelmap(1)=\lab$.
	
	We extract from $\pi$ the values that are written to global registers for
	the first time. Let
	$ \tau_1,\ldots,\tau_s $
	be the subsequence of transitions of $\pi$ that write to a global register.
	For each $j \in [1,s]$, suppose that $\tau_j$ writes a value $d$ into the
	global register $\greg$. We associate with this write a symbol
	$$ b_j = \begin{cases} \tuple{\greg,c}, & \text{if } d=c \text{ for some } c \in \constants,\\ \tuple{\greg,\ast}, & \text{otherwise.} \end{cases} $$
	Thus, $b_j$ records the global register being written and whether the written
	value is a constant or a value generated by a non-deterministic assignment.
	
	Let
	$ \word = a_1 \cdots a_m $
	be obtained from $b_1 \cdots b_s$ by retaining only the first occurrence of
	each symbol. In other words, each $a_i$ represents a distinct kind of global
	write, listed according to the first point at which it occurs in $\pi$.
	
	We claim that $\word$ is generatable. For every $a_i$, let $t_i$ be a thread
	that performs the first write represented by $a_i$. Projecting $\pi$ onto
	$t_i$ yields a $1$-pseudo-run $\rho_i$ that generates $a_i$. Before this first
	write, the run $\rho_i$ can only read values corresponding to symbols that
	occurred earlier in $\word$. Therefore,
	$ \supp(\rho_i) \in a_1^* a_1 \{a_1,a_2\}^* a_2 \cdots \{a_1,\ldots,a_{i-1}\}^* a_{i-1}, $
	and hence $a_i$ is generatable from the preceding symbols.
	
	Finally, project $\pi$ onto the witnessing thread $1$. This gives a
	$1$-pseudo-run $\rho_\lab$ reaching $\lab$. Since every value read by this
	thread was written somewhere earlier in $\pi$, its support can use only the
	symbols $a_1,\ldots,a_m$. Moreover, the first occurrence of $a_i$ in the
	support can only appear after the first occurrences of
	$a_1,\ldots,a_{i-1}$. Hence,
	$ \supp(\rho_\lab) \in a_1^* a_1 \{a_1,a_2\}^* a_2 \cdots a_m\{a_1,\ldots,a_m\}^*. $
	Thus, $\lab$ is pseudo-reachable by the required expression.
	
	\smallskip
	\noindent\textbf{($\Leftarrow$)}
	Assume that there is a generatable word
	$ \word = a_1 \cdots a_n $
	and that $\lab$ is pseudo-reachable by
	$ a_1^* a_1 \{a_1,a_2\}^* a_2 \cdots a_n\{a_1,\ldots,a_n\}^*$.
	For every $i \in [1,n]$, since $a_i$ is generatable, there is a
	$1$-pseudo-run $\rho_i$ that generates $a_i$. In addition, let
	$\rho_\lab$ be a $1$-pseudo-run reaching $\lab$.
	
	The difficulty is that pseudo-transitions may read values from the global
	register without specifying which thread produced them. We construct a genuine
	run by using sufficiently many copies of the generating pseudo-runs. These
	copies ensure that whenever another thread needs the value represented by
	$a_i$, some producer thread can write that value to the global register.
	
	For each $i \in [1,n]$, define the number $k_i$ of copies of $\rho_i$
	backwards as follows:
	$ k_n = |\rho_\lab| $,
	and, for every $i \in [1,n-1]$,
	$ k_i = \sum_{j=i+1}^{n} \bigl(k_j \cdot |\rho_j| + |\rho_\lab|\bigr). $
	
	Intuitively, $k_i$ is chosen large enough to supply the value represented by
	$a_i$ to every later generator $\rho_j$, for $j>i$, and to the witnessing
	pseudo-run $\rho_\lab$. Each of the $k_j$ copies of a later generator
	$\rho_j$ may ask for it up to $|\rho_j|$ times, and the witnessing run
	$\rho_\lab$ may ask for it up to $|\rho_\lab|$ times.
	
	We start with
	$ K = 1 + \sum_{i=1}^{n} k_i $
	threads. One distinguished thread, denoted by $\thread_\ell$, will simulate
	$\rho_\lab$. For every $i \in [1,n]$, we use $k_i$ many threads, denoted
	collectively by $\Thread_{k_i}$, each of which simulates $\rho_i$ and is able
	to produce the value represented by $a_i$.
	
	We now decompose the pseudo-runs according to the first occurrence of each
	required pseudo-read. Since $\rho_i$ generates $a_i$, its support has the
	form
	$$ \supp(\rho_i) \in a_1^* a_1 \{a_1,a_2\}^* a_2 \cdots \{a_1,\ldots,a_{i-1}\}^* a_{i-1}. $$
	Accordingly, we decompose
	$ \rho_i = \rho_i^0 \cdot \rho_i^1 \cdots \rho_i^{i-1} \cdot \tau_i $,
	where $\tau_i$ is the transition that generates $a_i$,
	$\supp(\rho_i^0)=\epsilon$, and, for every $j \in [1,i-1]$,
	$ \supp(\rho_i^j) \in a_j \{a_1,\ldots,a_j\}^*$.
	
	Thus, $\rho_i^j$ is the part of $\rho_i$ that starts when the value
	represented by $a_j$ becomes available and ends immediately before the first
	use of the next new value.
	
	Similarly, decompose the witnessing pseudo-run as
	$ \rho_\lab = \rho_\lab^0 \cdot \rho_\lab^1$, 
	where 	$ \supp(\rho_\lab^0)=\epsilon $ and  $ \supp(\rho_\lab^1)\in \{a_1,\ldots,a_n\}^*$.
		We construct a genuine run in $n+1$ rounds.
	
	\medskip
	\noindent\emph{Round $0$.}
	For each $i \in [1,n]$, every thread in $\Thread_{k_i}$ executes
	$\rho_i^0$. The thread $\thread_\ell$ executes $\rho_\lab^0$. Since these
	prefixes contain no pseudo-transitions, they only access the initial global
	memory and can be executed directly. At the end of this round, every thread
	in $\Thread_{k_1}$ is ready to execute $\tau_1$ and thereby produce $a_1$.
	
	\medskip
	\noindent\emph{Round $j$, for $j \in [1,n]$.}
	At the beginning of this round, sufficiently many threads are ready to
	execute $\tau_j$ and produce the value represented by $a_j$. Whenever a later
	pseudo-run needs a pseudo-transition supplying $a_j$, we first schedule one
	such producer thread to execute $\tau_j$, and then execute the corresponding
	segment of the consumer thread. Thus, every pseudo-read of $a_j$ is replaced
	by an actual write of the corresponding value immediately before it is read.
	
	The definition of $k_j$ guarantees that enough producer threads are available.
	Indeed, the value $a_j$ is needed only by copies of $\rho_i$ with $i>j$ and
	by the witnessing run $\rho_\lab$. The term
	$ k_i \cdot |\rho_i| $
	bounds the total number of times that copies of $\rho_i$ may require an
	earlier value, while the term $|\rho_\lab|$ bounds the number of requests
	made by the witnessing thread. Hence, after completing round $j$, enough
	threads remain to supply every later demand for $a_j$.
	
	After round $n$, for every $i \in [1,n]$, there are still enough threads
	available to produce the value represented by $a_i$ whenever it is requested
	by $\rho_\lab$. We therefore execute $\rho_\lab^1$, replacing each of its
	pseudo-transitions by a real write performed by a suitable producer thread.
	
	In this way, every pseudo-transition of every $\rho_i$ and of $\rho_\lab$ is
	replaced by a genuine transition that writes the required value to the global
	register. The resulting execution is therefore a valid parametrized run of
	$\pds$, and the distinguished thread $\thread_\ell$ reaches $\lab$.
	Consequently, $\lab$ is parametrized reachable in $\pds$.
\end{proof}

 Combining this with the results from \cite{DemriL09} that states reachability in a register automata over finite/infinite domain is {\sc Pspace hard} gives us the following theorem.

\begin{theorem}\label{thm:nofresh:dec}
	The PReach problem is {\sc Pspace-complete}.
\end{theorem}

\begin{proof}[Proof of Theorem \ref{thm:nofresh:dec1}]
	We guess a  sequence $\word$ as required by Lemma \ref{lem:nofresh:dec}, this is of size at most linear in the size of the given system. We check if this sequence is generatable by iterating over each letter in the word.  By a proof similar to one in Section \ref{apx:sec:oneglobal:acyclic}, this can be done in {\sc Pspace}. We then check if there is a witnessing $1$-pseudo run. This again can also be achieved in {\sc Pspace}. The correctness follows from Lemma \ref{lem:nofresh:dec}.
\end{proof}

	\section{High level macro implementation}\label{sec:highmac}
	
	\noindent This section details the explicit algorithmic translations for these macros—namely the assert constraint, deterministic if/else branching, and non-deterministic If(*) block selection. We have given an implementation for all the three macros as a code snippet below.

\begin{nolinenumbers}
	
	\begin{figure}[htpb]
		
		\begin{minipage}[t]{0.48\textwidth}
			\begin{algorithm}[H]
				\LinesNotNumbered
				\caption{Implementation of $\mathtt{assert}(e_1, \dots, e_n)$}
				\begin{algorithmic}[1]
					\Procedure{:}{}
					\State  \textbf{if} $\neg e_1$
					\textbf{ then goto trap}  \textbf{ end}				
					\State  \textbf{if} $\neg e_2$
					\textbf{ then goto trap}  \textbf{ end}				
					\State \hspace{\algorithmicindent} $\vdots$					
					\State \textbf{if} $e_n$
					\textbf{ then goto} $\ell_{\tt nxt}$ \textbf{ end}
					\State trap: \textbf{goto trap} 
					\State$\ell_{\tt nxt}$:  next instruction
					\EndProcedure
				\end{algorithmic}
			\end{algorithm}
		\end{minipage}\hfill
		\begin{minipage}[t]{0.51\textwidth}
			\begin{algorithm}[H]
				\caption{Implementation of
					\texttt{if(*) \{ins$_1$\}; $\dots$; \{ins$_k$\}}}
				\begin{algorithmic}[1]
					\State \textbf{if} $*$
					\textbf{ then goto} $\ell_1$ \textbf{ end}
					\State \textbf{if} $*$
					\textbf{ then goto} $\ell_2$ \textbf{ end}
					\Statex \hspace{\algorithmicindent} $\vdots$
					 \State \textbf{if} $*$ \textbf{then goto } $\ell_k$				
					\State $\ell_1$: $\{\text{instr}_1\}$ \textbf{goto} $8$				
					\Statex \hspace{\algorithmicindent} $\vdots$
					
					\State $\ell_k$: $\{\text{instr}_k\}$ \textbf{goto} $L_{\mathrm{end}}$
					
					\State $L_{\mathrm{end}}$: next instruction
				\end{algorithmic}
			\end{algorithm}
		\end{minipage}
			\end{figure}
			\begin{figure}[htpb]
			\vspace{0pt}
			\begin{algorithm}[H]
				\LinesNotNumbered
				\caption{Implementation of
					\texttt{if $e$ then \{instr$_1$\} else \{instr$_2$\} end}}
				\begin{algorithmic}[]
					\State \textbf{if} $e$
					\textbf{ then goto} $\ell_t$ \textbf{ end}
					
					\Statex $\ell_f$: $\{\text{instr}_2\}$
					\State \textbf{goto} $L_{\mathrm{end}}$
					
					\Statex $\ell_t$: $\{\text{instr}_1\}$
					
					\Statex $L_{\mathrm{end}}$:
				\end{algorithmic}
			\end{algorithm}
		
	\end{figure}
	
\end{nolinenumbers}

\end{document}